\documentclass[11pt,reqno]{amsart}
\usepackage{amssymb}

\usepackage{amsthm,amsfonts,amssymb,euscript,hyperref}
\usepackage{comment}
\usepackage{mathtools}
\usepackage{stmaryrd}
\usepackage{mathrsfs}
\usepackage{pifont}
\usepackage{amsfonts,amsmath,amssymb,amsthm,amscd,cancel}
\usepackage{graphicx}
\usepackage[]{geometry}
\usepackage{verbatim}
\usepackage{mathrsfs}
\usepackage{fancyhdr}
\usepackage{footnpag}

\newtheorem{theorem}{Theorem}
\newtheorem{lemma}{Lemma}
\newtheorem{proposition}{Proposition}

\newtheorem{remark}{Remark}

\allowdisplaybreaks[4]

\DeclareMathAlphabet{\mathsfsl}{OT1}{cmss}{m}{sl}
\numberwithin{equation}{section}

\newcommand{\D}{\mathrm{d}}

\newcommand{\tr}{\mathrm{tr}}

\renewcommand{\iint}{\int\!\!\!\!\!\int}

\def\alphab{\underline{\alpha}}
\def\betab{\underline{\beta}}
\def\chib{\underline{\chi}}
\def\chibh{\widehat{\underline{\chi}}}
\def\chih{\widehat{\chi}}
\def\etab{\underline{\eta}}

\def\Lb{\underline{L}}
\def\mub{\underline{\mu}}
\def\kappab{\underline{\kappa}}
\def\tr{\mathrm{tr}}
\def\omegab{\underline{\omega}}

\def\tensor{\widehat{\otimes}}
\def\DD{\mathcal{D}}
\def\ub{\underline{u}}
\def\Cb{\underline{C}}
\def\Rb{\underline{R}}
\def\Lb{\underline{L}}
\def\lot{\text{l.o.t.}}

\newcommand{\Db}{\underline{D}}
\newcommand{\Dh}{\widehat{D}}
\newcommand{\Dbh}{\widehat{\underline{D}}}

\newtheorem*{MainTheorem}{Main Theorem}
\newtheorem*{ChristodoulouMainEstimates}{Christodoulou's Main Estimates}
\newtheorem*{LocalDeformationTheorem}{Local Deformation Theorem}

\def\nablas{\mbox{$\nabla \mkern -13mu /$ }}
\def\Deltas{\mbox{$\Delta \mkern -13mu /$ }}
\def\Gammas{\mbox{$\Gamma \mkern -13mu /$ }}
\def\divs{\mbox{$\mathrm{div} \mkern -13mu /$ }}
\def\curls{\mbox{$\mathrm{curl} \mkern -13mu /$ }}
\def\ds{\mbox{$\mathrm{d} \mkern -9mu /$}}
\def\gs{\mbox{$g \mkern -9mu /$}}
\def\epsilons{\mbox{$\epsilon \mkern -9mu /$}}
\def\pis{\mbox{$\pi \mkern -9mu /$}}
\def\piOs{{}^{(O)}{}\pis}
\def\piOsh{{}^{(O)}{}\widehat{\pis}}
\def\Lie{\mbox{$\mathcal{L} \mkern -10mu/$}}
\def\LieO{\Lie_{O}}
\def\LieOh{\widehat{\Lie}_{O}}

\begin{document}

\title[Cauchy Data and Gravitational Collapses]{Construction of Cauchy Data of Vacuum Einstein field equations Evolving to Black Holes}

\author[Junbin Li]{Junbin Li}
\address{Department of Mathematics, Sun Yat-sen University\\ Guangzhou, China}
\email{mc04ljb@mail2.sysu.edu.cn}

\author[Pin Yu]{Pin Yu}
\address{Mathematical Sciences Center, Tsinghua University\\ Beijing, China}
\email{pin@math.tsinghua.edu.cn}
\thanks{JL is deeply indebted to his Ph.D advisor Prof. \emph{Xi-Ping Zhu} for support on this work and to Prof. \emph{Siu-Hung Tang} for discussions. He would like to thank the Mathematical Sciences Center of Tsinghua University where the work was partially done during his visit.\\
\indent PY is supported by NSF-China Grant 11101235. He would like to thank Prof. \emph{Sergiu Klainerman} for numerous suggestions on improving the earlier manuscripts; Prof. \emph{Shing-Tung Yau} for his idea of proving the absence of trapped surfaces; and Professors \emph{Demetrios Christodoulou}, \emph{Igor Rodnianski} and \emph{Richard Schoen} for their interests and discussions.}
\date{\today}

\maketitle

\begin{abstract}
We show the existence of complete, asymptotically flat Cauchy initial data for the vacuum Einstein field equations, free of trapped surfaces,  whose future development must admit a trapped surface. Moreover, the datum is exactly a constant time slice in Minkowski space-time inside and exactly a constant time slice in Kerr space-time outside.

The proof makes use of the full strength of Christodoulou's work on the dynamical formation of black holes and Corvino-Schoen's work on the constructions of initial data set.
\end{abstract}
\setcounter{tocdepth}{1}

\parskip=\baselineskip

\section{Introduction}

\subsection{Earlier Works}
Black holes are the central objects of study in general relativity. The presence of a black hole is usually detected through the existence of a trapped surface, namely a two dimensional space-like sphere whose outgoing and incoming expansions are negative. The celebrated Penrose singularity theorem states that under suitable assumptions, if the space-time has a trapped surface,  then the space-time is future causally geodesically incomplete, i.e., it must be singular, at least in some weak sense. On the other hand, the weak cosmic censorship conjecture (WCC) asserts that singularities have to be hidden from an observer at infinity by the event horizon of a black hole. Thus, assuming WCC,  the theorem of Penrose  predicts the existence of black holes, via the exhibition of a trapped surface. This is precisely the reason why the trapped surfaces are intimately related to the understanding of the mechanism of gravitational collapse.

In \cite{Chr}, Christodoulou discovered a remarkable mechanism responsible for the dynamical formation of trapped surfaces. He proved that a trapped surface can form, even in vacuum space-time,  from completely dispersed initial configurations (i.e., free of any trapped surfaces) and purely by the focusing effect of gravitational waves. Christodoulou described an open set of initial data (so called \textit{short pulse ansatz}) on an outgoing null hypersurfaces without any symmetry assumptions. Based on the techniques he and Klainerman developed in \cite{Ch-K} proving the global stability of the Minkowski space-time, he managed to understand the cancelations among the different components of connection and curvature. This eventually enabled him to obtain a complete picture of how the various geometric quantities propagate along the evolution. Christodoulou also proved a version of the above result for data prescribed at the past null infinity. He showed that strongly focused gravitational waves, coming in from past null infinity, lead to a trapped surface. More precisely, he showed that if the incoming energy per unit solid angle in each direction in an advanced time interval $[0,\delta]$ at null infinity is sufficiently large (and yet sufficiently dispersed so that no trapped surfaces are present) the focusing effect will lead to gravitational collapses. Besides its important physical significance, from the point of view of partial differential equations, it establishes the first result on the long time dynamics in general relativity for general initial data which are not necessarily close to the Minkowski space-time.

In \cite{K-R-09}, Klainerman and Rodnianski extended Christodoulou's result. They introduced a parabolic scaling  and studied a broader class of initial data.  The new scaling allowed them to capture the hidden smallness of the nonlinear interactions in Einstein equations. Another key observation in their paper is that,  if one enlarges the admissible set of initial conditions, the corresponding propagation estimates are much easier to derive. Based on this idea, they gave a significant simplification of Christodoulou's result. At the same time, their relaxation of the propagation estimates were just enough to guarantee the formation of a trapped surface. Based on the geometric sharp trace theorems, which they have introduced earlier in \cite{K-R-Sharp} and applied to local well posedness for vacuum Einstein equations in \cite{K-R-04},\cite{K-R-05}, \cite{K-R-05-Rough} and \cite{K-R-05-Microlocolized-Rough}, they were also able to reduce the number of derivatives needed in the estimates from two derivatives on curvature (in Christodoulou's proof) to just one. The price for such a simpler proof,  with a larger set of data, is that the natural propagation estimates, consistent with the new scaling, are weaker than those of Christodoulou's. Nevertheless, once the main existence results are obtained, improved propagation estimates can be derived by assuming more conditions on the initial data, such as those consistent with Christodoulou's assumptions. This procedure allows one to recover Christodoulou's stronger results in a straightforward manner. We remark that, from a purely analytic point of view, the main difficulty of all the aforementioned results on dynamical formation of black holes is the proof of the long time existence results. The work \cite{K-R-09} overcomes this difficulty in an elegant way \footnote{\quad We remark that Klainerman and Rodnianski only considered the problem on a finite region. We also note that the problem from past null infinity has also been studied by Reiterer and Trubowitz \cite{R-T}, by a  different approach.}.

When matter fields are present, some black-hole formation results have been established much earlier, under additional symmetry assumptions. The most important such work is by Christodoulou (see \cite{Ch-91}). He studied the evolutionary formation of singularities for the Einstein-scalar field equations, under the assumption of spherical symmetry (notice that according to Birkhoff theorem, this assumption is not interesting in vacuum).  The incoming energy of the scalar field was the main factor leading to the gravitational collapse. We note also that the work \cite{Ch-91} provides a much more precise picture for the large scale structure of the space-time, than that available in the general case.

In all the aforementioned works, the initial data are prescribed on null hypersurfaces. It is however natural to study the question of formation of trapped surfaces for Cauchy initial data. We first recall that Cauchy data, i.e., the data defined on a space-like hypersurface,  must satisfy a system of partial differential equations, namely, the constraint equations \eqref{constraint}. The main advantage of using characteristic initial data is that one has complete freedom in specifying data without any constraint. Though Christodoulou's results for initial data prescribed at past null infinity predicts, indirectly, the existence of asymptotically flat Cauchy data, leading to a future trapped surface, it is natural to provide a constructive approach to this problem. In this connection, we mention an interesting piece of work \cite{S-Y} by Schoen and Yau. They showed that on a space-like hypersurface, there exists a trapped surface when the matter field is condensed in a small region, see also \cite{Y} for an improvement. Their proof analyzed the constraint equations and made use of their earlier work on positive mass theorem (especially the resolution of Jang's equation). We remark that their work is not evolutionary and matter fields are essential to the existence of trapped surfaces.

The goal of the current paper is to exhibit Cauchy initial data for vacuum Einstein field equations with a precise asymptotic behavior at space-like infinity, free of trapped surfaces, which lead to trapped surfaces along the evolution. We give the precise statement of the result in next section.

\subsection{Main Result}\label{Main Result}
Let $\Sigma$ be a three dimensional differentiable manifold diffeomorphic to $\mathbb{R}^3$ and  $(x_1,x_2,x_3)$ be the standard coordinate system. We also use $|x|$ to denote the usual radius function. Let $r_0  > 1$ be a given number, $\delta > 0$ a small positive number and $\varepsilon_0>0$ another small positive number. We divide $\Sigma$ into four concentric regions $\Sigma = \Sigma_{M} \bigcup \Sigma_{C} \bigcup \Sigma_{S} \bigcup \Sigma_{K}$, where
\begin{align*}
\Sigma_M &=\{x \, |\, |x| \leq r_0\}, \quad \Sigma_C =\{x \,| \,r_0 \leq |x| \leq r_1 \}, \\
\Sigma_S &=\{x\, |\, r_1 \leq |x| \leq r_2\}, \quad \text{and} \quad \Sigma_K =\{x \,|\,  |x| \geq r_2\}.
\end{align*}
The numbers $r_1,r_2$ will be fixed in the sequel such that $r_1-r_0=O(\delta)$ and $r_2-r_1=O(\varepsilon_0)$.

A Cauchy initial datum for vacuum Einstein field equations on $\Sigma$ consists of a Riemannian metric $\bar{g}$ and a symmetric two tensor $\bar{k}$ (as the second fundamental form) subject to the following  constraint equations:
\begin{equation}\label{constraint}
\begin{split}
R(\bar{g}) -|\bar{k}|^2 + \bar{h}^2 &= 0,\\
\text{div}_{\bar{g}} \bar{k} - \D \, \bar{h} &=0,
\end{split}
\end{equation}
where $R(\bar{g})$ is the scalar curvature of the metric $\bar{g}$ and $\bar{h}$ is the mean curvature.

In order to state the main theorem, we also need to specify a mass parameter $m_0 >0$ which will be defined in an explicit way in the course of the proof. Our main result is as follows:
\begin{MainTheorem} For any sufficiently small $\varepsilon >0$, there is a Riemannian metric $\bar{g}$ and a symmetric two tensor $\bar{k}$ on $\Sigma$ satisfying \eqref{constraint}, such that
\begin{itemize}
\item[1.] $\Sigma_M$ is a constant time slice in Minkowski space-time, in fact, $(\bar{g},\bar{k})=(\delta_{ij},0)$;
\item[2.] $\Sigma_K$ is isometric to a constant time slice all the way up to space-like infinity in a Kerr space-time with mass $m$ and angular momentum $\mathbf{a}$. Moreover, $|m-m_0|+|\mathbf{a}| \lesssim \varepsilon$;
\item[3.] $\Sigma$ is free of trapped surfaces;
\item[4.] Trapped surfaces will form in the future domain of dependence of $\Sigma$.
\end{itemize}
\end{MainTheorem}
\begin{remark}\label{Remark on condition on initial mass}
The mass parameter $m_0$ reflects the amount of incoming gravitational energy that we inject  into the Minkowski space-time through an outgoing null hypersurface. It can be computed explicitly from the initial conditions as follows (see Section \ref{Preliminaries} for definitions):
\begin{equation*}
m_0 = \frac{1}{4}\int_0^{\delta}|u_0|^2|\chih(\ub,u_0,\theta)|^2_{\gs}\D\ub.
\end{equation*}
\end{remark}
\begin{remark}
The regions $\Sigma_C$ and $\Sigma_S$ do not appear in the theorem. In the proof, we shall see that $\Sigma_S$ will be $\varepsilon$-close to a constant time slice in a Schwarzschild space-time with mass $m_0$ and $\Sigma_C$ will be constructed from Christodoulou's short pulse ansatz.
\end{remark}

\includegraphics[width = 5.5 in]{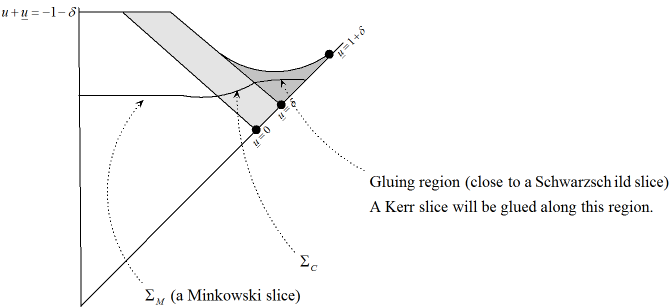}

We now sketch the heuristics of the proof with the help the above picture. We start by describing initial data on a truncated null cone (represented by the outgoing segment between the bottom vertex of the cone and the point $\ub=1+\delta$ in the picture).

The part of the cone between the vertex and $\ub=0$ will be a light cone in Minkowski space-time, thus its future development, i.e., the white region in the picture, is flat.

For data between $\ub=0$ and $\ub = \delta$, we use Christodoulou's initial data (from which a trapped surface will appear on the top of the light grey region) with an additional condition. This condition requires the incoming energy through this part of the cone be spherically symmetric. The main consequence of imposing this condition is that the incoming hypersurface from $\ub=\delta$ (represented by segment between the light grey region and the grey region in the picture) will be close to an incoming null cone in Schwarzschild space-time.

For data between $\ub=\delta$ and $\ub = 1+ \delta$, we require that its shear be identically zero. Together with the data on $\ub =\delta$, one can show that this part of the data will be close to an outgoing null cone in Schwarzschild space-time.

We construct the grey region (pictured above) by means of solving the vacuum Einstein equation using the initial data described in the previous two paragraphs. Thanks to the consequences of our additional condition, this region is close to a region in Schwarzschild space-time.

We then can choose a smooth space-like hypersurface such that it coincides with a constant time slice in Minkowski space-time in the white region (noted as $\Sigma_M$ in the picture) and it is also close to a constant time slice in Schwarzschild space-time (noted as the gluing region in the picture). Due to the closeness to Schwarzschild slice, we can do explicit computations to understand the obstruction space needed in Corvino-Schoen construction and we can eventually attach a Kerr slice to this region.

\subsection{Comments on the Proof}

We would like to address now the motivations for and difficulties in the proof.

As we stated in the Main Theorem, the ultimate goal is to obtain a Kerr slice $\Sigma_K$. To this end, we would like to use a gluing construction for the  constraint equations due to  Corvino and Schoen \cite{C-S}. Roughly speaking,  if two initial data sets are close to each other, this construction allows us to patch one to another  along a gluing region without changing the data outside. Similar results  have also been obtained by Chru\'{s}ciel and Delay in \cite{Chru-D}. We also note that in \cite{Chru-D}, there is a gluing construction with background metric close to the Minkowski metric, while in our work the  background will be close to a Schwarzschild slice instead. \cite{C-S} extended an earlier work of Corvino \cite{C} which proved a parallel  result for time symmetric data. Both constructions in \cite{C} and \cite{C-S} relied on the study of local deformations of the constraint maps (see Section \ref{The Work of Corvino-Schoen}). We will also have to make use of the deformation technique. Morally speaking, three ingredients are required to implement this approach:
\begin{itemize}
\item A reasonable amount of differentiability, say some $C^{k,\alpha}$ control on the background geometry. This is necessary since the local deformation techniques can only be proved for a relatively regular class of data, see Section \ref{The Work of Corvino-Schoen} or \cite{C-S} for details.
\item Some precise information on the background metric. This is essential since eventually we can only glue two data sets close to each other. A priori, one of them is a Kerr slice. Therefore, we expect to contract some space-time close to some Kerr slice on a given region. We call this requirement the \textit{smallness} condition since it will be captured eventually by the smallness of $\varepsilon$.
\item The gluing region must have a fixed size. This is important because the local deformation techniques work only on a fixed open set. If the gluing region shrinks to zero, it is not clear how one can proceed.
\end{itemize}

The first ingredient on the differentiability motivates us to derive the higher order energy estimates in Section \ref{Higher Order Energy Estimates}. In Christodoulou's work \cite{Chr}, he obtained energy estimates up to  two derivatives on the curvature components. He also had energy estimates for higher order derivatives, but the bounds were far from sharp. For our purposes, we require not only differentiabilities but also better bounds.

The second ingredient turns out to be the key of the entire proof. First of all, we remark that the higher order energy estimates are also important for the second ingredient. To clarify this connection, we review a technical part of Christodoulou's proof in \cite{Chr}. One of the most difficult estimates is the control on $\|\nabla \tensor \eta\|_{L^\infty}$. He used two mass aspect functions $\mu$ and $\mub$ coupled with $\eta$ and $\etab$ in Hodge systems. The procedure allowed him to gain one derivative from elliptic estimates to close the bootstrap argument. However, he had to incur a loss of smallness on the third derivatives of $\eta$. We remark that the loss only happened to the top order derivatives. In fact, for lower order derivatives, we can integrate some propagation equations to obtain smallness (but this integration loses one derivative instead!). In other words, we can afford one more derivative in exchange for the smallness. From the above, it is clear that once we have higher order energy estimates, we can expect better controls.

To obtain the second ingredient, we will show that a carefully designed data will evolve to a region that is $\varepsilon$-close to some Schwarzschild space-time; it is in this sense that we have precise information on the background metric. As a consequence, we will be able to write down explicitly a four dimensional space in the kernel of the formal adjoint of the linearized constraint maps. This kernel is the obstruction space for the gluing construction, and we shall use mass and angular momentum (which is of four dimensions in total) of the Kerr family to remove the obstruction. In this connection, we want to point out that in the work of Reiterer and Trubowitz \cite{R-T}, they also obtained such a near Schwarzschild space-time region (which to the best of our knowledge may not satisfy the requirement in our third ingredient).

To obtain an almost Schwarzschild slab, one needs a careful choice of initial data among the short pulse ansatz of Christodoulou. We shall impose the following condition:
\begin{equation*}
\int_0^{\delta} |u_0|^2|\chih(\ub,u_0,\theta)|^2_{\gs}\D\ub = \text{constant},
\end{equation*}
namely, independent of the $\theta$ variable. Heuristically, this says that the total incoming gravitational energy is the same for all the spherical directions. We expect that through such a pulse of incoming gravitational wave, after an advanced time interval $[0,\delta]$, we can regain  certain spherical symmetry. We remark that the above condition is the main innovation of the paper and will be crucial to almost all of our estimates.

Finally, we discuss our third ingredient. One may attempt to carry out the gluing directly on the space-time constructed in Christodoulou's work. However, when the parameter $\delta$ goes to zero, one runs the risk of shrinking the gluing region to zero. To solve this difficulty, we will further extend the short pulse ansatz beyond the the advanced time interval $[0,\delta]$ to $[0, 1+ \delta]$. This extension allows us to further solve Einstein equations to construct a space-time slab of a fixed size.

To close the section, we would like to discuss a way to prove the Main Theorem based on the work of Klainerman and Rodnianski \cite{K-R-09} instead of the stronger result \cite{Chr} of Christodoulou. Since we have to rely on higher order derivative estimates, we can actually start with the existence result in \cite{K-R-09} which only has control up to first derivatives on curvature. We then can assume more on the data, as we mentioned in the introduction, to derive more on the solution. As we will do in sequel, we can do the same induction argument on the number of derivatives on curvatures. This will yield the same higher order energy estimates and the rest of the proof remains the same.

\section{Preliminaries}\label{Preliminaries}

\subsection{Preliminaries on Geometry}\label{Preliminaries on Geometry}

We follow the geometric setup in \cite{Chr} and use the double null foliations for most of the paper.  We use $\DD$
 to denote the underlying space-time and use $g$ to denote the background 3+1 dimensional Lorentzian metric. We also use $\nabla$ to denote the Levi-Civita connection of the metric $g$.

Recall that we have two optical functions $\ub$ and $u$ defined on $\DD$ such that
\begin{equation*}
g(\nabla\ub,\nabla \ub)= g(\nabla u,\nabla u)=0.
\end{equation*}
The space-time $\DD$ is spanned by the level sets of $\ub$ and $u$. The functions $u$ and $\ub$ increase towards the future. We use $C_u$ to denote the outgoing null hypersurfaces generated by the level surfaces of $u$ and use ${\Cb}_{\ub}$ to denote the incoming null hypersurfaces generated by the level surfaces of $\ub$. We also use $S_{\ub,u}=\Cb_{\ub} \cap C_u$ to denote the space-like two sphere.  The notation $\Cb_{\ub}^{[u',u'']}$ refers to the part of the incoming cone $\Cb_{\ub}$ where $u'\leq u\leq u''$ and the notation $C_{u}^{[\ub',\ub'']}$ refers to the part of the outgoing cone $C_u$ where $\ub' \leq \ub \leq \ub''$.

Following Christodoulou \cite{Chr}, for the initial null hypersurface $C_{u_0}$ where $u_0<-2$ is a fixed constant, we require that $C_{u_0}^{[u_0,0]}$ is a flat light cone in Minkowski space-time and $S_{0,u_0}$ is the standard sphere with radius $|u_0|$. Thus, by solving Einstein vacuum equations, we know the future domain of dependence of $C_{u_0}^{[u_0,0]}$ is flat. It is equivalent to saying that the past of $\Cb_0$ can be isometrically embedded into Minkowski space-time. In particular, the incoming null hypersurface $\Cb_0$ coincides with an incoming light cone of the Minkowski space-time. In the sequel, on the initial hypersurface $C_{u_0}$, we shall only specify initial data on $C_{u_0}^{[0,\delta+1]}$.

We are ready to define various geometric quantities. The positive function $\Omega$ is defined by the formula $ \Omega^{-2}=-2g(\nabla\ub,\nabla u)$.  We then define the normalized null pair $(e_3, e_4)$ by $e_3=-2\Omega\nabla\ub$ and $e_4=-2\Omega\nabla u$. We also need two more null vector fields $\Lb=\Omega e_3$ and $L=\Omega e_4$. We remark that the flows generated by $\Lb$ and $L$ preserve the double null foliation. On a given two sphere $S_{\ub, u}$ we choose a local orthonormal frame $(e_1,e_2)$. We call $(e_1, e_2, e_3, e_4)$ a \emph{null frame}.  As a convention, throughout the paper, we use capital  Latin letters $A, B, C, \cdots$ to denote an index from $1$ to $2$, e.g. $e_A$ denotes either $e_1$ or $e_2$; we use little Latin letters $i, j, k, \cdots$ to denote an index from $1$ to $3$. Repeated indices should always be understood as summations.

Let $\phi$ to be a tangential tensorfield on $\DD$. By definition, $\phi$ being tangential means that $\phi$ is \textit{a priori} a tensorfield defined on the space-time $\DD$ and all the possible contractions of $\phi$ with either $e_3$ or $e_4$ are zeros. We use $D\phi$ and $\Db\phi$ to denote the projection to $S_{\ub,u}$ of usual Lie derivatives $\mathcal{L}_L\phi$ and $\mathcal{L}_{\Lb}\phi$. The space-time metric $g$ induces a Riemannian metric $\gs$ on $S_{\ub,u}$.  We use $\ds$ and $\nablas$ to denote the exterior differential and covariant derivative (with respect to $\gs$) on $S_{\ub,u}$.

Let $(\theta^A)_{A=1,2}$ be a local coordinate system on the two sphere $S_{0,u_0}$. We can extend $\theta^A$'s  to the whole $\DD$ by first setting $L(\theta^A)=0$ on $C_{u_0}$, and then setting $\Lb(\theta^A)=0$ on $\DD$. Therefore, we obtain a coordinate system $(\ub,u,\theta^A)$ on $\DD$. In such a coordinate, the Lorentzian metric $g$ takes the following form
\begin{align*}
g=-2\Omega^2(\D\ub\otimes\D u+\D u\otimes\D \ub)+\gs_{AB}(\D\theta^A-b^A\D\ub)\otimes(\D\theta^B-b^B\D\ub).
\end{align*}
The null vectors $\Lb$ and $L$ can be computed as $\Lb=\partial_u$ and $L=\partial_{\ub}+b^A\partial_{\theta^A}$. By construction, we have $b^A(\ub,u_0,\theta)=0$. In addition, we also require $\Omega(\ub,u_0,\theta)=\Omega(0,u,\theta)=1$.

We recall the definitions of null connection coefficients. Roughly speaking, the following quantities are Christoffel symbols of $\nabla$ according to the null frame $(e_1,e_2,e_3,e_4)$:
\begin{align*}
\chi_{AB}&=g(\nabla_Ae_4,e_B), \eta_A=-\frac{1}{2}g(\nabla_3e_A,e_4), \omega=\frac{1}{2}\Omega g(\nabla_4e_3,e_4),\\
\chib_{AB}&=g(\nabla_Ae_3,e_B), \etab_A=-\frac{1}{2}g(\nabla_4e_A,e_3), \omegab=\frac{1}{2}\Omega g(\nabla_3e_4,e_3).
\end{align*}
They are all tangential tensorfields. We define $\chi'=\Omega^{-1}\chi$, $\chib'=\Omega^{-1}\chi$ and $\zeta=\frac{1}{2}(\eta-\etab)$. The trace of $\chi$ and $\chib$ will play an important role in Einstein field equations and they are defined by $\tr\chi = \gs^{AB}\chi_{AB}$ and $\tr\chib = \gs^{AB}\chib_{AB}$. We remark that the trace is taken with respect to the metric $\gs$ and the indices are raised by $\gs$. By definition, we can check directly the following identities $\ds\log\Omega=\frac{1}{2}(\eta+\etab)$, $D\log\Omega=\omega$, $\Db\log\Omega=\omegab$ and $\Db b=4\Omega^2\zeta^\sharp$. Here, $\zeta^\sharp$ is the vector field dual to the $1$-form $\zeta$. In the sequel,  we will suppress the sign $\sharp$ and use the metric $\gs$ to identify $\zeta$ and $\zeta^\sharp$.

We can also decompose the curvature tensor into null curvature components:
\begin{align*}
\alpha_{AB}&=R(e_A,e_4,e_B,e_4),\beta_A=\frac{1}{2}R(e_A,e_4,e_3,e_4),\rho=\frac{1}{4}R(e_3,e_4,e_3,e_4),\\
\alphab_{AB}&=R(e_A,e_3,e_B,e_3),\betab_A=R(e_A,e_3,e_3,e_4),\sigma=\frac{1}{4}R(e_3,e_4,e_A,e_B)\epsilons^{AB},
\end{align*}
where $\epsilons$ is the volume form on $S_{\ub,u}$.

In order to express the Einstein vacuum equations with respect to a null frame, we have to introduce some operators. For a  symmetric tangential 2-tensorfield $\theta$, we use $\widehat{\theta}$ and $\tr\theta$ to denote the trace-free part and trace of $\theta$ (with respect to $\gs$). If $\theta$ is trace-free, $\Dh\theta$ and $\Dbh\theta$ refer to the trace-free part of $D\theta$ and $\Db\theta$. Let $\xi$ be a tangential $1$-form. We define some products and operators for later use. For the products, we define $(\theta_1,\theta_2)=\gs^{AC}\gs^{BD}(\theta_1)_{AB}(\theta_2)_{CD}$ and $\ (\xi_1,\xi_2)=\gs^{AB}(\xi_1)_A(\xi_2)_B$. This also leads to the following norms $|\theta|^2=(\theta,\theta)$ and $|\xi|^2=(\xi,\xi)$. We then define the contractions $(\theta\cdot\xi)_A=\theta_A{}^B\xi_B$, $(\theta_1\cdot \theta_2)_{AB}=(\theta_1)_A{}^C(\theta_2)_{CB}$, $\theta_1 \wedge\theta_2=\epsilons^{AC}\epsilons^{BD} (\theta_1)_{AB}(\theta_2)_{CD}$ and $\xi_1\tensor \xi_2=\xi_1\otimes\xi_2+\xi_2\otimes\xi_1-(\xi_1,\xi_2)\gs$. The Hodge dual for $\xi$ is defined by $\prescript{*}{}\xi=\epsilons_A{}^C\xi_C$. For the operators, we define $\divs\xi=\nablas^A\xi_A$, $\curls\xi=\epsilons^{AB}\nablas_A\xi_B$ and $(\divs\theta)_A=\nablas^B\theta_{AB}$. We finally define a traceless operator $(\nablas\tensor\xi)_{AB}=(\nablas\xi)_{AB}+(\nablas\xi)_{BA}-\divs\xi \,\gs_{AB}$.

For the sake of simplicity, we will use shorthands $\Gamma$ and $R$ to denote an arbitrary connection coefficient and an arbitrary null curvature component. We also introduce an schematic way to write products. Let $\phi$ and $\psi$ be arbitrary tangential tensorfields, we also use $\phi\cdot\psi$ to denote an arbitrary contraction of $\phi$ and $\psi$ by $\gs$ and $\epsilons$. This schematic notation only captures the quadratic nature of the product and it will be good enough for most of the cases when we derive estimates. As an example, the notation $\Gamma \cdot R$ means a sum of products between a connection coefficient and a curvature component.

We use a null frame $(e_1,e_2,e_3,e_4)$ to decompose the Einstein vacuum equations $\text{Ric}(g)=0$ into components. This leads to the following null structure equations (where  $K$ is the Gauss curvature of $S_{\ub,u}$):
\begin{align}\label{NSE_Dh_chih}
\Dh \chih'&=-\alpha,\\
\label{NSE_D_trchi}
D\tr\chi'&=-\frac{1}{2}\Omega^2(\tr\chi')^2-\Omega^2|\chih'|^2,\\
\label{NSE_Dbh_chibh}
\Dbh \chibh'&=-\alphab,\\
\label{NSE_Db_trchib}
\Db\tr\chib'&=-\frac{1}{2}\Omega^2(\tr\chib')^2-\Omega^2|\chibh'|^2,\\
\label{NSE_D_eta}
D\eta &= \Omega(\chi \cdot\etab-\beta),\\
\label{NSE_Db_etab}
\Db\etab &= \Omega(\chib \cdot\eta+\betab),\\
\label{NSE_D_omegab}
D  \omegab &=\Omega^2(2(\eta,\etab)-|\eta|^2-\rho),\\
\label{NSE_Db_omega}
\Db  \omega &=\Omega^2(2(\eta,\etab)-|\etab|^2-\rho),\\
\label{NSE_Gauss}
K&=-\frac{1}{4}\tr \chi\tr\chib+\frac{1}{2}(\chih,\chibh)-\rho,\\
\label{NSE_div_chih}
\divs \chih'&=\frac{1}{2}\ds \tr \chi'-\chih'\cdot\eta+\frac{1}{2}\tr \chi'\eta-\Omega^{-1}\beta,\\
\label{NSE_div_chibh}
\divs \chibh'&=\frac{1}{2}\ds \tr \chib'-\chibh'\cdot\etab+\frac{1}{2}\tr \chib'\etab-\Omega^{-1}\betab,\\
\label{NSE_curl_eta}
\curls \eta&=\sigma-\frac{1}{2}\chih \wedge\chibh,\\
\label{NSE_curl_etab}
\curls \etab&= -\sigma + \frac{1}{2}\chih \wedge\chibh,\\
\label{NSE_D_chibh}
\Dh(\Omega\chibh)&=\Omega^2(\nablas \tensor \etab + \etab \tensor \etab +\frac{1}{2}\tr\chi\chibh-\frac{1}{2}\tr\chib \chih),\\
\label{NSE_D_trchib}
D(\Omega\tr\chib)&=\Omega^2(2\divs\etab+2|\etab|^2-(\chih,\chibh)-\frac{1}{2}\tr\chi\tr\chib+2\rho),\\
\label{NSE_Db_chih}
\Dbh(\Omega\chih)&=\Omega^2(\nablas \tensor \eta + \eta \tensor \eta +\frac{1}{2}\tr\chib\chih-\frac{1}{2}\tr\chi \chibh),\\
\label{NSE_Db_trchi}
\Db(\Omega\tr\chi)&=\Omega^2(2\divs\eta+2|\eta|^2-(\chih,\chibh)-\frac{1}{2}\tr\chi\tr\chib+2\rho).
\end{align}

We also use the null frame to decompose second Bianchi identity $\nabla_{[a} R_{bc]de} = 0$ into components. This leads the following null Bianchi equations,
\begin{align}\label{NBE_Db_alpha}
\Dbh\alpha-\frac{1}{2}\Omega\tr\chib \alpha+2\omegab\alpha+\Omega\{-\nablas\tensor\beta -(4\eta+\zeta)\tensor \beta+3\chih \rho+3{}^*\chih \sigma\}&=0,\\
\label{NBE_D_alphab}
\Dh\alphab-\frac{1}{2}\Omega\tr\chi \alphab+2\omega\alphab+\Omega\{\nablas\tensor\betab +(4\etab-\zeta)\tensor \betab+3\chibh \rho-3{}^*\chibh \sigma\}&=0,\\
\label{NBE_D_beta}
D\beta+\frac{3}{2}\Omega\tr\chi\beta-\Omega\chih\cdot\beta-\omega\beta-\Omega\{\divs\alpha+(\etab+2\zeta)\cdot\alpha\}&=0,\\
\label{NBE_Db_betab}
\Db\betab+\frac{3}{2}\Omega\tr\chib\betab-\Omega\chibh\cdot\betab-\omegab\betab+\Omega\{\divs\alphab+(\eta-2\zeta)\cdot\alphab\}&=0,\\
\label{NBE_Db_beta}
\Db\beta+\frac{1}{2}\Omega\tr\chib\beta-\Omega\chibh \cdot \beta+\omegab \beta-\Omega\{\ds \rho+{}^*\ds \sigma+3\eta\rho+3{}^*\eta\sigma+2\chih\cdot\betab\}&=0,\\
\label{NBE_D_betab}
D\betab+\frac{1}{2}\Omega\tr\chi\betab-\Omega\chih \cdot \betab+\omega \betab+\Omega\{\ds \rho-{}^*\ds \sigma+3\etab\rho-3{}^*\etab\sigma-2\chibh\cdot\beta\}&=0,\\
\label{NBE_D_rho}
D\rho+\frac{3}{2}\Omega\tr\chi \rho-\Omega\{\divs \beta+(2\etab+\zeta,\beta)-\frac{1}{2}(\chibh,\alpha)\}&=0,\\
\label{NBE_Db_rho}
\Db\rho+\frac{3}{2}\Omega\tr\chib \rho+\Omega\{\divs \betab+(2\eta-\zeta,\betab)-\frac{1}{2}(\chih,\alphab)\}&=0,\\
\label{NBE_D_sigma}
D\sigma+\frac{3}{2}\Omega\tr\chi\sigma+\Omega\{\curls\beta+(2\etab+\zeta,{}^*\beta)-\frac{1}{2}\chibh\wedge\alpha\}&=0,\\
\label{NBE_Db_sigma}
\Db\sigma+\frac{3}{2}\Omega\tr\chib\sigma+\Omega\{\curls\betab+(2\etab-\zeta,{}^*\betab)+\frac{1}{2}\chih\wedge\alphab\}&=0.
\end{align}

To conclude this subsection, we recall how one prescribes characteristic data for Einstein vacuum equations on two transversally intersecting null hypersurfaces, say $C^{[0,\delta+1]}_{u_0}\cup \Cb_0$ in our current situation. In general, the initial data given on $C_{u_0}\cup\Cb_0$ should consist of the full metric $\gs_{0,u_0}$, torsion $\zeta$, outgoing expansion $\tr\chi$ and incoming expansion $\tr\chib$ on the intersecting sphere $S_{0,u_0}$ together with the conformal geometry on $C^{[0,\delta+1]}_{u_0}$ and $\Cb_{0}$. As we observed, the incoming surface $\Cb_{0}$ is a fixed cone in the Minkowski space-time, thus $\gs_{0,u_0}$, $\zeta$,$\tr\chi$ and $\tr\chib$ are already fixed on $S_{0,u_0}$. Therefore, to specify initial data, we only need to specify the conformal geometry on $C^{[0,\delta+1]}_{u_0}$. We will see how the \emph{short pulse} data of Christodoulou is prescribed in next subsection.

\subsection{The Work of Christodoulou}\label{The Work of Christodoulou}
We first discuss Christodoulou's short pulse ansatz presented in \cite{Chr}. As we mentioned earlier, we need to specify the conformal geometry on $C_{u_0}^{[0,\delta + 1]}$. Let $\Phi_{\ub}$ be the one parameter group generated by $L$. We can rewrite the induced metric $\gs|_{S_{\ub,u_0}}$ uniquely as $\gs|_{S_{\ub,u_0}}=(\phi|_{S_{\ub,u_0}})^2\widehat{\gs}|_{S_{\ub,u_0}}$, where $\phi|_{S_{\ub,u_0}}$ is a positive function, such that the metric $\Phi_{\ub}^*\widehat{\gs}|_{S_{\ub,u_0}}$ on $S_{0,u_0}$ has the same volume form as $\gs|_{S_{0,u_0}}$.  In this language, we only need to specify $\Phi_{\ub}^*\widehat{\gs}|_{S_{\ub,u_0}}$ freely on $S_{0,u_0}$ since it yields the conformal geometry on $C_{u_0}^{[0,\delta + 1]}$.

Let $\{(U_1,(\theta^A_1)),(U_2,(\theta^A_2))\}_{A=1,2}$ be the two stereographic charts on $S_{0,u_0}$. Thus, the round metric $\gs|_{S_{0,u_0}}$ is expressed as $(\gs|_{S_{0,u_0}})_{AB}(\theta)=\frac{|u_0|^2}{(1+\frac{1}{4}|\theta|^2)^2}\delta_{AB}$ with $\theta=\theta_1$ or $\theta_2$ and $|\theta|^2=|\theta^1|^2+|\theta^2|^2$. Since we require that $\widehat{\gs}(\ub)=\Phi_{\ub}^*\widehat{\gs}|_{S_{\ub,u_0}}$ on $S_{0,u_0}$ has the same volume form as $\gs|_{S_{0,u_0}}$, that is,
\begin{equation*}
\det(\widehat{\gs}(\ub)_{AB})(\theta)=\det(\gs(0)_{AB})(\theta)=\frac{|u_0|^4}{(1+\frac{1}{4}|\theta|^2)^4},
\end{equation*}
$\widehat{\gs}(\ub)$ is given by
\begin{equation*}
\widehat{\gs}(\ub)_{AB}(\theta)=\frac{|u_0|^2}{(1+\frac{1}{4}|\theta|^2)^2}m_{AB}(\ub,\theta)=\frac{|u_0|^2}{(1+\frac{1}{4}|\theta|^2)^2}\exp\psi_{AB}(\ub,\theta),
\end{equation*}
where $m_{AB}$ takes value in the set of positive definite symmetric matrices with determinant one and $\psi_{AB}$ takes value in the set of symmetric trace-free matrices. After this reduction, to prescribe initial data, we only need to specify a function
\begin{equation*}
\psi: [0,\delta+1] \times S_{0,u_0} \longrightarrow \widehat{S}_2, \quad (\ub, \theta) \mapsto \exp\psi_{AB}(\ub,\theta),
\end{equation*}
where $\widehat{S}_2$ denotes the set of $2\times 2$ symmetric trace-free matrices.

In Christodoulou's work \cite{Chr}, he only provided data on $C_{u_0}^{[0,\delta]}$, i.e. $\ub \in [0,\delta]$. First he chose a smooth compactly supported $\widehat{S}_2$-valued function $\psi_0 \in C^\infty_c((0,1) \times S_{0,u_0})$. Then he called the following specific data
\begin{equation}\label{shortpulse}
\psi(\ub,\theta)=\frac{\delta^{\frac{1}{2}}}{|u_0|}\psi_0(\frac{\ub}{\delta},\theta),
\end{equation}
the \emph{short pulse ansatz} and he called $\psi_0$ the \emph{seed data}.

\begin{remark}\label{Remark on initial ansatz}
A key ingredient for the current work is to give further restriction on the seed data $\psi_0$ (see \eqref{integral=m0} or \eqref{integral=m0 different}).  We shall further extend $\psi$ to the whole region $[0,\delta+1] \times S_{0,u_0}$ by zero.
\end{remark}

To state the main theorems in \cite{Chr}, especially the energy estimates, we also need to define some norms. Let $k \in \mathbb{Z}_{\geq 0}$ be a non-negative integer, on each outgoing cone $C_u^{[0,\delta]}$, we define
\begin{equation*}
\mathcal{R}_k(u)=\delta^{-\frac{1}{2}}|u|^{-1}(\|\delta^{\frac{3}{2}}|u|(|u|\nabla)^k\alpha\|
+\|\delta^{\frac{1}{2}}|u|^2(|u|\nabla)^k\beta\|+\||u|^3(|u|\nabla)^k(\rho,\sigma)\
+\|\delta^{-1}|u|^4(|u|\nabla)^k\underline{\beta}\|),
\end{equation*}
and
\begin{align*}
\mathcal{O}_{k+1}(u)&=\delta^{-\frac{1}{2}}|u|^{-1}\|\delta^{\frac{1}{2}}|u|(|u|\nabla)^{k+1}\widehat{\chi}\|
+\||u|^2(|u|\nabla)^{k+1}(\mathrm{tr}\chi-\frac{2}{|u|})\|+\|\delta^{-\frac{1}{2}}|u|^2(|u|\nabla)^{k+1}\underline{\widehat{\chi}}\|\\
&+\|\delta^{-1}|u|^3(|u|\nabla)^{k+1}(\mathrm{tr}\underline{\chi}+\frac{2}{|u|}-\frac{2\underline{u}}{|u|^2})\|+\||u|^2(|u|\nabla)^{k+1}(\eta,\underline{\eta})\| \\
&+\|\delta^{-1}|u|^3(|u|\nabla)^{k+1}\underline{\omega}\|+\|\delta^{\frac{1}{2}}|u|(|u|\nabla)^{k+1}\omega\|,
\end{align*}
where all the norms $\|\cdot \|$ are taken with respect to  $\|\cdot \|_{L^2(C_u^{[0,\delta]})}$; on each incoming cone $\Cb_{\ub}$, we define
\begin{equation*}
\underline{\mathcal{R}}_k[\underline{\alpha}](\underline{u})=\||u|^{-\frac{3}{2}}\delta^{-\frac{3}{2}}|u|^{\frac{9}{2}}(|u|\nabla)^k\underline{\alpha}\|_{L^2(\underline{C}_{\underline{u}})}.
\end{equation*}

Recall that we use $\Gamma$ or $R$ to denote an arbitrary connection coefficient or curvature component. We now use $\mathcal{R}_k[R](u),\mathcal{O}_{k+1}[\Gamma](u)$ to represent the corresponding norms of this given component. To be more precise, let $\phi = \Gamma$ or $R$, we have
\begin{align*}
&\mathcal{R}_k[\phi](u)=\delta^{-\frac{1}{2}}|u|^{-1}\|\delta^{-r(\phi)}|u|^{-s(\phi)}(|u|\nablas)^k\phi\|_{L^2(C_u)},\\
&\mathcal{\underline{R}}_k[\phi](\ub)=\||u|^{-\frac{3}{2}}\delta^{-r(\phi)}|u|^{-s(\phi)}(|u|\nablas)^k\phi\|_{L^2(\Cb_{\ub})},\\
&\mathcal{O}_{k+1}[\phi](u)=\delta^{-\frac{1}{2}}|u|^{-1}\|\delta^{-r(\phi)}|u|^{-s(\phi)}(|u|\nablas)^{k+1}\phi\|_{L^2(C_u)}.
\end{align*}
where the $r(\phi)$ or $s(\phi)$ can be easily retrieved from the precise definitions for each specific component. Similarly, we introduce
\begin{align*}
&\mathcal{R}^4_{k-1}[\phi](\ub,u)=|u|^{-\frac{1}{2}}\|\delta^{-r(\phi)}|u|^{-s(\phi)}(|u|\nablas)^{k-1}\phi\|_{L^4(S_{\ub,u})},\\
&\mathcal{R}^\infty_{k-2}[\phi](\ub,u)=\|\delta^{-r(\phi)}|u|^{-s(\phi)}(|u|\nablas)^{k-2}\phi\|_{L^\infty(S_{\ub,u})}.
\end{align*}
Finally, we define total norms
\begin{equation*}
\mathcal{R}_k=\sup_u\mathcal{R}_k(u), \, \mathcal{O}_{k+1}=\sup_u\mathcal{O}_{k+1}(u), \,\mathcal{R}_{\le k} = \sum_{j \leq k}\mathcal{R}_{j}, \,\mathcal{O}_{\le k} = \sum_{j \leq k}\mathcal{O}_{j}.
\end{equation*}

We state the main result in \cite{Chr}. Roughly speaking, it asserts that up to two derivatives on curvature, all the norms ($k\leq 2$) defined above can propagate along the evolution of Einstein vacuum equations.
\begin{ChristodoulouMainEstimates}[Theorem 16.1 in \cite{Chr}]\label{Christodoulou}
If $\delta>0$ is sufficiently small (depending on the $C^8$ bound of the seed data $\psi_0$), there exists a unique solution for the Einstein field equations on a domain $M$ corresponding to $0\leq\ub\leq\delta$ and $u_0\leq u\leq-1-\delta$. Moreover, the following total norms
\begin{equation*}
\mathcal{R}_{\le2},\,\underline{\mathcal{R}}_{\le2}[\alphab], \,\mathcal{R}^4_{\le1},\, \mathcal{R}^\infty_0,\, \mathcal{O}_{\le3},\, \mathcal{O}^4_{\le2},\, \mathcal{O}^\infty_{\le1},
\end{equation*}
are bounded by a constant depending on the $C^8$ bound of the seed data $\psi_0$.
\end{ChristodoulouMainEstimates}
The main consequence of this theorem is the dynamical formation of trapped surfaces. In \cite{Chr}, Christodoulou showed that if
\begin{equation*}
\frac{1}{8}\int_{0}^1 |\frac{\partial \psi_0}{\partial s}(s,\theta)|^2 \D s >1,
\end{equation*}
uniformly in $\theta$, then a trapped surface forms in the future of $C_{u_0}$.

\subsection{The Work of Corvino-Schoen}\label{The Work of Corvino-Schoen}
We follow closely the notations used by Corvino and Schoen in \cite{C-S} unless there are conflicts with the current work. We first recall some definitions on a given three dimensional space-like slice $\Sigma$ in a vacuum space-time. The vacuum initial data on $\Sigma$ consist of a Riemannian metric $\bar{g}$ and a symmetric $2$-tensor $\bar{k}$ subject to the following constraints (to distinguish from the notations in four dimensions, we shall always use barred notations in three dimensions, e.g. $\bar{\nabla}$ denotes the Levi-Civita connection associated to $\bar{g}$):
\begin{equation*}
\begin{split}
R(\bar{g}) -|\bar{k}|^2 + \bar{h}^2 &= 0,\\
\text{div}_{\bar{g}} \bar{k} - \D \, \bar{h} &=0,
\end{split}
\end{equation*}
where $\bar{h} = \text{tr}_{\bar{g}}\bar{k} = \bar{g}_{ij} \bar{k}^{ij}$ denotes the mean curvature, $\text{div}_{\bar{g}}\bar{k}_i =
\bar{\nabla}_j \bar{k}^{ij}$, and all quantities are computed with respect to $\bar{g}$. One then rewrites the constraint equations by introducing the momentum tensor $\bar{\pi}_{ij} = \bar{k}_{ij} -\bar{h}\cdot{\bar{g}}_{ij}$. Let $\mathcal{H}$ and $\Phi$ be the following maps:
\begin{equation*}
\begin{split}
\mathcal{H}(\bar{g},\bar{\pi}) &= R_{\bar{g}}+ \frac{1}{2}(\tr{\bar{\pi}})^2-|\bar{\pi}|^2,\\
\Phi(\bar{g},\bar{\pi})  &=(\mathcal{H}(\bar{g},\bar{\pi}), \text{div}_{\bar{g}}\bar{\pi}).
\end{split}
\end{equation*}
The constraints then take the form $\Phi(\bar{g},\bar{\pi})= 0$.

We use $\mathcal{M}^{k,\alpha}(\Sigma)$, $\mathcal{S}^{k,\alpha}(\Sigma)$ and $\mathcal{X}^{k,\alpha}(\Sigma)$ to denote the set of Riemannian metric, symmetric two tensors and vector fields on $\Sigma$ with $C^{k,\alpha}$ regularity respectively. Thus, we have
\begin{equation*}
\Phi : \mathcal{M}^{k+2,\alpha}(\Sigma) \times \mathcal{S}^{k+2,\alpha}(\Sigma) \rightarrow C^{k,\alpha}(\Sigma) \times \mathcal{X}^{k+1,\alpha}(\Sigma).
\end{equation*}
The formal $L^2$-adjoint operator $D \Phi^*_{(\bar{g},\bar{\pi})}$ of the linearization $D \Phi_{(\bar{g},\bar{\pi})}$ is then given by
\begin{equation}\label{adjoint_of_linearization}
\begin{split}
D\mathcal{H}^*_{(\bar{g},\bar{\pi})}(f) &= ((L^*_{\bar{g}}f)_{ij}+(\tr_{\bar{g}}\bar{\pi} \cdot \bar{\pi}_{ij}-2 \bar{\pi}_{ik}\bar{\pi}^{k}{}_{j})f,(\tr_{\bar{g}}\bar{\pi} \cdot \bar{g}_{ij}-2 \bar{\pi}_{ij})f),\\
D \text{div}^*_{(\bar{g},\bar{\pi})}(X)  &=\frac{1}{2}(\mathcal{L}_X \bar{\pi}_{ij}+ \bar{\nabla}_k X^k \pi_{ij}-(X_i (\bar{\nabla}_k \pi^k{}_j + X_j \bar{\nabla}_k \pi^{k}{}_i)\\
&\quad-(\bar{\nabla}_m X_k \pi^{km} + X_k \bar{\nabla}_m \pi^{mk})g_{ij}, \quad -\mathcal{L}_Xg_{ij}),
\end{split}
\end{equation}
where $L^*_{\bar{g}}f = -\triangle_{\bar{g}} f + \bar{\nabla}^2 f - f \cdot\text{Ric}(\bar{g})$ is the formal $L^2$-adjoint of the linearization of the scalar curvature operator.

Let $\Omega \subset \Sigma$ be a given bounded domain with smooth boundary and $\rho$ be a smooth positive function on $\Omega$ which near $\Omega$ decays as a power of distance to the boundary, i.e. $\rho \sim d^N$ where $d$ is the distance to the boundary and $N$ will be fixed later (see the theorem below). The weighted H\"{o}lder space $C^{k,\alpha}_{\rho^{-1}}(\Omega)$ is defined by the norm $\|f\|_{C^{k,\alpha}_{\rho^{-1}}} = \|f\rho^{-\frac{1}{2}}\|_{C^{k,\alpha}}$ in the obvious way; similarly, we can define those spaces for tensors. We are ready to state the local deformation theorem which will play a key role in our gluing construction.
\begin{LocalDeformationTheorem}[Theorem 2 in \cite{C-S}]
Let $\zeta \in C_0^{\infty}(\Omega)$ be a bump function and $(g_0,\pi_0) \in \mathcal{M}^{k+4,\alpha}(\Sigma) \times \mathcal{S}^{k+3,\alpha}(\Sigma)$. Then for $N$ sufficiently large, there is an $\varepsilon > 0$ such that for all pair $(v,W) \in C^{k,\alpha}(\bar{\Omega}) \times \mathcal{X}^{k+1,\alpha}(\bar{\Omega})$ with the support of $(v,W)-\Phi(g_0,\pi_0)$ contained in $\bar{\Omega}$ and with $\|(v,W)-\Phi(g_0,\pi_0)\|_{C_{\rho^{-1}}^{k,\alpha}(\bar{\Omega}) \times \mathcal{X}_{\rho^{-1}}^{k+1,\alpha}(\bar{\Omega})} < \varepsilon$, there is a pair $(g,\pi) \in \mathcal{M}^{k+2,\alpha}(\Sigma) \times \mathcal{S}^{k+2,\alpha}(\Sigma)$, such that $\Phi(g,\pi)-(v,W) \in \zeta \cdot \text{Ker} \,D\Phi^*_{(g_0,\pi_0)}$ in $\Omega$ and $(g,\pi) = (g_0,\pi_0)$ outside $\Omega$. Moreover, $(g,\pi) \in \mathcal{M}^{k+2,\alpha}(\Sigma) \times \mathcal{S}^{k+2,\alpha}(\Sigma)$ depends continuously on $(v,W)-\Phi(g_0,\pi_0) \in C_{\rho^{-1}}^{k,\alpha}(\bar{\Omega}) \times \mathcal{X}_{\rho^{-1}}^{k+1,\alpha}(\bar{\Omega})$.
\end{LocalDeformationTheorem}

In \cite{C-S}, Corvino and Schoen used this theorem to approximate asymptotically flat initial data for the vacuum Einstein field equations by solutions which agree with the original data inside a given domain and are identical to that of a suitable Kerr slice. We will use this theorem in a similar manner in a different situation to prove our main theorem.

\subsection{The Structure of the Proof}

This section is devoted to an outline of the proof. It consists of three steps.

\begin{itemize}
\item \textbf{Step 1}\, Higher Order Energy Estimates.
\end{itemize}
As we mentioned in the introduction, since we would like to use Corvino-Schoen construction (see the Local Deformation Theorem in Section \ref{The Work of Corvino-Schoen}) which requires certain regularity of the space-time, we are obliged to derive higher order energy estimates. This is accomplished in Section \ref{Higher Order Energy Estimates}.

We will derive the estimates on $M$ where $0\leq \ub \leq \delta$. Christodoulou's work \cite{Chr} already proved such estimates on the level of two derivatives on the curvature (see the Christodoulou Main Estimates in Section \ref{The Work of Christodoulou}). His estimates are already good enough to construct the space-time. Thus, in our case, we can use an induction (on the number of derivatives) argument instead a bootstrap argument.

Although this part should be regarded as the routine proof of the persistence of regularity for vacuum Einstein field equations, we would like to emphasize that by affording more derivatives, we can gain smallness in $\delta$ compared to Christodoulou's work. For example, in \cite{Chr}, in $L^\infty$ norm, we have
\begin{equation*}
|\nablas \eta|\lesssim \frac{1}{|u|^3},
\end{equation*}
while, in the current work, since we can control the third derivatives on the curvature components, we actually have
\begin{equation*}
|\nablas \eta|\lesssim \frac{1}{|u|^3}\delta^{\frac{1}{2}}.
\end{equation*}
The gain in $\delta$ will play a crucial role in Step 2.

\begin{itemize}
\item \textbf{Step 2}\, Construction of a Transition Region Close to the Schwarzschild Space-time with Mass $m_0$.
\end{itemize}

This step is the main innovation of the paper and it is completed in Section \ref{Construction of the Transition Region}.  Roughly speaking, this step build a bridge from Christodoulou's work to the Corvino-Schoen construction. By solving the vacuum Einstein field equations, we can construct a region with a fixed size (independent of the small parameter $\delta$) which is close (measured by $\delta$) to a region in the domain of outer communication in the Schwarzschild space-time with mass $m_0$. It is precisely on this region, or more precisely a space-like hypersurface inside this region, where we can use the Corvino-Schoen construction.

As we mentioned in Remark \ref{Remark on initial ansatz}, in addition to Christodoulou's short pulse ansatz described in Section \ref{The Work of Christodoulou}, we need to impose one more condition on the initial data defined on $C_{u_0}^{[0,\delta]}$, that is, for all $\theta$,
\begin{equation}\label{integral=m0}
\int_0^1\left|\frac{\partial\psi_0}{\partial s}(s,\theta)\right|^2\D s=16m_0.
\end{equation}
This condition is slightly different from the condition proposed in Remark \ref{Remark on condition on initial mass}, namely,
\begin{equation}\label{integral=m0 different}
|u_0|^2 \int_0^{\delta}|\chih(\ub,u_0,\theta)|^2_{\gs}\D\ub = 4 m_0.
\end{equation}
The advantage of using the former over the latter is that as $\delta$ changes, \eqref{integral=m0} is always valid, but \eqref{integral=m0 different} is not. However, under the condition \eqref{integral=m0}, \eqref{integral=m0 different} is valid up to an error of size $\delta$ to some positive power. The condition \eqref{integral=m0 different} is more physical and easier to use. In this paper, although we impose data via the seed function $\psi_0$, we shall stick to the the latter condition.

The condition \eqref{integral=m0 different} is the key ingredient to Step 2 and most of the estimates are directly tied to it. The condition \eqref{integral=m0 different} also has a clear physical interpretation: it requires the incoming gravitational energy per solid angle in the advanced time interval $[0,\delta]$ to be the same for all angles $\theta$. Roughly speaking, we impose certain spherical symmetry on the initial data.

In Section \ref{Section Geometry on intersecting sphere}, we show that the sphere $S_{\delta,u_0}$ is $\delta$-close to a given sphere in the Schwarzschild space-time with mass $m_0$; here, the closeness is measured in $C^k$ norms. The proof is based on condition \eqref{integral=m0 different}.

In Section \ref{Section Geometry on incoming cone}, we show that the incoming cone $\Cb_{\delta}$ is $\delta$-close to a given incoming cone in the Schwarzschild space-time with mass $m_0$. This cone is rooted on the given sphere Section \ref{Section Geometry on intersecting sphere}. The proof is based on condition \eqref{integral=m0 different} as well as the higher order energy estimates derived in Section \ref{Higher Order Energy Estimates}.

In Section \ref{Section Geometry on outgoing cone}, we further extend our data on $C_{u_0}^{[\delta, \delta+1]}$ smoothly by setting
\begin{equation}\label{data on C delta plus one}
\chih\equiv 0, \quad \delta\leq \ub \leq \delta +1.
\end{equation}
We then show that the outgoing cone $C_{u_0}^{[\delta, \delta+1]}$ is $\delta$-close to a given outgoing cone in the Schwarzschild space-time with mass $m_0$. This cone is also rooted on the given sphere Section \ref{Section Geometry on intersecting sphere}. The proof is based on condition \eqref{integral=m0 different}.

In Section \ref{Section transition slice}, since the Sections \ref{Section Geometry on intersecting sphere}, \ref{Section Geometry on incoming cone} and \ref{Section Geometry on outgoing cone} provide a characteristic data set for the vacuum Einstein field equations, we will solve the field equations to further extend Christodoulou's solution. By virtue of the closeness to the Schwarzschild data, we show that the resulting space-time is close to a region in the domain of outer communication in the Schwarzschild space-time with mass $m_0$. Furthermore, we will also see that this region has a fixed size, namely, it will not shrink to zero when we decrease the parameter $\delta$.

\begin{itemize}
\item \textbf{Step 3}\, Gluing a Kerr Slice.
\end{itemize}
This step is completed in Section \ref{Gluing Construction}. As a consequence of Section \ref{Section transition slice}, we can choose a three dimensional space-like region for gluing. Then we follow the procedure in \cite{C} and \cite{C-S}. Thanks to the closeness derived in Step 2, we can almost explicitly write down the kernel of the adjoint of the linearized operator of the constraint map. This allows one to use the Local Deformation Theorem in Section \ref{The Work of Corvino-Schoen}. The kernel will have four dimensions. Combined with a fixed point argument, we use a four parameter family of Kerr space-time, namely the mass $m$ and the angular momentum $\mathbf{a}$ to kill the kernel. We also show that the resulting initial data is free of trapped surfaces.

\section{Higher Order Energy Estimates}\label{Higher Order Energy Estimates}
In this section, we derive energy estimates for higher order derivatives on curvature components. At the same time, this also yields the control of higher order derivatives on connection coefficients. More precisely, we shall prove
\begin{theorem}\label{Higherorder}
 Let $k \in \mathbb{Z}_{\geq 3}$. If $\delta>0$ is sufficiently small depending $C^{k+N}$ norm of the seed data $\psi_0$ for sufficient large $N$, then the following quantities
\begin{align*}
\mathcal{R}_{\le k}, \underline{\mathcal{R}}_{\le k}[\alphab], \mathcal{R}^4_{\le{k-1}}, \mathcal{R}^\infty_{\le k-2}, \mathcal{O}_{\le{k+1}}, \mathcal{O}^4_{\le k}, \mathcal{O}^\infty_{\le k-1},
\end{align*}
are bounded by a constant depending on $C^{k+N}$ bounds of the seed data.
\end{theorem}
The integer $N$ is chosen such that for all $k$, $\mathcal{R}_k(u_0)$ is bounded by a constant depending on the $C^{k+N}$ bound of the seed data. We remark that for $k\leq 2$, the theorem was proved in \cite{Chr} with $N=6$. In fact the precise value of $N$ is not important in the current work. Then we shall proceed by induction on $k$. We then make the induction assumption as follows, for $k \geq 3$
\begin{align}\label{induction}
\mathcal{R}_{\le k-1} + \underline{\mathcal{R}}_{\le k-1}[\alphab]+\mathcal{R}^4_{\le{k-2}}+\mathcal{R}^\infty_{\le k-3}+\mathcal{O}_{\le{k}} + \mathcal{O}^4_{\le k-1}+\mathcal{O}^\infty_{\le k-2} \le F_{k-1+N},
\end{align}
where $F_i$ is depending only on the $C^i$ bound of seed data.

\subsection{Sobolev Inequalities and Elliptic Estimates for Hodge Systems}
In this section, we recall the elliptic estimates for Hodge systems on $S_{\ub,u}$. Together with Sobolev inequalities, this will serve as the basic tools for us to control the $L^\infty$ norms. We the refer to \cite{Chr} for the proof and we shall take $ \delta>0$ to be sufficiently small as in \cite{Chr}. For the sake of simplicity, we use $S$ to denote $S_{\ub,u}$, $C$ to denote $C_u$ and $\Cb$ to denote $\Cb_{\ub}$.

We first collect all the Sobolev inequalities as follows:
\begin{equation}
\label{Sobolev_L2_L4}\|\xi\|_{L^4(S)}\lesssim|u|^{\frac{1}{2}}\|\nablas\xi\|_{L^2(S)}+|u|^{-\frac{1}{2}}\|\xi\|_{L^2(S)},
\end{equation}
\begin{equation}
\label{Sobolev_L4_Linfinity}\|\xi\|_{L^\infty(S)}\lesssim|u|^{\frac{1}{2}}\|\nablas\xi\|_{L^4(S)}+|u|^{-\frac{1}{2}}\|\xi\|_{L^4(S)},
\end{equation}
\begin{equation}
\label{SobolevC_L2_L4}\sup_{\ub}(|u|^{\frac{1}{2}}\|\xi\|_{L^4(S)})\lesssim\|\xi\|_{L^4(S)}+\|D\xi\|^{\frac{1}{2}}_{L^2(C)}(\|\xi\|^{\frac{1}{2}}_{L^2(C)}+|u|^2\|\nablas\xi\|^{\frac{1}{2}}_{L^2(C)}),
\end{equation}
\begin{equation}
\label{SobolevCb_L2_L4}\sup_u(|u|^q\|\xi\|_{L^4(S)})\lesssim|u_0|^q\|\xi\|_{L^4(S_{\ub,u_0})}+\||u|^q\Db\xi\|^{\frac{1}{2}}_{L^2(\Cb)}(\||u|^{q-1}\xi\|^{\frac{1}{2}}_{L^2(\Cb)}+\||u|^q\nablas\xi\|^{\frac{1}{2}}_{L^2(\Cb)}).
\end{equation}

For Hodge systems, we have the following two cases. First of all, if $\theta$ is a traceless symmetric two tensor, such that
$$\divs\theta=f,$$
where $f$ is a one form on $S$, thus,
\begin{align}\label{Hodge1}
\|\nablas^{k+1}\theta\|&_{L^2(S)}\lesssim\|\nablas^kf\|_{L^2(S)}+|u|^{-1}\|\nablas^{k-1}f\|_{L^2(S)}+|u|^{-1}\|\nablas^{k}\theta\|_{L^2(S)}\\\nonumber
&+|u|^{-2}\|\nablas^{k-1}\theta\|_{L^2(S)}+\sum_{i=1}^{k-1}(\|\nablas^iK\|_{L^2(S)}+|u|^{-1}\|\nablas^{i-1}K\|_{L^2(S)})\|\nablas^{k-1-i}\theta\|_{L^\infty};
\end{align}
Secondly, for $1$-form $\xi$ satisfies
\begin{equation*}
\divs\xi=f, \quad \curls\xi=g,
\end{equation*}
we have
\begin{align}\label{Hodge2}
\nonumber \|\nablas^{k+1}\xi\|_{L^2(S)}\lesssim&\|\nablas^kf\|_{L^2(S)}+|u|^{-1}\|\nablas^{k-1}f\|_{L^2(S)} +\|\nablas^kg\|_{L^2(S)}+|u|^{-1}\|\nablas^{k-1}g\|_{L^2(S)}\\
&+|u|^{-1}(\|\nablas^{k}\xi\|_{L^2(S)}+|u|^{-1}\|\nablas^{k-1}\xi\|_{L^2(S)})\\\nonumber
&+\sum_{i=1}^{k-1}(\|\nablas^iK\|_{L^2(S)}+|u|^{-1}\|\nablas^{i-1}K\|_{L^2(S)})\|\nablas^{k-1-i}\xi\|_{L^\infty}.
\end{align}

Finally, we also collect two Gronwall type estimates for later uses. They can be found in Chapter 4 of \cite{Chr}.

If we have
$$
D\theta=\frac{\nu}{2}\Omega\tr\chi\theta+\gamma\cdot\theta+\xi,
$$
with $|\gamma|\le m\Omega|\chih|$ where $\nu$ and $m$ are constants, then
\begin{align*}
\|\theta\|_{L^p(S_{\ub,u})}\lesssim_{p,\nu,m} \|\theta\|_{L^p(S_{0,u})}+\int_0^{\ub}\|\xi\|_{L^p(S_{\ub',u})}.
\end{align*}
If we have
\begin{align*}
\Db\theta=\frac{\nu}{2}\Omega\tr\chib\theta+\gamma\cdot\theta+\xi
\end{align*}
where $\theta$ is a $r$-covariant tangential tensor field and $|\gamma|\le m\Omega|\chibh|$ with some constants $\nu$ and $m$, then
\begin{align*}
|u|^{r-\nu-\frac{2}{p}}\|\theta\|_{L^p(S_{\ub,u})}\lesssim_{p,r,\nu,m} |u_0|^{r-\nu-\frac{2}{p}}\|\theta\|_{L^p(S_{\ub,u_0})}+\int_{u_0}^{u}|u'|^{r-\nu-\frac{2}{p}}\|\xi\|_{L^p(S_{\ub,u'})}.
\end{align*}

\subsection{Estimates for Connection Coefficients}

\subsubsection{Estimates for $\nablas^{k+1}\chih'$ and $\nablas^{k+1}\tr\chi'$}
We first consider $\nablas^{k+1}\chih'$ and $\nablas^{k+1}\tr\chi'$ and we will give all the details for these two terms. In the following sections, since we shall proceed in a similar way, we will be sketchy and only emphasize the key points.

 In view of the terms on the right hand side of \eqref{NSE_div_chih}, by induction hypothesis (\ref{induction}), we compute
\begin{align*}
\|\nablas^{k}(\chih'\cdot\eta)\|_{L^2(C_u)}&\lesssim\|\nablas^k\chih'\|_{L^2(C_u)}\|\eta\|_{L^\infty}+\int_0^{\ub}\sum_{i=1}^{k-1}\|\nablas^i\chih'\|_{L^4(S_{\ub',u})}\|\nablas^{k-1-i}\eta\|_{L^4(S_{\ub',u})}\\
&+\|\chih'\|_{L^\infty}\|\nablas^k\eta\|_{L^2(C_u)}\lesssim\delta^{-\frac{1}{2}}|u|^{-3-k}\cdot\delta^{\frac{1}{2}}|u|,
\end{align*}
We remark that the factor $\delta^{\frac{1}{2}}|u|$ comes from the norms $\|\cdot\|_{L^2(C_u)}$. Similarly, we have
\begin{align*}
\|\nablas^{k}(\tr\chi'\eta)\|_{L^2(C_u)}\lesssim|u|^{-3-k}\cdot\delta^{\frac{1}{2}}|u|.
\end{align*}
If one replaces $k$ by $k-1$, it is obvious that the above estimates also hold. We denote $i=-\chih'\cdot\eta+\frac{1}{2}\tr \chi'\eta-\Omega^{-1}\beta$, i.e. those terms in \eqref{NSE_div_chih}, we have
\begin{equation}\label{estimates for i 1}
\|\nablas^k i\|_{L^2(C_u)}+|u|^{-1}\|\nablas^{k-1}i\|_{L^2(C_u)}\lesssim\|\nablas^k\beta\|_{L^2(C_u)}+\delta^{-\frac{1}{2}}|u|^{-3-k}\cdot\delta^{\frac{1}{2}}|u|.
\end{equation}
In view of \eqref{Hodge1} and \eqref{Hodge2}, we turn to the estimates on Gauss curvature $K$. For $i=0,\cdots,k-1$, by (\ref{NSE_Gauss}),
\begin{align*}
\|\nablas^{i}K\|_{L^2(C_u)}&\lesssim\|\nablas^i\rho\|_{L^2(C_u)}+\|\nablas^i(\chih,\chibh)\|_{L^2(C_u)}+\|\nablas^i(\tr\chi\tr\chib)\|_{L^2(C_u)}\lesssim|u|^{-2-i}\cdot\delta^{\frac{1}{2}}|u|.
\end{align*}
Therefore,
\begin{align*}
\sum_{i=1}^{k-1}(\|\nablas^iK\|_{L^2(C_u)}+|u|^{-1}\|\nablas^{i-1}K\|_{L^2(C_u)})\|\nablas^{k-1-i}\chih'\|_{L^\infty(S)}\lesssim\delta^{-\frac{1}{2}}|u|^{-2-k}\cdot\delta^{\frac{1}{2}}|u|.
\end{align*}
According to (\ref{NSE_div_chih}), (\ref{Hodge1}) and the inequality $\int_0^{\ub}\|\cdot\|_{L^2(S_{\ub',u})}\leq\delta^{\frac{1}{2}}\|\cdot\|_{L^2(C_u)}$, we obtain
\begin{align}\label{chih'}
\int_0^{\ub}\|\nablas^{k+1}\chih'\|_{L^2(S_{\ub',u})}\lesssim\int_0^{\ub}\|\nablas^{k+1}\tr\chi'\|_{L^2(S_{\ub',u})}+\delta^{\frac{1}{2}}\|\nablas^{k}\beta\|_{L^2(C_u)}+\delta^{-\frac{1}{2}}|u|^{-2-k}\cdot\delta|u|.
\end{align}
We turn to (\ref{NSE_D_trchi}), by commuting derivatives, we have
\begin{align}\label{Dk_trchi}D\nablas^{k+1}\tr\chi'=&[D,\nablas^{k}]\ds\tr\chi'-\Omega\tr\chi\nablas^{k+1}\tr\chi'-2\Omega\chih\otimes\nablas^{k+1}\chih'\\\nonumber
&-\sum_{i=0}^{k-1}\nablas^i\ds(\Omega\tr\chi)\otimes\nablas^{k-1-i}\ds\tr\chi'-\sum_{i=0}^{k-1}\nablas^i\nablas(\Omega\chih)\otimes\nablas^{k-1-i}\nablas\chih'\\\nonumber
&-\sum_{\substack{i+j+h=k\\i,j,h\ge0}}\nablas^i\ds\log\Omega\otimes(\nablas^j\tr\chi\otimes\nablas^h\tr\chi+2\nablas^j\chih\otimes\nablas^h\chih)
\end{align}
In fact, the commutator $[D,\nablas^i]$ can be written as
$$[D,\nablas^i]\theta=\sum_{j=0}^{i-1}\nablas^j\nablas(\Omega\chi)\cdot\nablas^{i-1-j}\theta,$$
so
$$[D,\nablas^{k}]\ds\tr\chi'=\sum_{i=0}^{k-1}\nablas^i\nablas(\Omega\chi)\cdot\nablas^{k-1-i}\ds\tr\chi'.$$
Recall that $\nablas^i\ds\log\Omega=\frac{1}{2}\nablas^i(\eta+\etab)$, we rewrite (\ref{Dk_trchi}) as
\begin{align*}
D\nablas^{k+1}\tr\chi'=-\Omega\tr\chi\nablas^{k+1}\tr\chi'-2\Omega\chih\otimes\nablas^{k+1}\chih'+\lot
\end{align*}
The notation $\lot$ means the terms have already been estimated by induction hypothesis \eqref{induction}. Therefore,
\begin{align}\label{trchi'}
\|\nablas^{k+1}\tr\chi'\|_{L^2(S_{\ub,u})}\lesssim\delta^{-\frac{1}{2}}|u|^{-1}\int_0^{\ub}\|\nablas^{k+1}\chih'\|_{L^2(S_{\ub',u})}+|u|^{-3-k}\cdot|u|.
\end{align}
Combining (\ref{chih'}) and (\ref{trchi'}), we obtain
\begin{align*}
\|\nablas^{k+1}\tr\chi'\|_{L^2(S_{\ub,u})}\lesssim\delta^{-\frac{1}{2}}|u|^{-1}\int_0^{\ub}\|\nablas^{k+1}\tr\chi'\|_{L^2(S_{\ub',u})}+|u|^{-1}\|\nablas^{k}\beta\|_{L^2(C_u)}+|u|^{-3-k}\cdot|u|.
\end{align*}
Thus,
\begin{align}\label{estimate_trchi}
\|\nablas^{k+1}\tr\chi'\|_{L^2(S_{\ub,u})}\lesssim|u|^{-3-k}(\mathcal{R}_k[\beta]+C)\cdot|u|.
\end{align}
Finally, according to \eqref{estimates for i 1}, \eqref{NSE_div_chih} and (\ref{Hodge1}), we obtain
\begin{align}\label{estimate_chih}
\|\nablas^{k+1}\chih'\|_{L^2(C_u)}&\lesssim\|\nablas^{k+1}\tr\chi'\|_{L^2(C_u)}+\|\nablas^{k}\beta\|_{L^2(C_u)}+\delta^{-\frac{1}{2}}|u|^{-2-k}\cdot\delta^{\frac{1}{2}}|u|\notag\\
&\lesssim\delta^{-\frac{1}{2}}|u|^{-2-k}(\mathcal{R}_k[\beta]+C)\cdot\delta^{\frac{1}{2}}|u|.
\end{align}

\subsubsection{Smallness on $\nablas^k(\eta,\etab)$}
In order to $\nablas^{k+1}\chibh'$ and $\nablas^{k+1}\tr\chib'$, we must improve the estimates on  $\|\nablas^k(\eta,\etab)\|_{L^2(C_u)}$. We also mentioned this in the introduction. We can actually afford one more derivative to gain this smallness in powers of $\delta$. We also remark that at this stage, we can only prove that $\|\nablas^{k+1}(\eta,\etab)\|_{L^2(C_u)}$ is bounded without any smallness. We commute $\nablas^k$ with (\ref{NSE_D_eta}) to derive
\begin{align}\label{Dk_eta}
D\nablas^k\eta=\Omega\chi \nablas^k\etab+\sum_{i=0}^{k-1}(\nablas^i\nablas(\Omega\chi)\nablas^{k-1-i}\eta+\nablas^i\nablas(\Omega\chi)\otimes\nablas^{k-1-i}\etab)-\nablas^k(\Omega\beta),
\end{align}
Therefore,
\begin{align}\label{keta}
\|\nablas^k\eta\|_{L^2(S_{\ub,u})}\lesssim\delta^{-\frac{1}{2}}|u|^{-1}\int_0^{\ub}\|\nablas^k\etab\|_{L^2(S_{\ub',u})}+\delta^{\frac{1}{2}}\|\nablas^k\beta\|_{L^2(C_u)}+\delta|u|^{-3-k}\cdot|u|.
\end{align}
Similarly, according to (\ref{NSE_Db_etab}), we have
\begin{align}\label{Dbk_etab}
\Db\nablas^k\etab=\Omega\chib \nablas^k\eta+\sum_{i=0}^{k-1}(\nablas^i\nablas(\Omega\chib)\nablas^{k-1-i}\etab+\nablas^i\nablas(\Omega\chib)\otimes\nablas^{k-1-i}\eta)+\nablas^k(\Omega\betab),
\end{align}
which implies
\begin{align}\label{ketab}
|u|^k\|\nablas^k\etab\|_{L^2(C_u)}\lesssim&|u_0|^k\|\nablas^k\etab\|_{L^2(C_{u_0})}+\int_{u_0}^u|u'|^k\cdot|u'|^{-1}\|\nablas^k\eta\|_{L^2(C_{u'})}\\\nonumber
&+\int_{u_0}^u|u'|^k\|\nablas^k\betab\|_{L^2(C_{u'})}+\int_{u_0}^u|u'|^k\cdot\delta|u'|^{-4-k}\cdot\delta^{\frac{1}{2}}|u'|.
\end{align}
We then substitute (\ref{ketab}) into (\ref{keta}) and take supremum on $\ub$. As a result, we have
\begin{align*}
\sup_{\ub}|u|^k\|\nablas^k\eta\|_{L^2(S_{\ub,u})}\lesssim&|u_0|^k|u|^{-1}\|\nablas^k\etab\|_{L^2(C_{u_0})}
+\delta^{\frac{1}{2}}|u|^{-1}\int_{u_0}^u\sup_{\ub}|u'|^{k-1}\|\nablas^k\eta\|_{L^2(S_{\ub,u'})}\\
&+\delta^{\frac{3}{2}}|u|^{-4}(\mathcal{R}_k[\betab]+C)\cdot|u|+\delta^{\frac{1}{2}}|u|^{-2}(\mathcal{R}_k[\beta]+\delta^{\frac{1}{2}}|u|^{-1}C)\cdot|u|.
\end{align*}
Thus,
\begin{align*}
|u|^k\|\nablas^k\eta\|_{L^2(S_{\ub,u})}\lesssim|u_0|^k|u|^{-1}\|\nablas^k\etab\|_{L^2(C_{u_0})}+\delta^{\frac{1}{2}}|u|^{-2}(\mathcal{R}_k+C)\cdot|u|.
\end{align*}
Recall that $\eta=-\etab$ on $C_{u_0}$. In view of (\ref{keta}), we have
\begin{align*}
|u_0|^k\|\nablas^k\eta\|_{L^2(S_{\ub,u_0})}\lesssim\delta^{\frac{1}{2}}|u_0|^{-2}(\mathcal{R}_k[\beta]+\delta^{\frac{1}{2}}|u_0|^{-1}C)\cdot|u_0|.
\end{align*}
Combining all the estimates above, we finally obtain
\begin{align*}
\|\nablas^k\eta\|_{L^2(S_{\ub,u})},\ \delta^{-\frac{1}{2}}\|\nablas^k\etab\|_{L^2(C_u)}\lesssim\delta^{\frac{1}{2}}|u|^{-2-k}(\mathcal{R}_k+C)\cdot|u|.
\end{align*}

\subsubsection{Estimates for $\nablas^{k+1}\chibh'$ and $\nablas^{k+1}\tr\chib'$}
The estimates will be derived in a similar way as for $\nablas^{k+1}\chih'$ and $\nablas^{k+1}\tr\chi'$. By (\ref{NSE_div_chibh}) and the above improved smallness on $\nablas^k\etab$, we have
\begin{align}\label{chibh'}
\|\nablas^{k+1}\chibh'\|_{L^2(C_u)}\lesssim\|\nablas^{k+1}\tr\chib'\|_{L^2(C_u)}+\|\nablas^{k}\betab\|_{L^2(C_u)}+\delta^{\frac{1}{2}}|u|^{-3-k}(\mathcal{R}_k+C)\cdot\delta^{\frac{1}{2}}|u|.
\end{align}
We then commute $\nablas^{k+1}$ with \eqref{NSE_Db_trchib} to derive
\begin{align*}\Db\nablas^{k+1}\tr\chib'=-\Omega\tr\chib\nablas^{k+1}\tr\chib'-2\Omega\chibh\otimes\nablas^{k+1}\chibh'+(\tr\chi)^2\nablas^{k+1}\log\Omega+\lot
\end{align*}
 where $\lot$ denotes the terms appeared in \eqref{induction} (we shall always understand $\lot$ in this way).

 Although  $\nablas^{k+1}\log\Omega=\frac{1}{2}\nablas^k(\eta+\etab)$ is a lower order term (in derivative), this naive estimate will not yield the expected smallness of $\nablas^{k+1}\tr\chib'$. To remedy this defect, we introduce an auxiliary bootstrap assumption as follows,
\begin{align}\label{bootstrap_Omega}
\textbf{Auxiliary Bootstrap Assumption:} \quad \|\nablas^{k+1}\log\Omega\|_{L^2(C_u)}\leq\delta|u|^{-3-k}\Delta_k\cdot\delta^{\frac{1}{2}}|u|
\end{align}
for some large constant $\Delta_k$ depending on $\mathcal{R}_k$. Therefore,
\begin{align*}
|u|^{k+2}\|\nablas^{k+1}\tr\chib'\|&_{L^2(C_u)}\lesssim|u_0|^{k+2}\|\nablas^{k+1}\tr\chib'\|_{L^2(C_{u_0})}+\int_{u_0}^u|u'|^{k+2}\cdot\delta^{\frac{1}{2}}|u'|^{-2}\|\nablas^{k+1}\tr\chib'\|_{L^2(C_{u'})}\\\nonumber
&+\int_{u_0}^u|u'|^{k+2}\cdot\delta^{\frac{1}{2}}|u'|^{-2}\|\nablas^{k}\betab\|_{L^2(C_{u'})}+\delta|u|^{-2}(\mathcal{R}_k+\Delta_k+C)\cdot\delta^{\frac{1}{2}}|u|.
\end{align*}
According to \eqref{NSE_D_trchib}, we have $|u_0|^{k+1}\|\nablas^{k+1}\tr\chib'\|_{L^2(C_{u_0})}\lesssim\delta|u_0|^{-3}C\cdot\delta^{\frac{1}{2}}|u_0|$, so
\begin{align}\label{estimate_trchib}
\|\nablas^{k+1}\tr\chib'\|_{L^2(C_u)}\lesssim\delta|u|^{-4-k}(\mathcal{R}_k+\Delta_k+C)\cdot\delta^{\frac{1}{2}}|u|.
\end{align}
Thanks to (\ref{chibh'}), we obtain
\begin{align}\label{estimate_chibh}
\|\nablas^{k+1}\chibh'\|_{L^2(C_u)}\lesssim\delta^{\frac{1}{2}}|u|^{-3-k}(\mathcal{R}_k+C+\delta^{\frac{1}{2}}|u|^{-1}\Delta_k)\cdot\delta^{\frac{1}{2}}|u|.
\end{align}

\subsubsection{Estimates for $\nablas^{k+1}\eta$ and $\nablas^{k+1}\etab$} The estimates in this section rely on the bootstrap assumption (\ref{bootstrap_Omega}) and the estimates established in previous sections. In view of \eqref{NSE_curl_eta} and \eqref{NSE_curl_etab}, we have two Hodge systems:
\begin{equation}\label{Hodge_eta}
\divs\eta=-\rho+\frac{1}{2}(\chih,\chibh)-\mu, \quad \curls\eta=\sigma-\frac{1}{2}\chih\wedge\chibh
\end{equation}
\begin{equation}\label{Hodge_etab}
\divs\etab=-\rho+\frac{1}{2}(\chih,\chibh)-\mub,\quad \curls\etab=-\sigma+\frac{1}{2}\chih\wedge\chibh,
\end{equation}
where $\mu$, $\mub$ are the mass aspect functions (see \cite{Chr} for definitions). To justify the names, we notice that the integrals of $\mu$ or $\mub$ on $S_{\ub,u}$ are the Hawking masses of $S_{\ub,u}$. In reality, we just take the above equations as definitions of $\mu$ and $\mub$. They satisfy the following propagation equations:
\begin{align}\label{D_mu}
D\mu&=-\Omega\tr\chi\mu-\frac{1}{2}\Omega\tr\chi\mub-\frac{1}{4}\Omega\tr\chib|\chih|^2+\frac{1}{2}\Omega\tr\chi|\etab|^2+\divs(2\Omega\chih\cdot\eta-\Omega\tr\chi\etab)\\
\label{Db_mub}
\Db\mub&=-\Omega\tr\chib\mub-\frac{1}{2}\Omega\tr\chib\mu-\frac{1}{4}\Omega\tr\chi|\chibh|^2+\frac{1}{2}\Omega\tr\chib|\eta|^2+\divs(2\Omega\chibh\cdot\etab-\Omega\tr\chib\eta).
\end{align}
We commute $\nablas^k$ with \eqref{D_mu} to derive
\begin{align*}
D\nablas^k\mu=&-\Omega\tr\chi\nablas^k\mu-\frac{1}{2}\Omega\tr\chi\nablas^k\mub+2\nablas^k\divs(\Omega\chih)\cdot\eta+2\Omega\chih\cdot\nablas^k\divs\eta\\
&-\nablas^k\divs(\Omega\tr\chi)\cdot\etab-\Omega\tr\chi\nablas^k\divs\etab +\lot.
\end{align*}
Since we can control $\nablas^k\divs(\Omega\chih)$ and $\nablas^k\divs(\Omega\tr\chi)$ by (\ref{estimate_trchi}) and (\ref{estimate_chih}), and we can control $\nablas^k\divs\eta$ and $\nablas^k\divs\etab$ by (\ref{Hodge2}) and (\ref{Hodge_eta}), we have
\begin{align}\label{mu}
\|\nablas^k\mu\|_{L^2(S_{\ub,u})}\lesssim&|u|^{-1}\int_0^{\ub}\|\nablas^k\mub\|_{L^2(S_{\ub',u})}+|u|^{-3-k}(C+\delta^{\frac{1}{2}}|u|^{-1}\mathcal{R}_k)\cdot|u|.
\end{align}
Similarly, we commute $\nablas^k$ with \eqref{Db_mub} to derive
\begin{align*}
\Db\nablas^k\mub=&-\Omega\tr\chib\nablas^k\mub-\frac{1}{2}\Omega\tr\chib\nablas^k\mu+2\nablas^k\divs(\Omega\chibh)\cdot\etab+2\Omega\chibh\cdot\nablas^k\divs\etab\\
&-\nablas^k\divs(\Omega\tr\chib)\cdot\eta-\Omega\tr\chib\nablas^k\divs\eta+\lot.
\end{align*}
In view of \eqref{Hodge2}, \eqref{estimate_trchib}, \eqref{estimate_chibh} and \eqref{Hodge_etab}, we obtain
\begin{align}\label{mub}
|u|^{k+1}\|\nablas^k\mub\|_{L^2(C_u)}\lesssim&|u_0|^{k+1}\|\nablas^k\mub\|_{L^2(C_{u_0})}+\int_{u_0}^u|u'|^{k}\|\nablas^k\mu\|_{L^2(C_{u'})}\\\nonumber
&+|u|^{-2}(C+\mathcal{R}_k+\delta^{\frac{3}{2}}|u|^{-2}\Delta_k)\cdot\delta^{\frac{1}{2}}|u|.
\end{align}
Since $\mu=-2\rho + (\chih, \chibh)-\mub$ on $C_{u_0}$, \eqref{mu} yields
\begin{align*}
\|\nablas^k\mu\|_{L^2(S_{\ub,u_0})}\lesssim|u_0|^{-3-k}(C+\delta^{\frac{1}{2}}|u_0|^{-1}\mathcal{R}_k)\cdot|u_0|.
\end{align*}
Together with \eqref{mu} and \eqref{mub}, we have
\begin{align*}
\|\nablas^k\mu\|_{L^2(S_{\ub,u})}\lesssim&|u|^{-3-k}(C+\delta^{\frac{1}{2}}|u|^{-1}\mathcal{R}_k+\delta^{\frac{5}{2}}|u|^{-3}\Delta_k)\cdot|u|\\
\|\nablas^k\mub\|_{L^2(C_u)}\lesssim&|u|^{-3-k}(C+\mathcal{R}_k+\delta^{\frac{3}{2}}|u|^{-2}\Delta_k)\cdot\delta^{\frac{1}{2}}|u|.
\end{align*}
Once again, we use \eqref{Hodge2}, \eqref{Hodge_eta} and \eqref{Hodge_etab} to derive the final estimates
\begin{align}
\label{estimate_eta}\|\nablas^{k+1}\eta\|_{L^2(C_u)}\lesssim&|u|^{-3-k}(C+\mathcal{R}_k+\delta^{\frac{5}{2}}|u|^{-3}\Delta_k)\cdot\delta^{\frac{1}{2}}|u|,\\
\label{estimate_etab}\|\nablas^{k+1}\etab\|_{L^2(C_u)}\lesssim&|u|^{-3-k}(C+\mathcal{R}_k+\delta^{\frac{3}{2}}|u|^{-2}\Delta_k)\cdot\delta^{\frac{1}{2}}|u|.
\end{align}

\subsubsection{Estimates for $\nablas^{k+1}\omegab$ and $\nablas^{k+1}\omega$}

Those are the last two connection coefficients. We will first establish the estimates for $\nablas^{k+1}\omegab$ and close the auxiliary bootstrap argument based on (\ref{bootstrap_Omega}) in this section. We first introduce two auxiliary quantities $\kappab$ and $\kappa$ as follows (see \cite{Chr})
\begin{equation}
\Deltas\omegab =\kappab+\divs(\Omega\betab), \quad \label{Hodge_omega}\Deltas\omega=\kappa-\divs(\Omega\beta).
\end{equation}
They satisfy the following propagation equations:
\begin{align}
\label{D_kappab}D\kappab+\Omega\tr\chi\kappab&=-2(\Omega\chih,\nablas^2\omegab)+\underline{m},\\
\label{Db_kappa}\Db\kappa+\Omega\tr\chib\kappa&=-2(\Omega\chibh,\nablas^2\omega)+m.
\end{align}
where
\begin{align*}
\underline{m}=&-2(\divs(\Omega\chih),\ds\omegab)+\frac{1}{2}\divs(\Omega\tr\chi\cdot\Omega\betab)-(\ds(\Omega^2),\ds\rho)-(\ds(\Omega^2),{}^*\ds\sigma)-\rho\Deltas(\Omega^2)\\
&+\Deltas(\Omega^2(2(\eta,\etab)-|\eta|^2))+\divs(\Omega^2(\chih\cdot\betab-2\chibh\cdot\beta+3\etab\rho-3{}^*\etab\sigma)),
\end{align*}
and
\begin{align*}
m=&-2(\divs(\Omega\chibh),\ds\omega)-\frac{1}{2}\divs(\Omega\tr\chib\cdot\Omega\beta)-(\ds(\Omega^2),\ds\rho)+(\ds(\Omega^2),{}^*\ds\sigma)-\rho\Deltas(\Omega^2)\\
&+\Deltas(\Omega^2(2(\eta,\etab)-|\etab|^2))-\divs(\Omega^2(\chibh\cdot\beta-2\chih\cdot\betab-3\eta\rho-3{}^*\eta\sigma)).
\end{align*}
We commute $\nablas^{k-1}$ with \eqref{D_kappab} to derive
\begin{align*}
D\nablas^{k-1}\kappab+\Omega\tr\chi\nablas^{k-1}\kappab=-2\Omega\chih\cdot\nablas^{k+1}\omegab+\nablas^{k-1}\underline{m}+\lot.
\end{align*}
For $\nablas^{k+1}\omegab$, it is controlled by the definition of $\kappab$ and elliptic estimates \eqref{Hodge2}. For $\nablas^{k+1}\underline{m}$, the highest order terms are $\nablas^{k+1}\underline{m}$ are $\nablas^{k}\beta$, $\nablas^{k}\betab$, $\nablas^{k}\rho$, $\nablas^{k}\sigma$ and $\nablas^{k+1}\eta$, $\nablas^{k+1}\etab$. They can be controlled either by \eqref{induction} , \eqref{estimate_eta} and \eqref{estimate_etab}. Therefore
\begin{equation*}
\|\nablas^{k-1}\kappab\|_{L^2(S_{\ub,u})}\lesssim \delta|u|^{-4-k}(C+\mathcal{R}_k+\delta^2|u|^{-3}\Delta_k)\cdot|u|.
\end{equation*}
By the definition of $\kappab$ and standard elliptic estimates \eqref{Hodge2}, we obtain
\begin{align}\label{estimate_omegab}
\|\nablas^{k+1}\omegab\|_{L^2(C_u)}\lesssim\delta|u|^{-4-k}(C+\mathcal{R}_k+\delta^2|u|^{-3}\Delta_k)\cdot\delta^{\frac{1}{2}}|u|.
\end{align}

We now close the bootstrap assumption (\ref{bootstrap_Omega}) using this bound. Recall that $\omegab=\Db\log\Omega$, by commuting with $\nablas^{k+1}$, we have
\begin{equation*}
\Db\nablas^{k+1}\log\Omega=\nablas^{k+1}\omegab+\sum_{i=0}^{k-1}\nablas^i\nablas(\Omega\chib)\cdot\nablas^{k-1-i}\ds\log\Omega.
\end{equation*}
By integrating, we derive
\begin{align*}
|u|^k\|\nablas^{k+1}\log\Omega\|_{L^2(C_u)}\leq
C'\delta|u|^{-3}(C+\mathcal{R}_k+\delta^2|u|^{-3}\Delta_k)\cdot\delta^{\frac{1}{2}}|u|.
\end{align*}
We then $\Delta_k=4C'(C+\mathcal{R}_k)$ in \eqref{bootstrap_Omega} and choose $\delta$ sufficiently small such that $4C'\delta^2\le1$, 
then $\|\nablas^{k+1}\log\Omega\|_{L^2(C_u)}\leq\frac{1}{2}\delta|u|^{-3-k}\Delta_k\cdot\delta^{\frac{1}{2}}|u|$. Hence, the bootstrap argument is closed. By the choice of $\Delta_k$, we rewrite \eqref{estimate_trchib}, \eqref{estimate_chibh}, \eqref{estimate_eta}, \eqref{estimate_etab} and \eqref{estimate_omegab} as
\begin{align}\label{estimates_connection}
\delta^{-1}|u|^3\|\nabla^{k+1}\tr\chib'\|_{L^2(C_u)}+\delta^{-\frac{1}{2}}|u|^2\|\nablas^{k+1}\chibh'\|_{L^2(C_u)}+
|u|^2\|\nablas^{k+1}\eta\|_{L^2(C_u)}\\+\nonumber|u|^2\|\nablas^{k+1}\etab\|_{L^2(C_u)}+
\delta^{-1}|u|^3\|\nablas^{k+1}\omegab\|_{L^2(C_u)}\lesssim|u|^{-k-1}(C+\mathcal{R}_k)\cdot\delta^{\frac{1}{2}}|u|.
\end{align}

Similarly (we now have \eqref{estimates_connection} at our disposal), we also obtain the estimates for $\nablas^{k+1}\omega$:
\begin{align*}
\|\nablas^{k+1}\omega\|_{L^2(C_u)}\lesssim\delta^{-\frac{1}{2}}|u|^{-2-k}(C+\mathcal{R}_k)\cdot\delta^{\frac{1}{2}}|u|.
\end{align*}

\subsection{Estimates for Deformation tensors}
Following \cite{Chr}, we define the rotation vector fields $O_i$ ($i=1,2,3$) to satisfy $DO_i=0$ on $C_{u_0}$ and $\Db O_i=0$ on $M$, and $O_i|_{S_{0,u_0}}$ the standard rotation vector field on the round $\mathbb{S}^2$. For the sake of simplicity, we will suppress the lower index $i$ in sequel. The deformation tensor is defined by $\prescript{(O)}{}\pi=\mathcal{L}_Og$. We also define its null components $\piOs_{AB}=\prescript{(O)}{}\pi_{AB}$, and $Z^A=\Omega\gs^{AB}\prescript{(O)}{}\pi_{4B}$. In particular, $Z=0$ on $C_{u_0}$ by construction. For tensor field $\phi$, $\LieO\phi$ is defined as the projection of $\mathcal{L}_O\phi$ to $S_{\ub,u}$. By definition, on $C_{u_0}$,
\begin{equation*}
D\tr\piOs =2O(\Omega\tr\chi),\  D\piOsh-\Omega\tr\chi\piOsh=-\Omega\chih\tr\piOs+2\LieO(\Omega\chih).
\end{equation*}
We also have on $\DD$:
\begin{align*}
\Db\tr\piOs=2O(\Omega\tr\chib),&\ \Db Z=-4\LieO(\Omega^2\zeta), \Db\piOsh-\Omega\tr\chib\piOsh=-\Omega\chibh\tr\piOs+2\LieO(\Omega\chibh).
\end{align*}
The purpose of the section is to prove the following estimates :
\begin{align}
\label{estimates_O}&\|\nablas^{k+1}O\|\lesssim|u|^{-k}(C+\delta^{\frac{1}{2}}|u|^{-2}\mathcal{R}_k)\cdot\delta^{\frac{1}{2}}|u|,\\
\label{estimates_deformation}\delta^{-1}|u|^2&\|\nablas^k\tr\piOs\|+\delta^{-\frac{1}{2}}|u|\|\nablas^k\piOsh\|+|u|\|\nablas^kZ\|\lesssim|u|^{-k}(C+\mathcal{R}_k)\cdot\delta^{\frac{1}{2}}|u|.
\end{align}
where all the norms are take with respect to $\|\cdot\|_{L^2(C_u)}$.

We first prove by induction that, for $i\le k-1$, the following quantities
\begin{align*}
\|\nablas^{i+1}O\|_{L^2(C_u)},\delta^{-1}|u|^2&\|\nablas^i\tr\piOs\|_{L^2(C_u)},
\delta^{-\frac{1}{2}}|u|\|\nablas^i\piOsh\|_{L^2(C_u)},|u|\|\nablas^iZ\|_{L^2(C_u)}
\end{align*}
are bounded by $|u|^{-i}C\cdot\delta^{\frac{1}{2}}|u|$ where $C=C(F_{k-1+N})$. The case for $i=0,1,2$ have been proved in \cite{Chr}. We assume the case for all $j\le i-1$ is true. We commute $\nablas^i$ with the above equations to derive (the last three equations only hold on $C_{u_0}$)
\begin{align}
\label{Dbk+1_O}\Db\nablas^{i+1}O=&\nablas^{i+1}(\Omega\chib)\cdot O+\lot,\\
\label{Dbk_tr_piOs}\Db\nablas^i\tr\piOs=&2\nablas^i\LieO(\Omega\tr\chib)+\lot,\\
\label{Dbk_piOsh}\Db\nablas^i\piOsh-\Omega\tr\chib\nablas^i\piOsh=&-\Omega\chibh\otimes\nablas^i\tr\piOs+2\nablas^i\LieO(\Omega\chibh)+\lot,\\
\label{Dbk_Z}\Db\nablas^iZ=&-4\nablas^i\LieO(\Omega^2\zeta)+\lot.\\
\label{Dk+1_O}D\nablas^{i+1}O=&\nablas^{i+1}(\Omega\chi)\cdot O+\lot,,\\
\label{Dk_tr_piOs}D\nablas^i\tr\piOs=&2\nablas^i\LieO(\Omega\tr\chi)+\lot,\\
\label{Dk_piOsh}D\nablas^i\piOsh-\Omega\tr\chi\nablas^i\piOsh=&-\Omega\chih\otimes\nablas^i\tr\piOs+2\nablas^i\LieO(\Omega\chih)+\lot.
\end{align}
In the current situation, $\lot$ consists of the terms containing at most $(i-1)$th order derivatives of deformation tensors which are bounded by induction hypothesis, and the $i$th order derivatives of connection coefficients which are bounded by \eqref{induction}. Hence, by \eqref{Dbk+1_O} and \eqref{Dk+1_O}, we conclude $\|\nablas^{i+1}O\|_{L^2(C_u)}\lesssim|u|^{-i}C\cdot\delta^{\frac{1}{2}}|u|$.

For $\phi=\Omega\tr\chib,\Omega\tr\chi,\Omega\chibh,\Omega\chih,\Omega^2\zeta$, we have
\begin{align*}
\nablas^i\LieO\phi=\LieO\nablas^i\phi+\sum_{j=0}^{i-1}\nablas^j\nablas\piOs\cdot\nablas^{i-1-j}\phi=\LieO\nablas^i\phi+\nablas^{i+1}O\cdot\phi+\lot
\end{align*}
Since $|\LieO\nablas^i\phi|\lesssim|O||\nablas^{i+1}\phi|+|\nablas O||\nablas^i\phi|$, combined with \eqref{Dbk_tr_piOs}, \eqref{Dk_tr_piOs}, \eqref{Dbk_Z}, the induction hypothesis \eqref{induction} and the estimates for $\nablas^{k+1}O$, we conclude $$\delta^{-1}|u|^2\|\nablas^i\tr\piOs\|_{L^2(C_u)},|u|\|\nablas^iZ\|_{L^2(C_u)}\lesssim|u|^{-i}C\cdot\delta^{\frac{1}{2}}|u|.$$
Combined with \eqref{Dbk_piOsh}, \eqref{Dk_piOsh}, we also conclude
$$\delta^{-\frac{1}{2}}|u|\|\nablas^i\piOsh\|_{L^2(C_u)}\lesssim|u|^{-i}C\cdot\delta^{\frac{1}{2}}|u|.$$

For the case when $i=k$, we can proceed in a same way as to the case when $i\le k-1$. The only difference is that the terms $\phi=\Omega\tr\chib,\Omega\tr\chi,\Omega\chibh,\Omega\chih,\Omega^2\zeta$, $\nablas^{k+1}\phi$ are not bounded by \eqref{induction} but \eqref{estimate_trchi}, \eqref{estimate_chih}, \eqref{estimates_connection}. Thus, we have completed the proof of \eqref{estimates_O} and \eqref{estimates_deformation}.

\subsection{Estimates for Curvature Components}
In this section, based on the induction hypothesis \eqref{induction}, we derive energy estimates for curvature components. Together with the estimates from previous sections, this will complete the induction argument. Usually, the derivation of energy estimates is based on the Bel-Robinson tensors, see \cite{Ch-K} or \cite{Chr}. We shall proceed directly by integration by parts without Bel-Robinson tensors. This is similar to \cite{Luk}.

\subsubsection{Energy Inequalities}
We commute $\LieO^k$ with (\ref{NBE_D_beta}), $\cdots$, (\ref{NBE_Db_sigma}) and commute $\LieOh^k$ with (\ref{NBE_Db_alpha}) and (\ref{NBE_D_alphab}). From \eqref{NBE_Db_alpha} and \eqref{NBE_Db_beta}, we have
\begin{align*}
\Dbh\LieOh^k\alpha-\frac{1}{2}\Omega\tr\chib\LieOh^k\alpha-\Omega\nablas\tensor\LieO^k\beta&=E_3^k(\alpha),\\
D\LieO^k\beta-\Omega\chih\cdot\LieO^k\beta-\Omega\divs\LieOh^k\alpha&=E_4^k(\beta).
\end{align*}
where $E_3^k(\alpha)$ and $E_4^k(\beta)$ are error terms which will be expressed explicitly later. Therefore,
\begin{align*}
\Db(|\LieOh^k\alpha|^2d\mu_{\gs})+D(2|\LieO^k&\beta|^2d\mu_{\gs}) =
(4\divs(\Omega\LieOh^k\alpha\cdot\LieO^k\beta)-4\Omega\LieOh^k\alpha(\ds\log\Omega,\LieO^k\beta))d\mu_{\gs}\\&+((\LieOh^k\alpha,2E_3^k(\alpha))+2(\LieO^k\beta,2E_4^k(\beta)))d\mu_{\gs}.
\end{align*}
We can integrate this identity on $M$ to derive
\begin{align}\label{alpha}
\int_{C_{u}}|\LieOh^k\alpha|^2+&\int_{\Cb_{\ub}}2|\LieO^k\beta|^2\le\int_{C_{u_0}}|\LieOh^k\alpha|^2\\\nonumber&+\iint_{M}|-4\Omega\LieOh^k\alpha(\ds\log\Omega,\LieO^k\beta)+(\LieOh^k\alpha,2E_3^k(\alpha))+2(\LieO^k\beta,2E_4^k(\beta))|.
\end{align}
Similarly, from \eqref{NBE_Db_beta}, \eqref{NBE_D_rho} and \eqref{NBE_D_sigma}, we have
\begin{align*}
\Db\LieO^k\beta-\Omega\chibh\cdot\LieO^k\beta-\Omega(\ds\LieO^k\rho+\prescript{*}{}\ds\LieO^k\sigma)&=E_3^k(\beta),\\
D\LieO^k\rho+\frac{1}{2}\Omega\tr\chi\LieO^k\rho-\Omega\divs\LieO^k\beta&=E_4^k(\rho),\\
D\LieO^k\sigma+\frac{1}{2}\Omega\tr\chi\LieO^k\sigma+\Omega\curls\LieO^k\beta&=E_4^k(\sigma).
\end{align*}
which implies
\begin{align}\label{beta}
\int_{C_u}|u|^2|\LieO^k\beta|^2+&\int_{\Cb_{\ub}}|u|^2(|\LieO^k\rho|^2+|\LieO^k\sigma|^2)\leq\int_{C_{u_0}}|u_0|^2|\LieO^k\beta|^2\\\nonumber&+\iint_M|u|^2(-\frac{2}{|u|}|\LieOh^k\beta|^2+2\Omega(\ds\log\Omega,\LieO^k\rho\LieO^k\beta+\LieO^k\sigma\prescript{*}{}\LieO^k\beta)\\\nonumber
&+(\LieO^k\beta,2E_3^k(\beta))+(\LieO^k\rho,2E_4^k(\rho))+(\LieO^k\sigma,2E_4^k(\sigma))).
\end{align}
We remark that $|u|^2$ appears as a weight. We deal the remaining equations in a similar way. In fact, we have
\begin{align*}
\Db\LieO^k\rho+\frac{1}{2}\Omega\tr\chib\LieO^k\rho+\Omega\divs\LieO^k\betab&=E_3^k(\rho),\\
\Db\LieO^k\sigma+\frac{1}{2}\Omega\tr\chib\LieO^k\sigma+\Omega\curls\LieO^k\betab&=E_3^k(\sigma),\\
D\LieO^k\betab-\Omega\chih\cdot\LieO^k\betab+\Omega(\ds\LieO^k\rho-\prescript{*}{}\ds\LieO^k\sigma)&=E_4^k(\betab),
\end{align*}
and
\begin{align*}
\Db\LieO^k\betab-\Omega\chibh\cdot\LieO^k\betab+\Omega\divs\LieOh^k\alphab&=E_3^k(\betab),\\
\Dh\LieOh^k\alphab-\frac{1}{2}\Omega\tr\chi\LieOh^k\alphab+\Omega\nablas\tensor\LieO^k\betab&=E_4^k(\alphab).
\end{align*}
Therefore, we have
\begin{align}\label{rhosigma}
\int_{C_u}|u|^4&(|\LieO^k\rho|^2+|\LieO^k\sigma|^2)+\int_{\Cb_{\ub}}|u|^4|\LieO^k\betab|\leq\int_{C_{u_0}}|u_0|^4(|\LieO^k\rho|^2+|\LieO^k\sigma|^2)\\\nonumber
&+\iint_M|u|^4|-\frac{4}{|u|}(|\LieO^k\rho|^2+|\LieO^k\sigma|^2)+2\Omega(\ds\log\Omega,\LieO^k\rho\LieO^k\betab-\LieO^k\sigma\prescript{*}{}\LieO^k\betab)\\\nonumber
&+(\LieO^k\rho,2E_3^k(\rho))+(\LieO^k\sigma,2E_3^k(\sigma))+(\LieO^k\betab,2E_4^k(\beta))|.
\end{align}
and
\begin{align}\label{betabalphab}
&\int_{C_u}2|u|^6|\LieO^k\betab|^2+\int_{\Cb_{\ub}}|u|^6|\LieO^k\alphab|^2\leq\int_{C_{u_0}}2|u_0|^6|\LieO^k\betab|^2\\\nonumber
&+\iint_M|u|^6|-\frac{12}{|u|}|\LieO^k\betab|^2+4\Omega\LieOh^k\alphab(\ds\log\Omega,\LieO^k\betab)+2(\LieO^k\betab,2E_4^k(\betab))+(\LieOh^k\alpha,2E_3^k(\alphab))|.
\end{align}

\subsubsection{Estimates for Error Terms}
We turn to the bound of the error terms appeared in the previous sections. We first recall various commutator formulas in \cite{Chr} and they will be used to derive the exact expression of the error terms.
For a traceless symmetric two tensor $\theta$, we have
\begin{align*}
[\LieO,\divs]\theta_A&=\frac{1}{2}\piOsh^{BC}\nablas_A\theta_{BC}-\divs_B(\piOsh^{BC}\theta_{AC})-\frac{1}{2}\tr\piOs(\divs\theta)_A,\\
[\LieOh,\Dh]\theta&=\widehat{\Lie}_{Z}\theta-\piOsh(\Omega\chih,\theta)+\Omega\chih(\piOsh,\theta),\\
[\LieOh,\Dbh]\theta&=-\piOsh(\Omega\chibh,\theta)+\Omega\chibh(\piOsh,\theta).
\end{align*}
For an one form $\xi$, we have
\begin{align*}
[\LieO,\divs]\xi&=-\divs(\piOsh\cdot\xi)-\frac{1}{2}\tr\piOs\divs\xi,\quad [\LieO,\curls]\xi=-\frac{1}{2}\tr\piOs\curls\xi,\\
[\LieOh, \nablas\tensor]\xi&=-2\nablas\piOsh\cdot\xi-\piOsh\divs\xi-\frac{1}{2}(\piOsh,\nablas\tensor\xi),\\
[\LieO, D]\xi&=\Lie_{Z}\xi,\quad [\LieO,\Db]\xi=0.
\end{align*}
For a function $\phi$, we have
\begin{align*}
[\LieO,\ds]\phi&=0, \quad [\LieO, \,{}^*\ds]\phi =\piOsh\cdot{}^*\ds(O_i\phi),\\
[\LieO, D]\phi&=Z\phi, \quad [\LieO,\Db]\phi=0.
\end{align*}

We turn to the estimates for $E_3^k(\alpha)$ and $E_4^k(\beta)$. The above formulas allow us to write
\begin{equation*}
E_3^k(\alpha)=\frac{1}{2}\sum_{i=0}^{k-1}\LieO^i\LieO(\Omega\tr\chib)\LieOh^{k-1-i}\alpha+\LieOh^kE_3^0(\alpha)+[\Dbh,\LieOh^k]\alpha-(\Omega\nablas\tensor\LieO^k\beta-\LieOh^k(\Omega\nablas\tensor\beta)),
\end{equation*}
where $E_3^0(\alpha)$ is from \eqref{NBE_Db_alpha}. For a traceless symmetric two tensor $\theta$, we have
\begin{align*}
\LieOh^l\theta=\sum_{i_1+\cdots+i_l+i_{l+1}=l}\nablas^{i_1}O\cdot\nablas^{i_2}O\cdots\nablas^{i_l}O\cdot\nablas^{i_{l+1}}\theta.
\end{align*}
Then first term in $E_3^k(\alpha)$ can be written as
\begin{align*}
\sum_{i_1+\cdots+i_{k+2}=k-1}\nablas^{i_1}O\cdot\nablas^{i_2}O\cdots\nablas^{i_{k}}O\cdot\nablas^{i_{k+1}}\nablas(\Omega\tr\chib)\cdot\nablas^{i_{k+2}}\alpha.
\end{align*}
We bound this term using induction hypothesis. If $i_{k+1}=k-1$ or $i_{k+2}=k-1$, we bound the corresponding factors in $L^2(C_u)$ and the others in $L^\infty$; if $i_j=k-1$ for $j=1,\cdots,k$, $i_{k+1}=k-2$ or $i_{k+2}=k-2$, we bound the corresponding factors in $L^4(S_{\ub,u})$ (note that when $k=3$, the case $i_{k+1}=i_{k+2}=k-2$ can happen) and the others in $L^\infty$; for the other possible cases, we simply bound all factors in $L^\infty$.

For the second term in $E_3^k(\alpha)$, notice that $E_3^0(\alpha)\emph{}=-2\omegab\alpha-\Omega(-(4\eta+\zeta)\tensor\beta+3\chih\rho+3\prescript{*}{}\chih\sigma)$.
We use $\Gamma$ to be an arbitrary connection coefficient and $R$ to be an arbitrary curvature component.\footnote{\,Throughout the paper, we shall use $\Gamma$ and $R$ as schematic notations. The expression $\Gamma \cdot R$ denotes a sum of products of this form.} We can compute $\LieOh^kE_3^0(\alpha)$ by the commutator formulas. All the terms can be bounded exactly in the same way as for the first term $E_3^k(\alpha)$ except for
\begin{align*}
O^k\cdot(-2\omegab\nablas^k\alpha+\Omega(4\eta+\zeta)\cdot\nablas^k\beta-3\Omega\chih\cdot\nablas^k\rho-3\Omega\chih\cdot\nablas^k\sigma).
\end{align*}
This term can only be bounded by $\mathcal{R}_k$. In fact, this collection of terms contains all the terms involving $k$th derivatives of the curvature components while the other terms are in the following form
\begin{align*}
\sum_{i_1+\cdots+i_{k+2}=k,i_{k+2}<k}\nablas^{i_1}O\cdot\nablas^{i_2}O\cdots\nablas^{i_{k}}O\cdot\nablas^{i_{k+1}}\Gamma\cdot\nablas^{i_{k+2}}R.
\end{align*}
which can be easily estimated by induction hypothesis.

For the third term in $E_3^k(\alpha)$, since $[\Dbh,\LieOh^k]\alpha=\sum_{i=0}^{k-1}\LieOh^{k-1-i}([\Db,\LieOh]\LieOh^i\alpha)$, according to the commutator formulas, this term can be rewritten as
\begin{align*}
\sum_{i_1+\cdots+i_{k+2}=k-1}\nablas^{i_1}O\cdot\nablas^{i_2}O\cdots\nablas^{i_{k-1}}O\cdot\nablas^{i_{k}}\piOsh\cdot\nablas^{i_{k+1}}(\Omega\chibh)\cdot\nablas^{i_{k+2}}\alpha.
\end{align*}
The observation is that none of terms contains $k$th derivative of the curvature components. Therefore, $[\Dbh,\LieOh^k]\alpha$ can also be estimated by induction hypothesis as above.

For the last term in $E_3^k(\alpha)$, we can write it as
\begin{equation*}
\Omega(\nablas\tensor\LieO^k\beta-\LieOh^k\nablas\tensor\beta)=\Omega\sum_{i=0}^{k-1}\LieOh^{k-1-i}(\nablas\tensor\LieO\LieO^i\beta-\LieOh\nablas\tensor\LieO^i\beta)+\lot
\end{equation*}
where $\lot$ denotes the terms with at most $k$th derivatives on the curvature components. They can be estimated by induction hypothesis. The main term reads as
\begin{align*}
\sum_{i_1+\cdots+i_{k+1}=k,i_k+i_{k+1}\ge1}\nablas^{i_1}O\cdot\nablas^{i_2}O\cdots\nablas^{i_{k-1}}O\cdot\nablas^{i_{k}}\piOsh\cdot\nablas^{i_{k+1}}\beta.
\end{align*}
Thus, all terms can be estimated by induction hypothesis except for
\begin{align*}
O^{k-1}\cdot(\nablas^k\piOsh\cdot\beta+\piOsh\cdot\nablas^k\beta)
\end{align*}
which should be bound by $\mathcal{R}_k$ according to the definition of $\mathcal{R}_k$.

Putting the estimates together, we obtain
\begin{equation*}
\iint_M|\LieOh^k\alpha,2E_3^k(\alpha)|\lesssim\delta^{-1}\mathcal{R}_k[\alpha](C+\mathcal{R}_k).
\end{equation*}

We move on to $E_4^k(\beta)$ which can be written as
\begin{align*}
E_4^k(\beta)=\sum_{i=0}^{k-1}\LieO^i\LieO(\Omega\chih^\sharp)\cdot\LieO^{k-1-i}\beta+\LieOh^kE_4^0(\beta)+[D,\LieO^k]\beta-(\Omega\divs\LieOh^k\alpha-\LieOh^k(\Omega\divs\alpha)).
\end{align*}
Since the estimates can be derived almost in the same way as before, we only emphasize the difference. In face, the third term in $E_4^k(\beta)$ is slightly different from before because $[D,\LieO]\ne0$. This commutator term contains $\nablas^k Z$ which should be estimated by $\mathcal{R}_k$. To be more precise, this term can be written as
\begin{align*}
[D,\LieO^k]\beta&=\sum_{i=0}^{k-1}\LieO^{k-1-i}([D,\LieO]\LieO^i\beta)=\sum_{i=0}^{k-1}\LieO^{k-1-i}\Lie_Z\LieO^i\beta\\ &=\sum_{i_1+\cdots+i_{k+1}=k}\nablas^{i_1}O\cdots\nablas^{i_{k-1}}O\cdot\nablas^{i_k}Z\cdot\nablas^{i_{k+1}}\beta.
\end{align*}
All the terms can be bounded by induction hypothesis expect for $i_k=k$ and $i_{k+1}=k$. Finally, we obtain
\begin{align*}
\iint_M|\LieO^k\beta,2E_4^k(\beta)|\lesssim\delta^{-\frac{1}{2}}\mathcal{R}_k[\beta](C+\mathcal{R}_k+\mathcal{D}_{k+1}[\tr\chib]).
\end{align*}

We switch to the second group of error terms $E_3^k(\beta)-\frac{1}{|u|}\LieO^k\beta$ and $E_4^k(\rho,\sigma)$. We have
\begin{align*}
E_3^k(\beta)=&\sum_{i=0}^{k-1}\LieO^i\LieO(\Omega\chibh^\sharp)\cdot\LieO^{k-1-i}\beta+\LieO^kE_3^0(\beta)\\
&+[\Db,\LieO^k]\beta-(\Omega(\ds\LieO^k\rho+\prescript{*}{}\ds\LieO^k\sigma)-\LieO^k(\Omega(\ds\rho+\prescript{*}{}\ds\sigma))).
\end{align*}
It is necessary to observe that we need the term $-\frac{1}{|u|}\LieO^k\beta$ to cancel the term $-\frac{1}{2}\tr\chib\LieO^k\beta$ in $\LieO^kE_3^k(\beta)$. The derivation of the estimates is almost the same as before and finally we obtain
\begin{align*}
\iint_M|u|^2|\LieO^k\beta,2E_3^k(\beta)-\frac{2}{|u|}\LieO^k\beta|\lesssim\delta^{\frac{1}{2}}\mathcal{R}_k[\beta](C+\mathcal{R}_k).
\end{align*}

For $E_3^k(\rho,\sigma)$, we need to argue more carefully because it contains some borderline terms (meaning that there is not positive power of $\delta$ in the estimates). First of all, we have
\begin{align*}
E_4^k(\rho)&=-\frac{1}{2}\sum_{i=0}^{k-1}\LieO^i\LieO(\Omega\tr\chi)\LieO^{k-1-i}\rho+\LieO^kE_4^0(\rho)+[D,\LieO^k]\rho+(\Omega\divs\LieO^k\beta-\LieO^k(\Omega\divs\beta)),\\
E_4^k(\sigma)&=-\frac{1}{2}\sum_{i=0}^{k-1}\LieO^i\LieO(\Omega\tr\chi)\LieO^{k-1-i}\sigma+\LieO^kE_4^0(\sigma)+[D,\LieO^k]\sigma-(\Omega\curls\LieO^k\beta-\LieO^k(\Omega\curls\beta)).
\end{align*}
The terms in $[D,\LieO^k]\rho$ and $[D,\LieO^k]\sigma$ do not have the factor $\nablas^kZ$ because $\rho$ and $\sigma$ are functions. This is slightly different from previous cases and the estimates are even easier. The key different terms are those $\LieO^k(\chibh\cdot\alpha)$ contained in $\LieO^kE_4^0(\rho)$ and $\LieO^kE_4^0(\sigma)$. They can be bounded by $\mathcal{R}_k[\alpha]$ and without any gain in $\delta$. Thus, they contribute to borderline terms. Finally, we obtain
\begin{align*}
\iint_M|u|^2|(\LieO^k\rho,2E_3^k(\rho))+(\LieO^k\sigma,2E_3^k(\sigma))|\lesssim\mathcal{R}_k[\rho,\sigma](C+\mathcal{R}_k[\alpha]+\delta^{\frac{1}{2}}(\mathcal{R}_k+\mathcal{D}_{k+1}[\tr\chib])).
\end{align*}

For the remaining error terms, namely $E_3^k(\rho,\sigma)$, $E_4^k(\betab)$, $E_3^k(\betab)$ and $E_4^k(\alphab)$, they can be expressed as
\begin{align*}
E_3^k(\rho)=&-\frac{1}{2}\sum_{i=0}^{k-1}\LieO^i\LieO(\Omega\tr\chib)\LieO^{k-1-i}\rho+\LieO^kE_3^0(\rho)+[\Db,\LieO^k]\rho-(\Omega\divs\LieO^k\betab-\LieO^k(\Omega\divs\betab)),\\
E_3^k(\sigma)=&-\frac{1}{2}\sum_{i=0}^{k-1}\LieO^i\LieO(\Omega\tr\chib)\LieO^{k-1-i}\sigma+\LieO^kE_3^0(\sigma)+[\Db,\LieO^k]\sigma-(\Omega\curls\LieO^k\betab-\LieO^k(\Omega\curls\betab)),\\
E_4^k(\betab)=&\sum_{i=0}^{k-1}\LieO^i\LieO(\Omega\chih^\sharp)\cdot\LieO^{k-1-i}\beta+\LieO^kE_4^0(\betab)\\
&+[D,\LieO^k]\betab+(\Omega(\ds\LieO^k\rho-\prescript{*}{}\ds\LieO^k\sigma)-\LieO^k(\Omega(\ds\rho-\prescript{*}{}\ds\sigma))),\\
E_3^k(\betab)=&\sum_{i=0}^{k-1}\LieO^i\LieO(\Omega\chibh^\sharp)\cdot\LieO^{k-1-i}\betab+\LieOh^kE_3^0(\betab)+[\Db,\LieO^k]\betab+(\Omega\divs\LieOh^k\alphab-\LieOh^k(\Omega\divs\alphab)),\\
E_4^k(\alphab)=&\frac{1}{2}\sum_{i=0}^{k-1}\LieO^i\LieO(\Omega\tr\chi)\LieOh^{k-1-i}\alphab+\LieOh^kE_4^0(\alphab)+[\Dh,\LieOh^k]\alphab+(\Omega\nablas\tensor\LieO^k\betab-\LieOh^k(\Omega\nablas\tensor\betab)).
\end{align*}
We can proceed exactly as before to derive
\begin{align*}
&\iint_M|u|^4|(\LieO^k\rho,2E_3^k(\rho)-\frac{4}{|u|}\LieO^k\rho)+(\LieO^k\sigma,2E_3^k(\sigma)-\frac{4}{|u|}\LieO^k\sigma)| \lesssim\delta^2\mathcal{R}_k[\rho,\sigma](C+\mathcal{R}_k+\underline{\mathcal{R}}_k[\alphab]),\\
&\iint_M|u|^4|(\LieO^k\betab,2E_4^k(\betab))|\lesssim\delta^2\mathcal{R}_k[\betab](C+\mathcal{R}_k),\\
&\iint_M|u|^6|(\LieO^k\betab,2E_3^k(\betab)-\frac{6}{|u|}\LieO^k\betab)|\lesssim\delta^{\frac{7}{2}}\mathcal{R}_k[\betab](C+\mathcal{R}_k+\underline{\mathcal{R}}_k[\alphab]),\\
&\iint_M|u|^6|(\LieO^k\alphab,2E_4^k(\alphab))|\lesssim\delta^{3}\underline{\mathcal{R}}_k[\alphab](C+\mathcal{R}_k[\rho,\sigma]+\delta^{\frac{1}{2}}(\mathcal{R}_k+\underline{\mathcal{R}}_k[\alphab])).
\end{align*}

\subsection{Completion of the Induction Argument}
We first close the induction argument on the $L^2$ level. In fact, based on the estimates derived in the previous sections, we show that if $\delta$ is sufficiently small depending on the $C^{k+N}$ bounds of the seed data $\psi_0$, then
\begin{align*}
\mathcal{R}_k,\underline{\mathcal{R}}_k[\alpha], \mathcal{O}_{k+1}\le F_{k+N}
\end{align*}
where $F_{k+N}$ is a constant depending only on the $C^{k+N}$ bounds of the seed data.

Recall that for $1$-form $\xi$, $\LieO^k\xi-\nablas_O^k\xi=\sum_{i_{k+1}<k}\nablas^{i_1}O\cdots\nablas^{i_k}O\cdot\nablas^{i_{k+1}}\xi$, then
\begin{align*}
\int_{C_u}|u|^{2k}|\nablas^k\xi|^2\lesssim\int_{C_u}|\LieO^i\xi|^2+\int_{C_u}\sum_{i=0}^{k-1}|u|^{2i}|\nablas^i\xi|^2.
\end{align*}
Similar inequalities hold for traceless symmetric two tensors and functions.

We suppose that $\mathcal{R}_k,\underline{\mathcal{R}}_k[\alpha]\le G$ for some large constant $G$ to be fixed later. In view of the above estimates, we multiply $\delta^2$ to (\ref{alpha}), $\delta^0$ to (\ref{beta}), $\delta^{-1}$ to (\ref{rhosigma}) and $\delta^{-3}$ to (\ref{betabalphab}), then substitute the above estimates into (\ref{alpha})-(\ref{betabalphab}) and  we choose $\delta>0$ sufficiently small depending on $G$ and $F_{k-1+N}$ \eqref{induction}, so that we obtain the following inequalities,
\begin{align*}
\mathcal{R}_k[\alpha]^2&\leq\mathcal{R}_k[\alpha](u_0)^2+1,\\
\mathcal{R}_k[\beta]^2&\leq\mathcal{R}_k[\beta](u_0)^2+C\mathcal{R}_k[\rho,\sigma](1+\mathcal{R}_k[\alpha])+1,\\
\mathcal{R}_k[\rho,\sigma]^2&\leq\mathcal{R}_k[\rho,\sigma](u_0)^2+1,\\
\mathcal{R}_k[\betab]^2+\underline{\mathcal{R}}_k[\alphab]^2&\leq\mathcal{R}_k[\betab](u_0)^2+C\underline{\mathcal{R}}_k[\alphab](1+\mathcal{R}_k[\rho,\sigma])+1.
\end{align*}
We substitute the first and the third to the second to derive
\begin{align*}
\mathcal{R}_k[\beta]^2&\leq\mathcal{R}_k[\beta](u_0)^2+C\sqrt{\mathcal{R}_k[\rho,\sigma](u_0)^2+1}(1+\sqrt{\mathcal{R}_k[\alpha](u_0)^2+1})+1.
\end{align*}
We then substitute the third to the fourth to derive
\begin{align*}
\underline{\mathcal{R}}_k[\alphab]\le C(1+\sqrt{\mathcal{R}_k[\rho,\sigma](u_0)^2+1})+\sqrt{\mathcal{R}_k[\betab](u_0)^2+1}.
\end{align*}
Thus, we can choose $G$ depending on $\mathcal{R}_k(u_0)$ and hence on $C^{k+N}$ bounds of the seed data, such that $\mathcal{R}_k, \underline{\mathcal{R}}_k[\alpha]\le G/2$. By continuity argument and the estimates for connection coefficients in previous sections, we have obtained the estimates for $\mathcal{R}_k,\underline{\mathcal{R}}_k[\alpha], \mathcal{O}_{k+1}$.

We turn to the following estimates
\begin{align*}
\mathcal{R}^4_{\le{k-1}}, \mathcal{R}^\infty_{\le k-2}, \mathcal{O}^4_{\le k}, \mathcal{O}^\infty_{\le k-1}\le C(F_{k+N}).
\end{align*}
We first commute $\nablas^{k-1}$ with (\ref{NBE_D_beta}), (\ref{NBE_D_rho}), (\ref{NBE_D_sigma}) and (\ref{NBE_D_betab}) to derive
\begin{align*}
\delta^{-\frac{1}{2}}|u|^{-1}&\left(\|\delta^{\frac{1}{2}}|u|^2(\delta D)(|u|\nabla)^{k-1}\beta\|_{L^2(C_u)}+\||u|^3(\delta D)(|u|\nabla)^{k-1}(\rho,\sigma)\|_{L^2(C_u)}\right.\\
&\left.+\|\delta^{-1}|u|^4(\delta D)(|u|\nabla)^{k-1}\underline{\beta}\|_{L^2(C_u)}\right)\le C(F_{k+N}).
\end{align*}
We then apply $\LieOh^{k-1}\Dh$ to (\ref{NBE_Db_alpha}), $\LieO^{k-1}D$ to (\ref{NBE_Db_beta}), $\LieOh^{k-1}\Dbh$ to (\ref{NBE_D_alphab}) and $\LieO^{k-1}\Db$ to (\ref{NBE_D_betab}) to derive
\begin{align*}
\Dbh\LieOh^{k-1}\Dh\alpha-\frac{1}{2}\Omega\tr\chib\LieOh^{k-1}\Dh\alpha-\Omega\nablas\tensor\LieO^{k-1}D\beta&=E_3^{(k-1)4}(\alpha),\\
D\LieO^{k-1}D\beta-\Omega\chih\cdot\LieO^{k-1}D\beta-\Omega\divs\LieOh^{k-1}\Dh\alpha&=E_4^{(k-1)4}(\beta),\\
\Db\LieO^{k-1}\Db\betab-\Omega\chibh\cdot\LieO^{k-1}\Db\betab+\Omega\divs\LieOh^{k-1}\Dbh\alphab&=E_3^{(k-1)3}(\betab),\\
\Dh\LieOh^{k-1}\Dbh\alphab-\frac{1}{2}\Omega\tr\chi\LieOh^{k-1}\Dbh\alphab+\Omega\nablas\tensor\LieO^{k-1}\Db\betab&=E_4^{(k-1)3}(\alphab).
\end{align*}
Therefore, we have
\begin{align}\label{D_alpha}
\int_{C_u}|\LieOh^{k-1}\Dh\alpha|^2\leq&\int_{C_{u_0}}|\LieOh^{k-1}\Dh\alpha|^2+\iint_M|-4\Omega\LieOh^{k-1}\Dh\alpha(\ds\log\Omega,\LieO^{k-1}D\beta)\\\nonumber
&+(\LieOh^{k-1}\Dh\alpha,2E_3^{(k-1)4}(\alpha))+2(\LieO^{k-1}D\beta,2E_4^{(k-1)4}(\beta))|,
\end{align}
and
\begin{align}\label{Db_alphab}
&\int_{\Cb_{\ub}}|u|^8|\LieO^{k-1}\Dbh\alphab|^2\leq\int_{C_{u_0}}2|u_0|^8|\LieO^{k-1}\Db\betab|^2 + \iint_M|u|^8|\frac{16}{|u|}|\LieO^{k-1}\Db\betab|^2\\\nonumber
&+\iint_M|u|^8|4\Omega\LieOh^{k-1}\Dbh\alphab(\ds\log\Omega,\LieO^{k-1}\Db\betab)+2(\LieO^{k-1}\Db\betab,2E_4^{(k-1)3}(\betab))+(\LieOh^{k-1}\Dbh\alpha,2E_3^{(k-1)3}(\alphab))|.
\end{align}
This leads to
\begin{align*}
\delta^{-\frac{1}{2}}|u|\|\delta^{\frac{3}{2}}|u|(\delta D)(|u|\nablas^{k-1})\alpha\|_{L^2(C_u)}
+\||u|^{-\frac{3}{2}}\delta^{-\frac{3}{2}}|u|^{\frac{9}{2}}(|u|\Db)(|u|\nabla)^{k-1}\underline{\alpha}\|_{L^2(\underline{C}_{\underline{u}})}\le C(F_{k+N}).
\end{align*}
The Sobolev inequalities (\ref{Sobolev_L4_Linfinity}), (\ref{SobolevC_L2_L4}) and (\ref{SobolevCb_L2_L4}) yield $\mathcal{R}^4_k, \mathcal{R}^\infty_{k-1}\le C(F_{k+N})$.

Similarly, we commute $\nablas^k$ with (\ref{NSE_Dh_chih}), (\ref{NSE_D_trchi}), (\ref{NSE_D_eta}), (\ref{NSE_D_omegab}), (\ref{NSE_D_chibh}) and (\ref{NSE_D_trchib}) to derive
\begin{align*}
\delta^{-\frac{1}{2}}|u|^{-1}&\left(\|\delta^{\frac{1}{2}}|u|(\delta D)(|u|\nabla)^{k}\widehat{\chi}\|_{L^2(C_u)}
+\||u|^2(\delta D)(|u|\nabla)^{k}\mathrm{tr}\chi\|_{L^2(C_u)}\right.\\
&+\|\delta^{-\frac{1}{2}}|u|^2(\delta D)(|u|\nabla)^{k}\underline{\widehat{\chi}}\|_{L^2(C_u)}
+\|\delta^{-1}|u|^3(\delta D)(|u|\nabla)^{k}\mathrm{tr}\underline{\chi}\|_{L^2(C_u)}\\
&\left.+\||u|^2(\delta D)(|u|\nabla)^{k}\eta\|_{L^2(C_u)}+\|\delta^{-1}|u|^3(\delta D)(|u|\nabla)^{k}\underline{\omega}\|_{L^2(C_u)}\right)\le C(F_{k+N}).
\end{align*}
According to the relation $\eta+\etab=2\ds\log\Omega$, we can also derive
\begin{align*}
\delta^{-\frac{1}{2}}|u|^{-1}\||u|^2(\delta D)(|u|\nabla)^{k}\etab\|_{L^2(C_u)}\le C(F_{k+N}).
\end{align*}
Now (\ref{Sobolev_L4_Linfinity}) and (\ref{SobolevC_L2_L4}) yield $\mathcal{O}^4_k,\mathcal{O}^\infty_k\le C(F_{k+N})$ except for the component $\omega$.

In order to bound $\omega$, we commute $\nablas^k$ with (\ref{NSE_Db_omega}).
We also take $\nablas^k$ to both sides of the equation and we hope to bound $L^2(\Cb_{\ub})$ norm for $\Db\nablas^k\omega$. The commutator $[\Db,\nablas^k]\omega$ is easy to control thanks to the induction hypothesis \eqref{induction}. Because we do have $L^4(S_{\ub,u})$ bound for $\nablas^k\eta$ and $\nablas^k\etab$ at the moment, then we can bound $\||u|^{-\frac{3}{2}}|u|^2(|u|\nablas)^k((\eta,\etab)-|\etab|^2)\|_{L^2(\Cb_{\ub})}\le C(F_{k+N})$. In view of (\ref{beta}), we have
$\||u|^{-\frac{3}{2}}|u|^{\frac{5}{2}}(|u|\nablas)^k\rho\|\le C(F_{k+N})$ and we can deduce that
$$\||u|^{-\frac{3}{2}}|u|(|u|\Db)(|u|\nablas)^k\omega\|\le C(F_{k+N}),$$
Finally, (\ref{Sobolev_L4_Linfinity}) and (\ref{SobolevCb_L2_L4}) yield $\mathcal{O}^4_k[\omega], \mathcal{O}^\infty_{k-1}[\omega]\le C(F_{k+N})$.

We have completed the proof of Theorem \ref{Higherorder}.

\section{Construction of the Transition Region}\label{Construction of the Transition Region}
We have already mentioned that, in addition to Christodoulou's short pulse ansatz described, we need to impose one more condition on the initial data defined on $C_{u_0}^{[0,\delta]}$, that is \eqref{integral=m0} or \eqref{integral=m0 different}. This condition is the key to all the construction in sequel. We start from the condition \eqref{integral=m0} which is imposed on the seed data, that is,
\begin{equation}
\int_0^1\left|\frac{\partial\psi_0}{\partial s}(s,\theta)\right|^2\D s=16m_0.
\end{equation}
We show that this condition is equivalent to the more physical one \eqref{integral=m0 different} up to an error of size $\delta$ to some positive power. More precisely, we have
\begin{lemma}On the sphere $S_{0,\delta}$, for all $k$, we have
\begin{align*}
|u_0|^k|\nablas^k(\int_0^{\delta}|u_0|^2|\chih(\ub,u_0)|^2\D\ub-4m_0)|\lesssim_k\delta^{\frac{1}{2}}|u_0|^{-1}.
\end{align*}
\end{lemma}
The notation $A\lesssim_kB$ means that $A\le C_kB$ where $C_k$ is a constant depending on the $C^{k+N}$ norms of the seed data for sufficient large $N$.
\begin{proof}
In stereographic coordinates, we have
\begin{align*}
4\int_0^{\delta}|\chih(\ub,u_0)|^2\D\ub&=\int_0^{\delta}\widehat{\gs}^{AC}\widehat{\gs}^{BD}D\widehat{\gs}_{AB}D\widehat{\gs}_{CD}\D\ub=\int_0^{\delta}m^{AC}m^{BD}\partial_{\ub}m_{AB}\partial_{\ub}m_{CD}\D\ub.
\end{align*}
In view of the short pulse ansatz (\ref{shortpulse}), we have $|\partial^k_\theta(m^{AB}-\delta^{AB})|\lesssim_k\delta^{\frac{1}{2}}|u_0|^{-1}$. Now we can make use of the definition $m=\exp\psi$ to derive
\begin{align*}
|\partial^k_\theta(4\int_0^{\delta}|u_0|^2|\chih(\ub,u_0)|^2\D\ub-\int_0^{\delta}|u_0|^2|\partial_{\ub}\psi(\ub)|^2\D\ub)|\lesssim_k\delta^{\frac{1}{2}}|u_0|^{-1}.
\end{align*}
In view of (\ref{shortpulse}) and (\ref{integral=m0}), we have $\int_0^{\delta}|u_0|^2|\partial_{\ub}\psi|^2\D\ub=\int_0^1|\partial_s\psi_0|^2\D s=16m_0$,
therefore
\begin{align*}
|u_0|^k|\nablas^k(\int_0^{\delta}|u_0|^2|\chih(\ub,u_0)|^2\D\ub-4m_0)|\lesssim_k\delta^{\frac{1}{2}}|u_0|^{-1}.
\end{align*}
\end{proof}
Thanks to the lemma, we shall not differentiate conditions \eqref{integral=m0} and \eqref{integral=m0 different}. We turn to the geometric consequences of these conditions.

\subsection{Geometry on $S_{\delta,u_0}$.}\label{Section Geometry on intersecting sphere}
The purpose of this section is to prove the following lemma which, roughly speaking, says the geometry of the two sphere $S_{\delta,u_0}$ is close to the geometry of a given two sphere in the Schwarzschild space-time with mass $m_0$.
\begin{lemma}\label{geometryonS}On the sphere $S_{\delta,u_0}$, we have $\alpha \equiv 0$ and for all $k$
\begin{align*}
|u_0|^k\left(|\nablas^k(\tr\chi-(\frac{2}{|u_0|}-\frac{4m_0}{|u_0|^2}))|,|\nablas^k\beta|,
|u_0||\nablas^k(\rho+\frac{2m_0}{|u_0|^3})|,|u_0||\nablas^k\sigma|\right)\lesssim_k\delta^{\frac{1}{2}}|u_0|^{-2}.
\end{align*}
\end{lemma}
\begin{remark}
The connection coefficients and the curvature components appearing in the lemma have improvement on ``smallness'', which means that they are controlled by $\delta^{\frac{1}{2}}$ instead of $\delta$ to some nonpositive power in \cite{Chr}. This improvement comes from the condition \eqref{integral=m0 different}. The other connection coefficients and the curvature components, which are not mentioned in the lemma, have already been bounded by $\delta$ to some positive power in \cite{Chr}. This suggests that the geometry of the sphere $S_{\delta,u_0}$ is close to the geometry of a sphere in Schwarzschild spacetime.
\end{remark}
\begin{proof}
First of all, $\alpha=0$ on $S_{\delta,u_0}$ follows immediately from \eqref{NSE_Dh_chih} and the fact that $\chih$ has compact support in $C_{u_0}^{(0,\delta)}$.

For $\tr\chi$, on $C_{u_0}$, we can write \eqref{NSE_D_trchi} as
\begin{equation*}
D\tr\chi=-\frac{1}{2}(\tr\chi)^2-|\chih|^2,
\end{equation*}
since where $\Omega=1$. We integrate this equation along $[0,\delta]$ to derive
\begin{equation*}
\tr\chi(\delta,u_0)-\frac{2}{|u_0|}=-\frac{1}{2}\int_0^{\delta}(\tr\chi(\ub,u_0))^2\D\ub-\int_0^{\delta}|\chih(\ub,u_0)|^2\D\ub.
\end{equation*}
Since $|u_0|^k|\nablas^k(\tr\chi-\frac{2}{|u_0|})|\lesssim_k|u_0|^{-2}$ (see Chapter 2 of \cite{Chr}),  taking $\nablas^k$ on the above equation yields the desired estimates for $\tr\chi$ on the sphere $S_{\delta,u_0}$.

For $\beta$, on $C_{u_0}$, we can write \eqref{NSE_div_chih} as
\begin{equation*}
\divs \chih=\frac{1}{2}\ds \tr \chi-\chih\cdot\eta+\frac{1}{2}\tr \chi\eta-\beta.
\end{equation*}
Since $|u_0|^k|\nablas^k\eta|\lesssim_k\delta^{\frac{1}{2}}|u_0|^{-2}$ (see Chapter 2 of \cite{Chr}), $\chih=0$ on $S_{\delta,u_0}$, combined with the estimates on $\tr\chi$ just derived,  we can apply $\nablas^k$ on the above equation to derived the desired estimates for $\beta$ on $S_{\delta,u_0}$.

For $\sigma$, Ton $C_{u_0}$, we can write \eqref{NSE_curl_eta} as
\begin{align*}
\curls \eta=\sigma-\frac{1}{2}\chih \wedge\chibh.
\end{align*}
Since $|u_0|^k|\nablas^k\eta|\lesssim_k\delta^{\frac{1}{2}}|u_0|^{-2}$ and $\chih=0$ on $S_{\delta,u_0}$, we obtain the desired estimates for $\sigma$ on $S_{\delta,u_0}$ by simply applying $\nablas^k$ to the above equation.

Finally, we consider $\rho$. On $C_{u_0}$, the Gauss equation \eqref{NSE_Gauss} can be written as
\begin{align*}
K=-\frac{1}{4}\tr \chi\tr\chib+\frac{1}{2}(\chih,\chibh)-\rho.
\end{align*}
We rewrite this equation in the following renormalized form
\begin{align*}
(K-\frac{1}{|u_0|^2})+\frac{1}{4}\tr\chi(\tr\chib+\frac{2}{|u_0|})=-(\rho-\frac{1}{2|u_0|}(\tr\chi-\frac{2}{|u_0|})).
\end{align*}
Since $|u_0|^k|\nablas^k(K-\frac{1}{|u_0|^2})|\lesssim_k\delta^{\frac{1}{2}}|u_0|^{-3}$ and $|u_0|^k|\nablas^k(\tr\chib+\frac{2}{|u_0|})|\lesssim_k\delta|u_0|^{-2}$ on $C_{u_0}$ (see again Chapter 2 of \cite{Chr}), combined with the estimates for $\tr\chi$ derived above, we obtain the desired estimates for $\rho$ by applying $\nablas^k$ to the above equation. This completes the proof of the lemma.
\end{proof}

\subsection{Geometry on $\underline{C}_{\delta}$.}\label{Section Geometry on incoming cone}

We shall study the geometry of $\underline{C}_{\delta}$ by integrating the equations \eqref{NSE_Db_chih}, \eqref{NSE_Db_trchi}, \eqref{NBE_Db_rho}, \eqref{NBE_Db_sigma}, \eqref{NBE_Db_beta}, \eqref{NBE_Db_alpha} , \eqref{NSE_Db_omega} along $\underline{C}_{\delta}$. These equations can be viewed as ordinary differential equations for the connection coefficients and curvature components which are not bounded by $\delta$ to some positive power in \cite{Chr}. We will see that they are actually bounded by $\delta$ to some positive power under the additional condition \eqref{integral=m0 different}. We will regain the ``smallness'' from Theorem \ref{Higherorder} and improved smallness from Lemma \ref{geometryonS}. More precisely, we prove the following lemma:
\begin{lemma}\label{geometryonCb}Suppose that $\delta>0$ has been chosen such that Theorem \ref{Higherorder} is true for some $k$. Then on the incoming cone $\Cb_{\delta}$, we have the following estimates
\begin{align*}
\|(|u|\nablas)^{k-3}(\Omega\chih)\|_{L^\infty}&\lesssim\delta^{\frac{1}{2}}|u|^{-1},\quad \text{for $i\le k-3$},\\ |u|^{-\frac{1}{2}}\|(|u|\nablas)^{k-2}&(\Omega\chih)\|_{L^4(S_{\delta,u})}+|u|^{-1}\|(|u|\nablas)^{k-1}(\Omega\chih)\|_{L^2(S_{\delta,u})}\lesssim\delta^{\frac{1}{2}}|u|^{-2};\\
\|(|u|\nablas)^i(\rho+\frac{2m_0}{|u|^3},\sigma)\|&_{L^\infty}\lesssim\delta^{\frac{1}{2}}|u|^{-3}\quad \text{for $i\le k-3$},\\ |u|^{-\frac{1}{2}}\|(|u|\nablas)^{k-2}&(\rho,\sigma)\|_{L^4(S_{\delta,u})}+ |u|^{-1}\|(|u|\nablas)^{k-1}(\rho,\sigma)\|_{L^2(S_{\delta,u})}\lesssim\delta^{\frac{1}{2}}|u|^{-3};\\
\|(|u|\nablas)^i\beta\|_{L^\infty}&\lesssim\delta^{\frac{1}{2}}|u|^{-2}\quad \text{for $i\le k-4$},\\ |u|^{-\frac{1}{2}}\|(|u|\nablas)^{k-3}&\beta\|_{L^4(S_{\delta,u})}+ |u|^{-1}\|(|u|\nablas)^{k-2}\beta\|_{L^2(S_{\delta,u})}\lesssim\delta^{\frac{1}{2}}|u|^{-2};\\
\|(|u|\nablas)^i\alpha\|_{L^\infty}&\lesssim\delta^{\frac{1}{2}}|u|^{-2}\quad \text{for $i\le k-5$},\\ |u|^{-\frac{1}{2}}\|(|u|\nablas)^{k-4}&\alpha\|_{L^4(S_{\delta,u})}+ |u|^{-1}\|(|u|\nablas)^{k-3}\alpha\|_{L^2(S_{\delta,u})}\lesssim\delta^{\frac{1}{2}}|u|^{-2};\\
\|(|u|\nablas)^i(\tr\chi-(\frac{2}{|u|}&-\frac{4m_0}{|u|^2}),\omega-\frac{m_0}{|u|^2})\|_{L^\infty}\lesssim\delta^{\frac{1}{2}}|u|^{-3}\quad \text{for $i\le k-3$},\\
|u|^{-\frac{1}{2}}\|(|u|\nablas)^{k-4}&(\tr\chi,\omega)\|_{L^4(S_{\delta,u})}+|u|^{-1}\|(|u|\nablas)^{k-3}(\tr\chi,\omega)\|_{L^2(S_{\delta,u})}\lesssim\delta^{\frac{1}{2}}|u|^{-3}.
\end{align*}
\end{lemma}
\begin{remark}
The smallness gained in the lemma comes from Lemma \ref{geometryonS} therefore from the condition \eqref{integral=m0 different}. Heuristically, the worst control for the geometric quantities on $\Cb_{\ub}$ are propagated from the initial surface $C_{u_0}$. This is a very important technical feature of the work of Christodoulou \cite{Chr}. Therefore, the condition \eqref{integral=m0 different} allows one to improve the geometry on $S_{\delta, u_0}$ hence the geometry on $\Cb_{\delta}$.
\end{remark}
\begin{proof} We will derive the estimates for $\chih$, $\rho$, $\sigma$, $\beta$, $\alpha$, $\omega$ and $\tr\chi$ one by one. The proof for each component is by induction on the number of derivatives. The parameter $\delta>0$ is chosen to be sufficiently small so that Theorem \ref{Higherorder} holds for $k \gg 5$.

For $\chih$, we rewrite \eqref{NSE_Db_chih} as
\begin{align}\label{DbOmegachih}
\Db(\Omega\chih)=\frac{1}{2}\Omega\tr\chib(\Omega\chih)+2(\Omega\chibh,\Omega\chih)\gs+\Omega^2(\nablas \tensor \eta + \eta \tensor \eta -\frac{1}{2}\tr\chi \chibh),
\end{align}
where the second term is the trace of $\Db(\Omega\chih)$.
\begin{remark} In Christodoulou's work, we only obtained the bounds for the second derivatives of the curvature components which implied the $L^\infty$ estimates $|\nablas\tensor\eta|\lesssim{|u|^{-3}}$. If we have estimates for the third derivatives of the curvature components, as stated in Theorem \ref{Higherorder}, we have $|\nablas\tensor\eta|\lesssim\delta^{\frac{1}{2}}|u|^{-3}$. Therefore, we can bound $\Omega^2(\nablas \tensor \eta + \eta \tensor \eta -\frac{1}{2}\tr\chi \chibh)$ in $L^\infty$ by $\delta$ to some positive power. So the right hand side of the above equation can be treated as a small term. This illustrate the idea of the proof and the importance of higher order energy estimates.
\end{remark}
We go back to the proof and integrate $\Db(\Omega\chih)$ along $\Cb_{\delta}$. The first two terms on the right hand side can be absorbed by Gronwall's inequality. In view of the fact that $\chih(\delta,u_0)=0$, according to Theorem \ref{Higherorder}, we have
\begin{align*}
|u||\Omega\chih(\delta,u)|\lesssim\int_{u_0}^u|u'||\Omega^2(\nablas \tensor \eta + \eta \tensor \eta -\frac{1}{2}\tr\chi \chibh)|\D u'\lesssim\delta^{\frac{1}{2}}|u|^{-1}.
\end{align*}
We then use an induction argument to obtain control for higher order derivatives. Suppose that for all $j<i\le k-3$, $|u|^{j}\|\nablas^{j}(\Omega\chih)\|_{L^\infty(\delta,u)}\lesssim\delta^{\frac{1}{2}}|u|^{-2}$, then we apply $\nablas^i$ to \eqref{DbOmegachih} to derive
 \begin{align*}
 \Db\nablas^i(\Omega\chih)&=\frac{1}{2}\Omega\tr\chib\nablas^i(\Omega\chih)+2(\Omega\chibh,\nablas^i(\Omega\chih))\\
 &+\nablas^i(\Omega^2(\nablas \tensor \eta + \eta \tensor \eta -\frac{1}{2}\tr\chi \chibh))+\lot.
 \end{align*}
 Therefore, the third term on the right hand side is bounded by $|u|^i|\nablas^i(\Omega^2(\nablas \tensor \eta + \eta \tensor \eta -\frac{1}{2}\tr\chi \chibh))|\lesssim\delta^{\frac{1}{2}}|u|^{-2}$ thanks to Theorem \ref{Higherorder}. The $\lot$ is of the form
\begin{equation*}
\sum_{j=0}^{i-1}\nablas^j\nablas(\Omega\tr\chib)\cdot\nablas^{i-j-1}(\Omega\chih)+
 \sum_{j=0}^{i-1}\nablas^j\nablas(\Omega\chibh)\cdot\nablas^{i-j-1}(\Omega\chih).
 \end{equation*}
We can bound $\nablas^j\nablas(\Omega\tr\chib)$ and $\nablas^j\nablas(\Omega\chibh)$ by Theorem \ref{Higherorder} and bound $\nablas^{i-j-1}(\Omega\chih)$ by the induction hypothesis. By integrating the above equation and using the Gronwall's inequality to absorb the first two terms, we obtain that
 \begin{align}\label{improve_chih}
\|(|u|\nablas)^i(\Omega\chih)\|_{L^\infty}\lesssim\delta^{\frac{1}{2}}|u|^{-1}
\end{align}
for $i\le k-3$.

For $i=k-2$, we only expect that the $L^4(S_{\delta,u})$ norm on $\chih$ is controlled by $\delta$ to some positive power, and for $i=k-1$, only the $L^2(S_{\delta,u})$ norm is small. This is because we only have information on $\|(|u|\nablas)^{k-1}\eta\|_{L^4(S_{\delta,u})},\|(|u|\nablas)^k\eta\|_{L^2(S_{\delta,u})}\lesssim\delta^{\frac{1}{2}}|u|^{-2}$ from Theorem \ref{Higherorder} but no information for higher derivatives. By a similar argument and integrating along $\Cb_{\delta}$, we obtain
\begin{align}\label{improve_chih1}
|u|^{-\frac{1}{2}}\|(|u|\nablas)^{k-2}(\Omega\chih)\|_{L^4(S_{\delta,u})}+|u|^{-1}\|(|u|\nablas)^{k-1}(\Omega\chih)\|_{L^2(S_{\delta,u})}\lesssim\delta^{\frac{1}{2}}|u|^{-2}.
\end{align}
It is easy to see that \eqref{improve_chih} and \eqref{improve_chih1} also hold if we replace $\Omega\chih$ by $\chih$.

We perform a similar argument for the other components. We now bound $\rho$ and $\sigma$. We recall \eqref{NBE_Db_rho} and \eqref{NBE_Db_sigma},
\begin{align*}
\Db\rho+\frac{3}{2}\Omega\tr\chib \rho&=-\Omega\{\divs \betab+(2\eta-\zeta,\betab)-\frac{1}{2}(\chih,\alphab)\},\\
\Db\sigma+\frac{3}{2}\Omega\tr\chib\sigma&=-\Omega\{\curls\betab+(2\etab-\zeta,{}^*\betab)+\frac{1}{2}\chih\wedge\alphab\}.
\end{align*}
Because $k\ge3$, we have $|\nablas\betab|\lesssim\delta|u|^{-5}$. Thus,
\begin{align*}
|-\Omega\{\divs \betab+(2\eta-\zeta,\betab)-\frac{1}{2}(\chih,\alphab)\}|+|-\Omega\{\curls\betab+(2\etab-\zeta,{}^*\betab)+\frac{1}{2}\chih\wedge\alphab\}|\lesssim\delta|u|^{-5}.
\end{align*}
We renormalize the first equation in the following form
\begin{align*}
\Db(\rho+\frac{2m_0}{|u|^3})+\frac{3}{2}\Omega\tr\chib(\rho+\frac{2m_0}{|u|^3})=\frac{3m_0}{|u|^3}(\Omega\tr\chib+\frac{2}{|u|})-\Omega\{\divs \betab+(2\eta-\zeta,\betab)-\frac{1}{2}(\chih,\alphab)\}.
\end{align*}
In view of the facts that $|\Omega\tr\chib+\frac{2}{|u|}|\lesssim\delta|u|^{-2}$ and the improved smallness for initial data on $S_{\delta,u_0}$ in Lemma \ref{geometryonS}, we integrate the above equations along $\Cb_{\delta}$ and we obtain
\begin{align*}
|\rho+\frac{2m_0}{|u|^3}| \lesssim \delta^{\frac{1}{2}}|u|^{-3}\quad \text{and} \quad |\sigma|\lesssim\delta^{\frac{1}{2}}|u|^{-3}.
\end{align*}
Similarly, we can also run an induction argument as above to obtain
\begin{equation}\label{improve_rhosigma}
\begin{split}
&\|(|u|\nablas)^i(\rho+\frac{2m_0}{|u|^3},\sigma)\|_{L^\infty}\lesssim\delta^{\frac{1}{2}}|u|^{-3}\ \text{for $i\le k-3$},\\ |u|^{-\frac{1}{2}}&\|(|u|\nablas)^{k-2}(\rho,\sigma)\|_{L^4(S_{\delta,u})}, |u|^{-1}\|(|u|\nablas)^{k-1}(\rho,\sigma)\|_{L^2(S_{\delta,u})}\lesssim\delta^{\frac{1}{2}}|u|^{-3}.
\end{split}
\end{equation}

We now bound $\beta$. We examine the terms in \eqref{NBE_Db_beta},
\begin{align*}
\Db\beta+\frac{1}{2}\Omega\tr\chib\beta-\Omega\chibh \cdot \beta=-\omegab \beta+\Omega\{\ds \rho+{}^*\ds \sigma+3\eta\rho+3{}^*\eta\sigma+2\chih\cdot\betab\}.
\end{align*}
If $k\ge 4$, we have the bound $|\ds \rho+{}^*\ds \sigma|\lesssim\delta^{\frac{1}{2}}|u|^{-4}$ thanks to \eqref{improve_rhosigma} and the bound $|-\omegab \beta+\Omega(3\eta\rho+3{}^*\eta\sigma+2\chih\cdot\betab)|\lesssim\delta^{\frac{1}{2}}|u|^{-5}$ thanks to Theorem \ref{Higherorder}. We remark that there is a loss of derivatives here because we control $\beta$ in terms of  derivatives of $\rho$ and $\sigma$ via Bianchi equations. We then integrate the equation along $\Cb_{\delta}$, combined with the estimates on the initial data on $S_{\delta,u_0}$ obtained in Lemma \ref{geometryonS},  we have
\begin{equation}\label{improve_beta}
\begin{split}
&\|(|u|\nablas)^i\beta\|_{L^\infty(S_{\delta,u})}\lesssim\delta^{\frac{1}{2}}|u|^{-2}\ \text{for $i\le k-4$},\\ |u|^{-\frac{1}{2}}&\|(|u|\nablas)^{k-3}\beta\|_{L^4(S_{\delta,u})}+ |u|^{-1}\|(|u|\nablas)^{k-2}\beta\|_{L^2(S_{\delta,u})}\lesssim\delta^{\frac{1}{2}}|u|^{-2}.
\end{split}
\end{equation}
Notice that we only have $L^\infty$ estimates on $(k-4)$-th derivatives on $\beta$ instead of $(k-3)$-th derivatives due to the derivative loss we just mentioned.

We now bound $\alpha$. We rewrite\eqref{NBE_Db_alpha} as follows,
\begin{align*}
\Db\alpha-\frac{1}{2}\Omega\tr\chib \alpha+2\omegab\alpha-2(\Omega\chibh,\alpha)\gs=-\Omega\{-\nablas\tensor\beta -(4\eta+\zeta)\tensor \beta+3\chih \rho+3{}^*\chih \sigma\}
\end{align*}
where $2(\Omega\chibh,\alpha)$ is the trace of $\Db\alpha$. If $k\ge5$, by \eqref{improve_beta}, we have $|\nablas\tensor\beta|\lesssim\delta^{\frac{1}{2}}|u|^{-3}$. We remark that at this point there is a loss of derivatives here because we use the derivatives of $\beta$.  We must also use $|\chih|+|\beta|\lesssim\delta^{\frac{1}{2}}|u|^{-2}$ by \eqref{improve_chih} and \eqref{improve_beta} to get $|-\Omega\{-\nablas\tensor\beta -(4\eta+\zeta)\tensor \beta+3\chih \rho+3{}^*\chih \sigma\}|\lesssim\delta^{\frac{1}{2}}|u|^{-3}$. Therefore, in view of the fact that $\alpha(\delta,u_0)=0$ on $S_{\delta,u_0}$, by integrating the equation along $\Cb_{\delta}$, we obtain
\begin{equation}\label{improve_alpha}
\begin{split}
&\|(|u|\nablas)^i\alpha\|_{L^\infty(S_{\delta,u})}\lesssim\delta^{\frac{1}{2}}|u|^{-2}\ \text{for $i\le k-5$},\\ |u|^{-\frac{1}{2}}&\|(|u|\nablas)^{k-4}\alpha\|_{L^4(S_{\delta,u})}, |u|^{-1}\|(|u|\nablas)^{k-3}\alpha\|_{L^2(S_{\delta,u})}\lesssim\delta^{\frac{1}{2}}|u|^{-2}.
\end{split}
\end{equation}
We remark that the terms $-\frac{1}{2}\Omega\tr\chib \alpha+2\omegab\alpha-2(\Omega\chibh,\alpha)\gs$ on the left hand side are absorbed by Gronwall's inequality. We also want to point out that we will not use the $L^2(S_{\delta,u})$ estimates on $\nablas^{k-3}\alpha$ in sequel. Notice also that we only have $L^\infty$ estimates on $(k-5)$-th derivatives on $\alpha$ due to the derivative loss we just mentioned.

Finally, we bound $\tr\chi$ and $\omega$. We renormalize \eqref{NSE_Db_trchi} and \eqref{NSE_Db_omega} as follows
\begin{align*}
&\quad \Db(\Omega\tr\chi-(\frac{2}{|u|}-\frac{4m_0}{|u|^2}))+\frac{1}{2}\Omega\tr\chib(\Omega\tr\chi-(\frac{2}{|u|}-\frac{4m_0}{|u|^2}))\\
&=-\frac{1}{|u|}(1-\frac{2m_0}{|u|})(\Omega\tr\chib+\frac{2}{|u|})+(2\Omega^2\rho+\frac{4m_0}{|u|^3})+\Omega^2(2\divs\eta+2|\eta|^2-(\chih,\chibh)),\\
&\quad \Db  (\omega-\frac{m_0}{|u|^2}) =-(\Omega^2\rho+\frac{2m_0}{|u|^3})+\Omega^2(2(\eta,\etab)-|\etab|^2).
\end{align*}
We have already obtained the following estimates: $|\Omega\tr\chib+\frac{2}{|u|}|\lesssim\delta|u|^{-2}$, $|\Omega^2\rho+\frac{2m_0}{|u|^3}|\lesssim\delta^{\frac{1}{2}}|u|^{-3}$,  $|\Omega-1|\lesssim\delta|u|^{-2}$, $|\Omega^2(2\divs\eta+2|\eta|^2-(\chih,\chibh)|+|\Omega^2(2(\eta,\etab)-|\etab|^2)|\lesssim\delta^{\frac{1}{2}}|u|^{-3}$ and $|\Omega\tr\chi(\delta,u_0)-(\frac{2}{|u|}-\frac{4m_0}{|u|^2})|\lesssim\delta^{\frac{1}{2}}|u|^{-2}$. In view of the fact that $\omega(\delta,u_0)=0$, by integrating the equation along $\Cb_{\delta}$, we obtain
\begin{align*}
\|\tr\chi-(\frac{2}{|u|}-\frac{4m_0}{|u|^2})\|_{L^\infty}+\|\omega-\frac{m_0}{|u|^2}\|_{L^\infty}\lesssim\delta^{\frac{1}{2}}|u|^{-2}.
\end{align*}
where $\Omega\tr\chi$ has been replaced by $\tr\chi$ thanks again to $|\Omega-1|\lesssim\delta|u|^{-2}$. Similarly, we can derive estimates for higher order derivatives. Finally, we have
\begin{equation}\label{improve_trchiomega}
\begin{split}
&\|(|u|\nablas)^i(\tr\chi-(\frac{2}{|u|}-\frac{4m_0}{|u|^2}),\omega-\frac{m_0}{|u|^2})\|_{L^\infty(S_{\delta,u})}\lesssim\delta^{\frac{1}{2}}|u|^{-3}\ \text{for $i\le k-3$},\\ |u|^{-\frac{1}{2}}&\|(|u|\nablas)^{k-2}(\tr\chi,\omega)\|_{L^4(S_{\delta,u})}, |u|^{-1}\|(|u|\nablas)^{k-1}(\tr\chi,\omega)\|_{L^2(S_{\delta,u})}\lesssim\delta^{\frac{1}{2}}|u|^{-3}.
\end{split}
\end{equation}
We complete the proof of the lemma.
\end{proof}

\subsection{Geometry on $C_{u_0}^{[\delta, \delta+1]}$.}\label{Section Geometry on outgoing cone}
We extend the initial data on $C_{u_0}^{[0,\delta]}$ to $C_{u_0}^{[0,\delta+1]}$ by setting $\chih \equiv 0$ for $\ub\in[\delta,\delta+1]$ and we will derive the estimates on $C_{u_0}^{[\delta,\delta+1]}$ for all the connection coefficients and curvature components. For the sake of simplicity, we denote $C_{u_0}=C_{u_0}^{[\delta,\delta+1]}$ in this subsection.

We first give another initial data on $C_{u_0} \cup \Cb_{\delta}$ and it is in fact an initial data set for the Schwarzschild space-time with mass $m_0$. To distinguish this new data from the old one,  we shall use a lower index $m_0$ for its connection coefficients and curvature components. The new initial data set consists of the following quantities: a $\gs_{m_0}(\delta,u_0)$ of $S_{\delta,u_0}$ which is the round metric of the sphere with radius $|u_0|$, the torsion $\zeta_{m_0}(\delta,u_0) \equiv 0$ on $S_{\delta,u_0}$, two null expansions $\tr\chi_{m_0}(\delta,u_0)\equiv\frac{2}{|u_0|}-\frac{4m_0}{|u_0|^2}$ and $\tr\chib_{m_0}(\delta,u_0) \equiv -\frac{2}{|u_0|}$ on $S_{\delta,u_0}$, the lapse function $\Omega_{m_0} \equiv 1$ on $C_{u_0}\cup\Cb_{\delta}$ and two shears $\chih_{m_0} \equiv 0$ on $C_{u_0}$ and $\chibh_{m_0} \equiv 0$ on $\Cb_\delta$.

We also impose the the condition $|u_0|>2m_0$. This condition guarantees that the future development of the data is in the domain of outer communication of the Schwarzschild space-time. We remark that the only nonzero null components for the new data on $C_{u_0}$ are $\tr\chi_{m_0}$, $\tr\chib_{m_0}$, $\omegab_{m_0}$ and $\rho_{m_0}$ with $\tr\chi_{m_0} >0$, and the only nonzero null components for the new data on $\Cb_{\delta}$ are $\tr\chi_{m_0}$, $\tr\chib_{m_0}$, $\omegab_{m_0}$ and $\rho_{m_0}$.

We now restate the estimates on the sphere $S_{\delta,u_0}$ derived Lemma \ref{geometryonS} as well as the estimates given in Chapter 2 of \cite{Chr}: we have $\chih\equiv0$ $\alpha\equiv0$, $\omega\equiv0$, and for all $k$,
\begin{align}\label{geometryonS1}
|\nablas^k(\chibh,\eta,\etab,\omegab-\omegab_{m_0},
\tr\chi-\tr\chi_{m_0},\tr\chib-\tr\chib_{m_0},\beta,\rho-\rho_{m_0},\sigma,\betab,\alphab)|\lesssim_k\delta^{\frac{1}{2}}.
\end{align}
Notice that we have dropped the weights $|u_0|$. We will not use the weights in sequel since we can fix a $|u_0| > 2m_0$ now and they are not relevant from now on. We shall prove that the above estimates actually hold on $C_{u_0}$:
\begin{lemma}\label{geometry on Cu0 delta to delta plus 1}On the incoming cone $C_{u_0}^{[\delta, \delta+1]}$, we have $\chih\equiv0$ $\alpha\equiv0$, $\omega\equiv0$, and for all $k$, we have
\begin{equation*}
|\nablas^k(\chibh,\eta,\etab,\omegab-\omegab_{m_0},
\tr\chi-\tr\chi_{m_0},\tr\chib-\tr\chib_{m_0},\beta,\rho-\rho_{m_0},\sigma,\betab,\alphab)|\lesssim_k\delta^{\frac{1}{2}}.
\end{equation*}
\end{lemma}
\begin{proof}By construction, we also have $\chih=0$, $\alpha=0$ and $\omega=0$ on $C_{u_0}$.  The idea to prove the estimate is to use null structure equations in a correct order.

We first control $\tr\chi$. According to \eqref{geometryonS1} and choosing $\delta$ sufficiently small, we have $\tr\chi(\delta,u_0)>0$ so that $\tr\chi$ can be solved for $\ub\in[\delta,\delta+1]$ and $0<\tr\chi(\ub,u_0)<\tr\chi(\delta,u_0)$. We then rewrite \eqref{NSE_D_trchi} as
\begin{align}\label{D_trchi-trchim0}
D(\tr\chi-\tr\chi_{m_0})=-\frac{1}{2}(\tr\chi+\tr\chi_{m_0})(\tr\chi-\tr\chi_{m_0}).
\end{align}
Because $(\tr\chi+\tr\chi_{m_0})$ is bounded pointwisely, thanks to Gronwall's inequality and \eqref{geometryonS1}, we have
\begin{align*}
|\tr\chi-\tr\chi_{m_0}|\lesssim |\tr\chi(\delta,u_0)-(\frac{2}{|u_0|}-\frac{4}{|u_0|^2})|\lesssim\delta^{\frac{1}{2}}.
\end{align*}
We can also apply $\nablas^k$ to the above argument and use an induction argument as before to derive
\begin{align*}
|\nablas^k(\tr\chi-\tr\chi_{m_0})|\lesssim |\nablas^k(\tr\chi(\delta,u_0)-(\frac{2}{|u_0|}-\frac{4}{|u_0|^2}))|\lesssim_k\delta^{\frac{1}{2}}.
\end{align*}

We now bound $\eta=-\etab$. We eliminate $\beta$ in (\ref{NSE_D_eta}) and (\ref{NSE_div_chih}) to derive
\begin{align*}
D\eta+\tr\chi\eta=-\frac{1}{2}\ds\tr\chi.
\end{align*}
Therefore, by Gronwall's inequality and an induction argument, we obtain
\begin{align*}
|\nablas^k(\eta,\etab)(\ub,u_0)|\lesssim_k\delta^{\frac{1}{2}}.
\end{align*}

Before going to the next component, we have to digress to the estimates for the Gauss curvature $K$ on $S_{\ub,u_0}$. Recall that $K$ satisfies the following propagation equations (see Chapter 5 of \cite{Chr}):
\begin{align*}
DK+\Omega\tr\chi K=\divs\divs(\Omega\chih)-\frac{1}{2}\Deltas(\Omega\tr\chi),
\end{align*}
We rewrite the equation on $C_{u_0}$ as
\begin{align*}
D(K-K_{m_0})+\tr\chi(K-K_{m_0})=-(\tr\chi-\tr\chi_{m_0})K_{m_0}-\frac{1}{2}\Deltas\tr\chi,
\end{align*}
which yields $|\nablas^k(K(\ub,u_0)-K_{m_0}(\ub,u_0))|\lesssim_k\delta^{\frac{1}{2}}$.

To bound $\chib$ and $\chibh$, we eliminate $\rho$ in (\ref{NSE_D_trchib}) and (\ref{NSE_Gauss}) to derive the following propagation equation for $\tr\chib-\tr\chib_{m_0}$:
\begin{align*}
D(\tr\chib-\tr\chib_{m_0})+\tr\chi(\tr\chib-\tr\chib_{m_0})
=2(K_{m_0}-K)-\divs\eta+|\eta|^2,
\end{align*}
and we rewrite \eqref{NSE_D_chibh} as
\begin{align*}
D\chibh-\frac{1}{2}\tr\chi\chibh=-\nablas\tensor \eta+\eta\tensor\eta.
\end{align*}
These equations yield $|\nablas^k(\tr\chib(\ub,u_0)-\tr\chib_{m_0}(\ub,u_0))| + |\nablas^k\chibh(\ub,u_0)|\lesssim_k\delta^{\frac{1}{2}}$.

We then use (\ref{NSE_div_chih}), (\ref{NSE_div_chibh}) and the estimates above to derive $|\nablas^k(\beta,\betab)(\ub,u_0)|\lesssim_k\delta^{\frac{1}{2}}$.

We rewrite (\ref{NSE_Gauss}) as
\begin{align*}
K-K_{m_0}=-\frac{1}{4}(\tr\chi\tr\chib-\tr\chi_{m_0}\tr\chib_{m_0})-(\rho-\rho_{m_0}),
\end{align*}
which implies that $|\nablas^k(\rho(\ub,u_0)-\rho_{m_0}(\ub,u_0))|\lesssim_k\delta^{\frac{1}{2}}$.

We rewrite (\ref{NSE_D_omegab}) as
\begin{align*}
D(\omegab-\omegab_{m_0})=-3|\eta|^2-(\rho-\rho_{m_0}),
\end{align*}
which implies $|\nablas^k(\omegab(\ub,u_0)-\omegab_{m_0}(\ub,u_0))|\lesssim_k\delta^{\frac{1}{2}}$.

Finally, (\ref{NBE_D_alphab}) implies $|\nablas^k\alphab(\ub,u_0)|\lesssim_k\delta^{\frac{1}{2}}$. We complete the proof of the lemma.
\end{proof}

\subsection{Construction of the Transition Region.}\label{Section transition slice}

We consider a characteristic initial data problem on $C_{u_0}^{[\delta,\delta+1]}\cup \Cb_{\delta}^{[u_0,-1-\delta]}$ where the data are given by with $\chih=0$ on $C_{u_0}^{[\delta,\delta+1]}$ and the data on $\Cb_{\delta}^{[u_0,-1-\delta]}$ is induced by the solution constructed by Christodoulou in \cite{Chr} under the condition \eqref{integral=m0} with $|u_0|>2m_0$. We continue to use the space-time metric $g$ to denote the solution of this problem.

Adapted to this problem, we introduce the following norms:
\begin{align*}
\mathscr{R}_k^{\ub}(u)&=\|\nablas^k(\alpha,\beta,\rho-\rho_{m_0},\sigma,\betab)\|_{L^2(C_u^{\ub})}, \,\,\, \underline{\mathscr{R}}_k^u(\ub)=\|\nablas^k(\beta,\rho-\rho_{m_0},\sigma,\betab,\alphab)\|_{L^2(\Cb_{\ub}^u)},\\
\mathscr{O}_k(\ub,u)&=\|\nablas^k(\chih,\chibh,\eta,\etab,\tr\chi-\tr\chi_{m_0},\tr\chib-\tr\chib_{m_0},\omega-\omega_{m_0},\omegab-\omegab_{m_0})\|_{L^2(S_{\ub,u})}.
\end{align*}
We use short hand notations $C_u^{\ub}$ referring to $C_u^{[\delta,\ub]}$ and $\Cb_{\ub}^u$ referring to $\Cb_{\ub}^{[u_0,u]}$. We take a small parameter $\epsilon$ to be proportional to $\delta^\frac{1}{2}$, then the results in Section \ref{Section Geometry on intersecting sphere}, Section \ref{Section Geometry on incoming cone} and Section \ref{Section Geometry on outgoing cone} can be summarized in the following proposition:
\begin{proposition}\label{initialb}
Fix $k \gg 5$. For any $\epsilon>0$, there exists $\delta_0>0$ depending only on the $C^{k+N}$ bounds of the seed data for some sufficient large $N$ and on $\epsilon$, such that for all $\delta<\delta_0$ and for all $(\ub,u) \in [\delta, \delta+1] \times [u_0,-1-\delta]$, we have,
\begin{align*}
\mathscr{R}_{\le k-2}^{\delta+1}(u_0)+\underline{\mathscr{R}}_{\le k-2}^{-1-\delta}(\delta)+\mathscr{O}_{\le k-2}(\ub,u_0)+\mathscr{O}_{\le k-2}(\delta,u)\leq\epsilon,\\
\|\nablas^{\le k-4}\alpha\|_{L^4(S_{\delta,u})}+ \|\nablas^{\le k}\alphab\|_{L^\infty(S_{\ub,u_0})}\le\epsilon.
\end{align*}
\end{proposition}
Based on this proposition, we will prove the following theorem in this section.
\begin{theorem}\label{Schwarzschild}
Let $k \gg 5$ and $\epsilon>0$.  If $\delta > 0$ is sufficiently small depending on $C^{k+N}$ bounds of the seed data for sufficient large $N$ and on $\epsilon>0$, then there is a unique smooth solution $(M_{\epsilon_0},g)$ of vacuum Einstein field equations to the characteristic initial data problem described above. The space-time $M_{\epsilon_0}$ corresponds to the region $\delta\le\ub\le\delta+\epsilon_0$ and $u_0\le u\le u_0+\epsilon_0$. The parameter  $\epsilon_0>0$ depends only on $m_0$ and $u_0$ and it is independent of $\delta$ when $\delta$ is sufficiently small. Moreover, the space-time  $(M_{\epsilon_0},g)$ is $\epsilon$-close to the Schwarzschild space-time with mass $m_0$ in $C^{k-3}$ norms, namely,
 \begin{equation*}
 \|g-g_{m_0}\|_{C^{k-3}(M_{\epsilon_0},g_{m_0})}\lesssim\epsilon.
 \end{equation*}
\end{theorem}
\begin{remark}
We would like to emphasize that once $m_0$ and $u_0$ is fixed, the size of $\epsilon_0$ is also fixed. In the introduction, we mentioned that in order to use the Corvino-Schoen construction, the gluing region must have a fixed size when $\delta$ goes to zero. The size of $\epsilon_0$ meets this requirement.
\end{remark}
\begin{remark}It is not hard to see that $g_{m_0}$ is isometric to the Schwarzschild space-time with mass $m_0$. In fact, the initial data for $g_{m_0}$ on $C_{u_0}^{[\delta,\delta+\epsilon_0]}\cup\Cb_{\delta}$ is spherical symmetric, then $g_{m_0}$ is also spherical symmetric. The orbits of the isometry group $SO(3)$ are simply $S_{\ub,u}$. Therefore by Birkhoff Theorem, $g_{m_0}$ is isometric to the Schwarzschild space-time. The mass parameter $m_0$ can be figured out by computing the Hawking mass of $S_{\delta,u_0}$ in $g_{m_0}$.
\end{remark}

\begin{proof} The existence of local solutions is due to the classical result \cite{Ren} of Rendall. We make use of a recent improvement \cite{Luk} due to Luk. Rendall's original proof only gave the local existence around the intersection of two initial null hypersurface and Luk showed that the existence region can be enlarged to a full neighborhood of initial hypersurfaces in their future development. The main advantage of \cite{Luk} is that it also gives the estimates on the solution which depend only on the initial data. We now rephrase the results of \cite{Luk} in our situation: there exists $\epsilon_0$ which depends only on $m_0$ and $u_0$ and is independent of $\delta$ hence $\epsilon$, such that we can solve the metric $g$ for $(\ub,u) \in [\delta, \delta+\epsilon_0] \times [u_0,u_0+\epsilon_0]$. Moreover, we have the following bounds:
\begin{proposition}\label{bound_luk}
We have a smooth solution $(M_{\epsilon_0},g)$ where the $M_{\epsilon_0}$ corresponds to  $(\ub,u) \in [\delta, \delta+1] \times [u_0,-1-\delta]$. Moreover, we have the following estimates:
\begin{align*}
\mathscr{R}^{\delta+\epsilon_0}(u)+\underline{\mathscr{R}}^{u_0+\epsilon_0}(\ub)+\mathscr{O}(\ub,u_0)+\mathscr{O}(\delta,u)\le C(m_0,u_0,\epsilon_0).
\end{align*}
We also have the Sobolev inequalities:
\begin{align*}
\|\phi\|_{L^4(S_{\ub,u})}&\le C(m_0,u_0,\epsilon_0)(\|\nablas\phi\|_{L^2(S_{\ub,u})}+\|\phi\|_{L^2(S_{\ub,u})}),\\
\|\phi\|_{L^\infty(S_{\ub,u})}&\le C(m_0,u_0,\epsilon_0)\sum_{i=0}^2\|\nablas^i\phi\|_{L^2(S_{\ub,u})}.
\end{align*}
\end{proposition}
For the sake of simplicity, we will suppress the lower index $\le3$ in the norms, e.g. $\mathscr{R}=\mathscr{R}_{\le3}$. In the remaining part of the current section, $C$ refers to various constants depending only on $m_0$, $u_0$ and $\epsilon_0$ and $A \lesssim B$ refers to $A\le CB$.

Since we expect the solution is close to some Schwarzschild space-time, we have to derive more precise energy estimates on $M_{\epsilon_0}$. Combined with Rendall's result \cite{Ren} and usual bootstrap arguments, our estimates will be good enough to prove existence. We remark that Luk's result \cite{Luk} is extremely convenient to use at this stage:  we can avoid a long bootstrap argument and we can take the existence for granted.

We start to prove Theorem \ref{Schwarzschild} by deriving the estimates only for $\mathscr{R}_k$, $\underline{\mathscr{R}}_k$ and $\mathscr{O}_k$ for $k=0,1,2,3$. In this case, we can take $K=5$.  The estimates for higher order derivatives can be obtained by a routine induction argument. Since it is similar to the higher order energy estimates derived in Section \ref{Higher Order Energy Estimates}, we will omit the proof.

We first prove the following lemma which claims that the connection coefficients can be controlled by the initial data and the curvature components.
\begin{lemma}\label{connectionb}
$\mathscr{O}(\ub,u)\lesssim\sup_{\delta\le\ub'\le\ub}(\mathscr{O}(\ub',u_0)+\underline{\mathscr{R}}^u(\ub'))+\sup_{u_0\le u'\le u}(\mathscr{O}(\delta,u')+\mathscr{R}^{\ub}(u')).$
\end{lemma}

\begin{proof} The proof is once again by integrating the null structure equations (\ref{NSE_Dh_chih})-(\ref{NSE_Db_omega}). At the beginning, we remark that for a tensor $\phi$, we have
\begin{align*}
C^{-1}\sum_{i=0}^3\|\Omega\phi\|_{L^2(S_{\ub,u})}\le\sum_{i=0}^3\|\phi\|_{L^2(S_{\ub,u})}\le C\sum_{i=0}^3\|\Omega\phi\|_{L^2(S_{\ub,u})}.
\end{align*}
The can be derived by the bounds for $\Omega$ and $\ds\log\Omega=(\eta+\etab)/2$. Similar inequalities also hold for $L^4(S_{\ub,u})$ and $L^\infty$ norms.

We start with the bounds on $\chih$. In view of (\ref{NSE_Dh_chih}), we have
\begin{equation*}
D|\chih'|^2+2\Omega\tr\chi|\chih'|^2\le2|\chih'||\alpha|.
\end{equation*}
We then use the $L^\infty$ bounds for $\Omega$ and $\tr\chi$ in Proposition \ref{bound_luk} as well as the Sobolev inequalities to derive
\begin{align*}
\|\chih'\|_{L^\infty(S_{\ub,u})}&\lesssim\|\chih'\|_{L^\infty(S_{\delta,u})}+\int_\delta^{\ub}\|\alpha\|_{L^\infty(S_{\ub',u})}\\
&\lesssim\sum_{i=0}^2\left(\|\nablas^i\chih'\|_{L^2(S_{\delta,u})}+\int_\delta^{\ub}\|\nablas^i\alpha\|_{L^2(S_{\ub',u})}\right) \lesssim\mathscr{O}(\delta,u)+\mathscr{R}^{\ub}(u).
\end{align*}
We now turn to the derivatives of $\chih$. We first rewrite (\ref{NSE_Dh_chih}) as $D\chih'=\Omega^2|\chih'|^2\gs-\alpha$ and then commute with $\nablas$ to derive
\begin{align*}
D\nablas\chih'=\nablas(\Omega\chi)\cdot\chih'+\nablas(\Omega^2)|\chih'|^2\gs+\Omega^2(\chih',\nablas\chih')\gs-\nablas\alpha.
\end{align*}
The last term $\nablas\alpha$ can be bounded by $\mathscr{R}^{\ub}(u)$. For the remaining terms, each of them contains a factor $|\chih'|$. We can bound the other factors by $C$ and bound $|\chih'(\ub,u)|$ by $\mathscr{O}(\delta,u)+\mathscr{R}^{\ub}(u)$. Thus, we obtain
\begin{align*}
\|\nablas\chih'\|_{L^\infty(S_{\ub,u})}\lesssim\mathscr{O}(\delta,u)+\mathscr{R}^{\ub}(u).
\end{align*}
The estimates on higher order derivatives are similar. For second and third derivatives of $\chih$, we have
\begin{align*}
D\nablas^{k}\chih'=\sum_{i=0}^{k-1}\nablas^i\nablas(\Omega\chi)\cdot\nablas^{k-1-i}\chih'+\sum_{i=0}^k\nablas^i(\Omega^2)\cdot\nablas^{k-i}|\chih'|^2\gs-\nablas^k\alpha.
\end{align*}
 When $k=2$, for the first and the second terms, we can use H\"{o}lder's inequality by placing the factors with highest order derivatives in $L^4(S_{\ub,u})$ and the others in $L^\infty$. In this way,  for each product, one factor is bounded by $\mathscr{O}(\delta,u)+\mathscr{R}^{\ub}(u)$ and the other one is bounded by $C$. The last term $\nablas^k\alpha$ is also bounded by $\mathscr{R}^{\ub}(u)$. When $k=3$, there are two cases: if there are factors containing third order derivatives, we can bound it in $L^2(S_{\ub,u})$ and bound the other factors in $L^\infty$; if not, we bound the factors containing second order derivatives in $L^4(S_{\ub,u})$ and bound the other in $L^\infty$. Therefore, we obtain
\begin{align*}
\|\nablas^2\chih'\|_{L^4(S_{\ub,u})} + \|\nablas^3\chih'\|_{L^2(S_{\ub,u})}\lesssim\mathscr{O}(\delta,u)+\mathscr{R}^{\ub}(u).
\end{align*}

For $\chibh$, in a similar manner, we can use (\ref{NSE_Dbh_chibh}) to derive
\begin{align*}
\|\chibh\|_{L^\infty}+\|\nablas\chibh\|_{L^\infty}+\|\nablas^2\chibh\|_{L^4(S_{\ub,u})}+\|\nablas^3\chibh\|_{L^2(S_{\ub,u})}\lesssim\mathscr{O}(\ub,u_0)+\underline{\mathscr{R}}^{u}(\ub).
\end{align*}

For $\eta$ and $\etab$, we consider (\ref{NSE_D_eta}) and (\ref{NSE_Db_etab}) as a coupled system and we derive
\begin{align*}
\|\eta\|_{L^\infty(S_{\ub,u})}\lesssim\|\eta\|_{L^\infty(S_{\delta,u})}+\int_\delta^{\ub}\|\etab\|_{L^\infty(S_{\ub',u})}+\int_\delta^{\ub}\|\beta\|_{L^\infty(S_{\ub',u})},\\
\|\etab\|_{L^\infty(S_{\ub,u})}\lesssim\|\etab\|_{L^\infty(S_{\ub,u_0})}+\int_{u_0}^{u}\|\eta\|_{L^\infty(S_{\ub,u'})}+\int_{u_0}^{u}\|\betab\|_{L^\infty(S_{\ub,u'})}.
\end{align*}
We substitute the second inequality to the first one and take $\sup_{\delta\le\ub'\le\ub}$ on both sides, we have
\begin{align*}
\sup_{\delta\le\ub'\le\ub}\|\eta\|_{L^\infty(S_{\ub',u})}&\lesssim\|\eta\|_{L^\infty(S_{\delta,u})}+\sup_{\delta\le\ub'\le\ub}\|\etab\|_{L^\infty(S_{\ub',u_0})}+\int_{u_0}^{u}\sup_{\delta\le\ub'\le\ub}\|\eta\|_{L^\infty(S_{\ub',u'})}\\
&+\int_{u_0}^{u}\sup_{\delta\le\ub'\le\ub}\|\betab\|_{L^\infty(S_{\ub',u'})}+\int_\delta^{\ub}\|\beta\|_{L^\infty(S_{\ub',u})}.
\end{align*}
Thanks to Gronwall and Sobolev inequalities, we have
\begin{align*}
\sup_{\delta\le\ub'\le\ub}&\|\eta\|_{L^\infty(S_{\ub',u})}\lesssim\|\eta\|_{L^\infty(S_{\delta,u})}+\sup_{\delta\le\ub'\le\ub}\|\etab\|_{L^\infty(S_{\ub,u_0})}+\int_{u_0}^{u}\sup_{\delta\le\ub'\le\ub}\|\betab\|_{L^\infty(S_{\ub',u'})}\\
&+\int_\delta^{\ub}\|\beta\|_{L^\infty(S_{\ub',u})} \lesssim\mathscr{O}(\delta,u)+\sup_{\delta\le\ub'\le\ub}\mathscr{O}(\ub',u_0)+\sup_{\delta\le\ub'\le\ub}\underline{\mathscr{R}}^u(\ub')+\mathscr{R}^{\ub}(u).
\end{align*}
We then define $\mathscr{P}(\ub,u)$ and rewrite the above estimates as
\begin{align*}
\sup_{\delta\le\ub'\le\ub}\|\etab\|_{L^\infty(\ub',u)}\lesssim\mathscr{P}(\ub,u)\triangleq\sup_{\delta\le\ub'\le\ub}(\mathscr{O}(\ub',u_0)+\underline{\mathscr{R}}^u(\ub'))+\sup_{u_0\le u'\le u}(\mathscr{O}(\delta,u')+\mathscr{R}^{\ub}(u')).
\end{align*}
We turn to the derivatives of $\eta$ and $\etab$. Recall that we can commute derivatives with (\ref{NSE_D_eta}) and (\ref{NSE_Db_etab}) to derive (\ref{Dk_eta}) and (\ref{Dbk_etab}). When $k=1$, we have
\begin{align*}
\|\nablas\eta\|_{L^\infty(S_{\ub,u})}\lesssim\|\nablas\eta\|_{L^\infty(S_{\delta,u})}+\int_\delta^{\ub}\|\nablas\etab\|_{L^\infty(S_{\ub',u})}+\int_\delta^{\ub}\|\nablas\beta\|_{L^\infty(S_{\ub',u})}+\mathscr{P}(\ub,u),\\
\|\nablas\etab\|_{L^\infty(S_{\ub,u})}\lesssim\|\nablas\etab\|_{L^\infty(S_{\ub,u_0})}+\int_{u_0}^{u}\|\nablas\eta\|_{L^\infty(S_{\ub,u'})}+\int_{u_0}^{u}\|\nablas\betab\|_{L^\infty(S_{\ub,u'})}+\mathscr{P}(\ub,u).
\end{align*}
By a similar argument as above, we obtain
\begin{align*}
\|\nablas\eta\|_{L^\infty(S_{\ub,u})},\|\nablas\etab\|_{L^\infty(S_{\ub,u})}\lesssim\mathscr{P}(\ub,u).
\end{align*}
Similar arguments also apply to the cases when $k=2,3$. We finally obtain
\begin{align*}
\|\nablas^2(\eta,\etab)\|_{L^4(S_{\ub,u})},\|\nablas^3(\eta,\etab)\|_{L^2(S_{\ub,u})}\lesssim\mathscr{P}(\ub,u).
\end{align*}

For $\omega$ and $\omegab$, the estimates rely on (\ref{NSE_D_omegab}) and (\ref{NSE_Db_omega}). We rewrite (\ref{NSE_D_omegab}) as
\begin{equation}\label{D omegab minus omegab m0}
D(\omegab-\omegab_{m_0})=(\Omega^2-\Omega_{m_0}^2)(2(\eta,\etab)-|\eta|^2-\rho)+\Omega_{m_0}^2(\rho-\rho_{m_0}).
\end{equation}
In view of the facts that $\Db\log\Omega=\omegab$, $|\Omega|+|\Omega^{-1}| \le C$,  we have
\begin{align*}
\|\Omega-\Omega_{m_0}\|_{L^\infty(S_{\ub,u})}&\lesssim\|\log\Omega-\log\Omega_{m_0}\|_{L^\infty(S_{\ub,u})}\\
&\lesssim\int_{u_0}^u\|\omegab-\omegab_{m_0}\|_{L^\infty(S_{\ub,u'})}\lesssim\sup_{u_0\le u'\le u}\|\omegab-\omegab_{m_0}\|_{L^\infty(S_{\ub,u'})}.
\end{align*}
We can integrate \eqref{D omegab minus omegab m0} to derive
\begin{align*}
\|\omegab-\omegab_{m_0}\|_{L^\infty(S_{\ub,u})}&\lesssim\|\omegab-\omegab_{m_0}\|_{L^\infty(S_{\delta,u})}+\int_0^{\ub}\|\Omega-\Omega_{m_0}\|_{L^\infty(S_{\ub',u})}+\mathscr{R}^{\ub}(u)\\
&\lesssim\int_0^{\ub}\sup_{u_0\le u'\le u}\|\omegab-\omegab_{m_0}\|_{L^\infty(S_{\ub',u'})}+\mathscr{R}^{\ub}(u).
\end{align*}
Therefore, Gronwall's inequality yields
\begin{align*}
\sup_{u_0\le u'\le u}\|\omegab-\omegab_{m_0}\|_{L^\infty(S_{\ub,u'})}\lesssim\sup_{u_0\le u'\le u}(\mathscr{O}(\delta,u')+\mathscr{R}^{\ub}(u')).
\end{align*}
As a byproduct, we also have
\begin{align*}
\|\Omega-\Omega_{m_0}\|_{L^\infty(S_{\ub,u})}\lesssim\sup_{u_0\le u'\le u}(\mathscr{O}(\delta,u')+\mathscr{R}^{\ub}(u')).
\end{align*}
We then rewrite (\ref{NSE_Db_omega}) as
\begin{align*}
\Db(\omega-\omega_{m_0})=(\Omega^2-\Omega_{m_0}^2)(2(\eta,\etab)-|\etab|^2-\rho)+\Omega_{m_0}^2(\rho-\rho_{m_0}).
\end{align*}
We can integrate this equations and use the bounds for $|\Omega-\Omega_{m_0}|$ to derive
\begin{align*}
\|\omega-\omega_{m_0}\|_{L^\infty(S_{\ub,u})}\lesssim\sup_{u_0\le u'\le u}(\mathscr{O}(\delta,u')+\mathscr{R}^{\ub}(u'))+\underline{\mathscr{R}}^u(\ub).
\end{align*}
We turn to the derivatives of $\omega$ and $\omegab$. Since $\omega_{m_0}$ and $\omegab_{m_0}$ are constants on $S_{\ub,u}$, applying $\nablas^k$ on (\ref{NSE_D_omegab}) and (\ref{NSE_Db_omega}) will kill those terms. We also notice that $\rho_{m_0}$ is constant on $S_{\ub,u}$ and $\ds\log\Omega=(\eta+\etab)/2$ has already been bounded by $\mathscr{P}$. We then take derivatives
\begin{align*}
\Db\nablas^k\omega&=\sum_{i=0}^{k-2}\nablas^i\nablas(\Omega\chib)\cdot\nablas^{k-2-i}\ds\omega+\nablas^k(\Omega^2(2(\eta,\etab)-|\etab|^2-\rho)),\\
\Db\nablas^k\omega&=\sum_{i=0}^{k-2}\nablas^i\nablas(\Omega\chib)\cdot\nablas^{k-2-i}\ds\omega+\nablas^k(\Omega^2(2(\eta,\etab)-|\etab|^2-\rho)).
\end{align*}
For $k=1,2,3$, we can integrate as before to derive
\begin{align*}
\|\ds(\omega,\omegab)\|_{L^\infty(S_{\ub,u})}+\|\nablas^2(\omega,\omegab)\|_{L^4(S_{\ub,u})}+\|\nablas^3(\omega,\omegab)\|_{L^2(S_{\ub,u})}\lesssim\mathscr{P}(\ub,u).
\end{align*}

Finally, we consider $\tr\chi'$ and $\tr\chib'$. We rewrite (\ref{NSE_D_trchi}) as
\begin{align*}
D(\tr\chi'-\tr\chi'_{m_0})=-\frac{1}{2}(\tr\chi+\tr\chi_{m_0})(\Omega\tr\chi'-\Omega_{m_0}\tr\chi'_{m_0})-|\chih|^2.
\end{align*}
We can bound $|\Omega\tr\chi-\Omega_{m_0}\tr\chi'_{m_0}|\lesssim|\Omega-\Omega_{m_0}|+|\tr\chi'-\tr\chi'_{m_0}|$. We then integrate the above equation to derive
\begin{align*}
\|\tr\chi'-\tr\chi'_{m_0}\|_{L^\infty(S_{\ub,u})}\lesssim&\|\tr\chi'-\tr\chi'_{m_0}\|_{L^\infty(S_{\delta,u})}+\int_\delta^{\ub}\|\tr\chi'-\tr\chi'_{m_0}\|_{L^\infty(S_{\ub',u})}\\
&+\sup_{u_0\le u'\le u}(\mathscr{O}(\delta,u')+\mathscr{R}^{\ub}(u')).
\end{align*}
Thanks to Gronwall's inequality, we have
\begin{align*}
\|\tr\chi'-\tr\chi'_{m_0}\|_{L^\infty(S_{\ub,u})}\lesssim\sup_{u_0\le u'\le u}(\mathscr{O}(\delta,u')+\mathscr{R}^{\ub}(u')).
\end{align*}
In a similar manner. we have the following bound for $\tr\chib'$:
\begin{align*}
\|\tr\chib'-\tr\chib'_{m_0}\|_{L^\infty(S_{\ub,u})}\lesssim\mathscr{O}(\ub,u_0)+\underline{\mathscr{R}}^u(\ub)+\sup_{u_0\le u'\le u}(\mathscr{O}(\delta,u')+\mathscr{R}^{\ub}(u')).
\end{align*}
We turn to the derivatives of $\tr\chi'$ and $\tr\chib'$. We apply $\ds$ to (\ref{NSE_D_trchi}) to derive
\begin{align*}
D\ds\tr\chi'=-\tr\chi\ds(\Omega\tr\chi')-\ds(|\chih|^2)=-\Omega\tr\chi\ds\tr\chi'-\tr\chi\tr\chi'\ds\Omega-\ds(|\chih|^2).
\end{align*}
The last two terms can be bounded by $\mathscr{P}(\ub,u)$ and the first term can be absorbed by Gronwall's inequality. Thus, we obtain
\begin{align*}
\|\ds\tr\chi'\|_{L^\infty(\ub,u)}\lesssim\mathscr{P}(\ub,u).
\end{align*}
For higher order derivatives, we have
\begin{align*}
D\nablas^k\tr\chi'=-\Omega\tr\chi\nablas^k\tr\chi'+\lot,
\end{align*}
where $\lot$ can be bounded directly by $\mathscr{P}(\ub,u)$. Thanks to Gronwall's inequality, we obtain
\begin{align*}
\|\nablas^2\tr\chi'\|_{L^4(\ub,u)}+\|\nablas^3\tr\chi'\|_{L^2(\ub,u)}\lesssim\mathscr{P}(\ub,u).
\end{align*}
Similarly, we have the following estimates for $\tr\chib'$:
\begin{align*}
\|\ds\tr\chib'\|_{L^\infty(\ub,u)}+\|\nablas^2\tr\chib'\|_{L^4(\ub,u)}+\|\nablas^3\tr\chib'\|_{L^2(\ub,u)}\lesssim\mathscr{P}(\ub,u).
\end{align*}
Finally, using the bound of $\nablas^i(\Omega-\Omega_{m_0})$ in terms of $\mathscr{P}(\ub,u)$ for $i=0,1,2,3$, the above estimates for $\nablas^i(\tr\chi'-\tr\chi'_{m_0})$ and $\nablas^i(\tr\chib'-\tr\chib'_{m_0})$ also hold for $\nablas^i(\tr\chi-\tr\chi_{m_0})$ and $\nablas^i(\tr\chib-\tr\chib_{m_0})$, or $\nablas^i(\Omega\tr\chi-\Omega_{m_0}\tr\chi_{m_0})$ and $\nablas^i(\Omega\tr\chib-\Omega_{m_0}\tr\chib_{m_0})$, for $i=0,1,2,3$.

Therefore, we have completed the proof of Lemma \ref{connectionb}.
\end{proof}

We rewrite (\ref{NBE_D_rho}) and (\ref{NBE_Db_rho}) in terms of the renormalized quantities such as $\rho-\rho_{m_0}$. Recall that $\rho_{m_0}$ satisfies the following equations
\begin{align*}
D\rho_{m_0}+\frac{3}{2}\Omega_{m_0}\tr\chi_{m_0}\rho_{m_0}=0, \quad \Db\rho_{m_0}+\frac{3}{2}\Omega_{m_0}\tr\chib_{m_0}\rho_{m_0}=0.
\end{align*}
We can subtract these two equations from (\ref{NBE_D_rho}) and (\ref{NBE_Db_rho}) to derive
\begin{align}
\label{nNBE_D_rho-rho0}
D(\rho-\rho_{m_0})+\frac{3}{2}(\Omega\tr\chi-\Omega_{m_0}\tr\chi_{m_0})\rho&+\frac{3}{2}\Omega_{m_0}\tr\chi_{m_0}(\rho-\rho_{m_0})\\\nonumber&-\Omega\{\divs \beta+(2\etab+\zeta,\beta)-\frac{1}{2}(\chibh,\alpha)\}=0,\\
\label{nNBE_Db_rho-rho0}
\Db(\rho-\rho_{m_0})+\frac{3}{2}(\Omega\tr\chib-\Omega_{m_0}\tr\chib_{m_0})\rho&+\frac{3}{2}\Omega_{m_0}\tr\chib_{m_0}(\rho-\rho_{m_0})\\\nonumber&+\Omega\{\divs \betab+(2\eta-\zeta,\betab)-\frac{1}{2}(\chih,\alphab)\}=0.
\end{align}
Because $\rho_{m_0}$ is constant on each $S_{\ub,u}$, we can also rewrite (\ref{NBE_Db_beta}) and (\ref{NBE_D_betab}) as
\begin{align}
\label{nNBE_Db_beta}\Db\beta+\frac{1}{2}\Omega\tr\chib\beta-\Omega\chibh \cdot \beta+\omegab \beta-\Omega\{\ds (\rho-\rho_{m_0})+{}^*\ds \sigma+3\eta\rho+3{}^*\eta\sigma+2\chih\cdot\betab\}&=0,\\
\label{nNBE_D_betab}D\betab+\frac{1}{2}\Omega\tr\chi\betab-\Omega\chih \cdot \betab+\omega \betab+\Omega\{\ds (\rho-\rho_{m_0})-{}^*\ds \sigma+3\etab\rho-3{}^*\etab\sigma-2\chibh\cdot\beta\}&=0.
\end{align}

We call \eqref{nNBE_D_rho-rho0}, \eqref{nNBE_Db_rho-rho0}, \eqref{nNBE_Db_beta} and \eqref{nNBE_D_betab} together with the following six original Bianchi equations \eqref{NBE_Db_alpha}, \eqref{NBE_D_alphab}, \eqref{NBE_D_beta}, \eqref{NBE_Db_betab}, \eqref{NBE_D_sigma} and \eqref{NBE_Db_sigma} \emph{the renormalized Bianchi equations}. In a similar manner as we derived energy inequalities (\ref{alpha}), (\ref{beta}), (\ref{rhosigma}) and (\ref{betabalphab}), we have the following energy inequalities
\begin{align*}
\sum_{R\in\{\alpha,\beta,\rho-\rho_{m_0},\sigma,\betab\}}&\int_{C_u^{\ub}}|R|^2+\sum_{\underline{R}\in\{\beta,\rho-\rho_{m_0},\sigma,\betab,\alphab\}}\int_{\Cb_{\ub}^u}|\underline{R}|^2\\
\lesssim\sum_{R\in\{\alpha,\beta,\rho-\rho_{m_0},\sigma,\betab\}}&\int_{C_{u_0}^{\ub}}|R|^2+\sum_{\underline{R}\in\{\beta,\rho-\rho_{m_0},\sigma,\betab,\alphab\}}\int_{\Cb_{\delta}^u}|\underline{R}|^2\\
+\sum_{R\in\{\alpha,\beta,\rho-\rho_{m_0},\sigma,\betab\}\atop { \underline{R}\in\{\beta,\rho-\rho_{m_0},\sigma,\betab,\alphab\}\atop
\mathscr{R}_1,\mathscr{R}_2\in\{\alpha,\beta,\rho-\rho_{m_0},\sigma,\betab,\alphab\}}}\iint_{M_{\ub,u}}&|\Omega\Gamma\cdot\mathscr{R}_1\cdot\mathscr{R}_2|\\
+&|\rho||(\Omega\tr\chi-\Omega_{m_0}\tr\chi_{m_0})(\rho-\rho_{m_0})+\chih\cdot\alpha+\eta\cdot\beta+\chibh\cdot\alphab+\etab\cdot\betab|.
\end{align*}
where $M_{\ub,u}$ refers to the region corresponding to $\delta\le\ub'\le\ub$, $u_0\le u'\le u$. We also have
\begin{align*}
\sum_{R\in\{\alpha,\beta,\rho-\rho_{m_0},\sigma,\betab\}}&\int_{C_u^{\ub}}|\nablas^i R|^2+\sum_{\underline{R}\in\{\beta,\rho-\rho_{m_0},\sigma,\betab,\alphab\}}\int_{\Cb_{\ub}^u}|\nablas^i\underline{R}|^2\\
\lesssim\sum_{R\in\{\alpha,\beta,\rho-\rho_{m_0},\sigma,\betab\}}&\int_{C_{u_0}^{\ub}}|\nablas^i R|^2+\sum_{\underline{R}\in\{\beta,\rho-\rho_{m_0},\sigma,\betab,\alphab\}}\int_{\Cb_{\delta}^u}|\nablas^i\underline{R}|^2\\
+\sum_{R\in\{\alpha,\beta,\rho-\rho_{m_0},\sigma,\betab\}\atop { \underline{R}\in\{\beta,\rho-\rho_{m_0},\sigma,\betab,\alphab\}\atop
\mathscr{R}_1,\mathscr{R}_2\in\{\alpha,\beta,\rho-\rho_{m_0},\sigma,\betab,\alphab\}}}\iint_{M_{\ub,u}}&\sum_{j=0}^{i-1}|\nablas^j\nablas(\Omega\chi)\cdot\nablas^{i-1-j}\Rb\cdot\nablas^i \Rb+\nablas^j\nablas(\Omega\chib)\cdot\nablas^{i-1-j} R\cdot\nablas^i R|\\
&+\sum_{j=0}^{i-1}|\nablas^jK\cdot\nablas^{i-1-j}R\cdot\nablas^i R+\nablas^jK\cdot\nablas^{i-1-j}\Rb\cdot \nablas^i\Rb|\\
&+|\nablas^i((\Omega\Gamma)\cdot\mathscr{R}_1)\cdot\nablas^i\mathscr{R}_2|+|\rho||\nablas^i(\Omega\tr\chi-\Omega_{m_0}\tr\chi_{m_0})\cdot\nablas^i\rho|\\
&+|\rho||\nablas^i\chih\cdot\nablas^i\alpha+\nablas^i\eta\cdot\nablas^i\beta+\nablas^i\chibh\cdot\nablas^i\alphab+\nablas^i\etab\cdot\nablas^i\betab|,
\end{align*}
for $i=1,2,3$.

We remark that in terms $\mathscr{R}_1\cdot\mathscr{R}_2$, the term $\alpha\cdot\alphab$ do not appear. Thus, we can regard $\mathscr{R}_1\cdot\mathscr{R}_2$ as either $R_1\cdot R_2$ or $\Rb_1\cdot\Rb_2$. We then have the following estimates:
\begin{align*}
\sum_{j=0}^1\iint_{M_{\ub,u}}|\nablas^j(\Omega\Gamma)\cdot \nablas^{i-j} R_1|^2\lesssim\int_{u_0}^u\int_{C_{u'}^{\ub}}\sum_{j=0}^{i}|\nablas^jR_1|^2\ \ \textrm{for}\ \ i\le3,
\end{align*}
\begin{align*}
\iint_{M_{\ub,u}}|\nablas^2(\Omega\Gamma)\cdot \nablas^iR_1|^2&\lesssim\int_{u_0}^u\int_{\delta}^{\ub}\|\nablas^iR_1\|_{L^4(\ub',u')}^2\lesssim\int_{u_0}^u\int_{\delta}^{\ub}\sum_{j=0}^{i+1}\|\nablas^jR_1\|_{L^2(\ub',u')}^2\\
&\lesssim\int_{C_{u'}^{\ub}}\sum_{j=0}^{i+1}|\nablas^j R_1|^2\ \ \textrm{for}\ \ i\le1,
\end{align*}
\begin{align*}
\iint_{M_{\ub,u}}|\nablas^3(\Omega\Gamma)\cdot R_1|^2&\lesssim\int_{u_0}^u\int_{\delta}^{\ub}\|R_1\|_{L^\infty(\ub',u')}^2 \lesssim\int_{u_0}^u\int_{\delta}^{\ub}\sum_{j=0}^{2}\|\nablas^jR_1\|_{L^2(\ub',u')}^2\\
&\lesssim\int_{C_{u'}^{\ub}}\sum_{j=0}^{2}|\nablas^j R_1|^2.
\end{align*}
Therefore, by Cauchy-Schwartz inequality, we have
\begin{align*}
\sum_{i=0}^3\iint_{M_{\ub,u}}|\nablas^i(\Omega\Gamma\cdot R_1)\cdot \nablas^i R_2|\lesssim\sum_{i=0}^3\int_{u_0}^u\int_{C_{u'}^{\ub}}(|\nablas^iR_1|^2+|\nablas^iR_2|^2).
\end{align*}
Similarly, we have
\begin{align*}
\sum_{i=0}^3\iint_{M_{\ub,u}}|\nablas^i(\Omega\Gamma\cdot \Rb_1)\cdot \nablas^i \Rb_2|\lesssim\sum_{i=0}^3\int_{\delta}^{\ub}\int_{\Cb_{\ub'}^u}(|\nablas^i\Rb_1|^2+|\nablas^i\Rb_2|^2).
\end{align*}

We turn to the following terms
\begin{align*}
\sum_{i=0}^2\iint_{M_{\ub,u}}|\nablas^i(\nablas(\Omega\chi)\cdot\Rb)\cdot\nablas \Rb+\nablas^i(\nablas(\Omega\chib)\cdot R)\cdot\nablas R|+|\nablas^i(KR\cdot\nablas R+K\Rb\cdot \Rb)|.
\end{align*}
They are treated in the same way as above by virtue of the bounds on $K$. They are eventually bounded by
\begin{align*}
\sum_{i=0}^3\left(\int_{u_0}^u\int_{C_{u'}^{\ub}}|\nablas^iR|^2+\int_{\delta}^{\ub}\int_{\Cb_{\ub'}^u}|\nablas^i\Rb|^2\right).
\end{align*}

Finally, we consider the following terms:
\begin{align*}
\sum_{i=0}^3\iint_{M_{\ub,u}} &|\rho||\nablas^i(\Omega\tr\chi-\Omega_{m_0}\tr\chi_{m_0})\cdot\nablas^i(\rho-\rho_{m_0})\\
&+\nablas^i\chih\cdot\nablas^i\alpha+\nablas^i\eta\cdot\nablas^i\beta+\nablas^i\chibh\cdot\nablas^i\alphab+\nablas^i\etab\cdot\nablas^i\betab|.
\end{align*}
According to Lemma \ref{connectionb}, they are bounded by
\begin{align*}
\sup&_{\delta\le\ub\le\delta+\epsilon_0}\mathscr{O}(\ub,u_0)^2 +\sup_{u_0\le u\le -1-\ub}\mathscr{O}(\delta,u)^2\\
&+\sum_{i=0}^3\left(\int_{u_0}^u\sup_{u_0\le u''\le u'}\int_{C_{u''}^{\ub}}|\nablas^iR|^2+\int_{\delta}^{\ub}\sup_{\delta\le u''\le u'}\int_{\Cb_{\ub''}^{u}}|\nablas^i\Rb|^2\right).
\end{align*}

We define
\begin{align*}
\mathcal{E}(u)=\sum_{R\in\{\alpha,\beta,\rho-\rho_{m_0},\sigma,\betab\}}\sup_{u_0\le u'\le u}\sum_{i=0}^3\int_{C_{u'}}|\nablas^iR|^2,\\
\mathcal{F}(\ub)=\sum_{\Rb\in\{,\beta,\rho-\rho_{m_0},\sigma,\betab,\alphab\}}\sup_{\delta\le \ub'\le\ub}\sum_{i=0}^3\int_{\Cb_{\ub}}|\nablas^i\Rb|^2.
\end{align*}
Therefore, the above estimates can be summarized as
\begin{align*}
\mathcal{E}(u)+\mathcal{F}(\ub)&\lesssim \mathcal{E}(u_0)+\mathcal{F}(\delta) \\
&+\sup_{\delta\le\ub\le\delta+\epsilon_0}\mathscr{O}(\ub,u_0)^2+\sup_{u_0\le u\le -1-\ub}\mathscr{O}(\delta,u)^2+\left(\int_{u_0}^u\mathcal{E}(u')+\int_{\delta}^{\ub}\mathcal{F}(\ub')\right).
\end{align*}
Thanks to Gronwall's inequality, we have proved:
\begin{proposition}\label{smallness}
If we have the following smallness on initial data
\begin{align*}
\mathscr{R}^{1+\delta}(u_0),\underline{\mathscr{R}}^{-1-\delta}(\delta),\mathscr{O}(\ub,u_0),\mathscr{O}(\delta,u)\le\epsilon,
\end{align*}
then for  $(\ub,u) \in [\delta, \delta+\epsilon_0] \times [u_0, u_0+\epsilon_0]$, we have
\begin{align*}
\mathscr{R}^{\ub}(u),\underline{\mathscr{R}}^u(\ub),\mathscr{O}(\ub,u),\mathscr{O}(\ub,u)\lesssim\epsilon.
\end{align*}
\end{proposition}
The smallness of the data is of course guaranteed by Proposition \ref{initialb}.

Equipped with Proposition \ref{smallness}, we are able to control $C^k$ norms. More precisely, we shall prove that
\begin{equation}\label{control the C2 norm}
\|g-g_{m_0}\|_{C^2(M_{\epsilon_0},g_{m_0})}\lesssim\epsilon.
\end{equation}

We start with the $C^0$ norms. We write $g$ in canonical double null coordinates as follows
\begin{align*}
g=-2\Omega^2(\D\ub\otimes\D u+\D u\otimes\D \ub)+\gs_{AB}(\D\theta^A-b^A\D\ub)\otimes(\D\theta^B-b^B\D\ub).
\end{align*}
Since $\Db(\Omega^2)=2\Omega^2\omegab$, $\Db\gs=2\Omega\chib$ and $\Db b=4\Omega^2\zeta$, by virtue of the $L^\infty$ bounds of $\omegab-\omegab_{m_0}$, $\chib-\chib_{m_0}$ and $\zeta$ in Proposition \ref{smallness} as well as the Sobolev inequalities, we have
\begin{equation*}
\|g-g_{m_0}\|_{C^0(M_{\epsilon_0},g_{m_0})}\lesssim\epsilon.
\end{equation*}

For $C^1$ norms of $g-g_{m_0}$, we use $\nabla_{m_0}(g-g_{m_0})=(\nabla-\nabla_{m_0})g$. We also need the $L^\infty$ bounds on $\nabla-\nabla_{m_0}$, which can be obtained by $\Gamma-\Gamma_{m_0}$ and $\Gammas-\Gammas_{m_0}$ where $\Gamma$ refers to the null connection coefficients and $\Gammas$ refers to the Christoffel symbols of $\gs$. We have already derived
\begin{equation*}
|\Gamma-\Gamma_{m_0}|\lesssim\epsilon
\end{equation*}
thanks to Proposition \ref{smallness} and the Sobolev inequalities. To estimate $\Gammas-\Gammas_{m_0}$, we need the following propagation equation
\begin{align*}
\Db(\Gammas-\Gammas_{m_0})_{AB}^C=\frac{1}{2}\gs^{CD}(\nablas_A(\Omega\chib)_{BD}+\nablas_B(\Omega\chib)_{AD}-\nablas_D(\Omega\chib)_{AB}).
\end{align*}
Since we have already obtained the $L^\infty$ bounds on $\nablas(\Omega\chib)$,  by directly integrating the above equation, we obtain $|\Gammas-\Gammas_{m_0}|\lesssim\epsilon$. Finally, we have
\begin{equation*}
\|g-g_{m_0}\|_{C^1(M_{\epsilon_0},g_{m_0})}\lesssim\epsilon.
\end{equation*}

For $C^2$ bounds on $g-g_{m_0}$, we write
\begin{equation*}
\nabla_{m_0}^2(g-g_{m_0})=(\nabla^2-\nabla_{m_0}^2)g=\nabla((\nabla-\nabla_{m_0})g)+(\nabla-\nabla_{m_0})\nabla_{m_0}g.
\end{equation*}
The last term $(\nabla-\nabla_{m_0})\nabla_{m_0}g$ has already been controlled. Thus, we need the $L^\infty$ bounds of the following quantities:
\begin{align*}
D(\Gamma-\Gamma_{m_0}),\ D(\Gammas-\Gammas_{m_0}),\ \Db(\Gamma-\Gamma_{m_0}),\ \Db(\Gammas-\Gammas_{m_0}),\ \nablas(\Gamma-\Gamma_{m_0}),\ \nablas(\Gammas-\Gammas_{m_0}).
\end{align*}
The estimates for the first four quantities can be obtained using the null structure equations provided we have $L^\infty$ bounds of all first derivatives of null curvature components, but we do not have control on
\begin{align*}
D(\omega-\omega_{m_0}),\ \Db(\omegab-\omegab_{m_0});
\end{align*}
The estimates for the fifth quantity are obtained directly by Proposition \ref{smallness}; The bounds for the last quantity are obtained also by the propagation equation of $\Gammas-\Gammas_{m_0}$ provided we have $L^\infty$ bound of
\begin{equation*}
\nablas^2(\Omega\chib).
\end{equation*}

The strategy is clear now: we have to control $\|(D(\omega-\omega_{m_0})\|_{L^\infty}$, $\|\Db(\omegab-\omegab_{m_0})\|_{L^\infty}$ and $\|\nablas^2(\Omega\chib) )\|_{L^\infty}$. For this purpose, we first claim that
\begin{equation*}
\|\alpha,\beta,\rho-\rho_{m_0},\sigma,\betab,\alphab\|_{L^\infty}\lesssim\epsilon.
\end{equation*}

This can be obtained by using a variation of (\ref{SobolevC_L2_L4}) and (\ref{SobolevCb_L2_L4}) by setting the constants depending on $m_0$, $\epsilon_0$, $u_0$ and the fact that $\|\nablas^{\le1}\alpha\|_{L^4(S_{\delta,u})}\lesssim\epsilon$ which is stated in Proposition \ref{initialb}. To apply (\ref{SobolevC_L2_L4}) and (\ref{SobolevCb_L2_L4}), besides $\mathscr{R}_2, \underline{\mathscr{R}}_2\lesssim\epsilon$, we also need $\|\nablas D\alpha\|_{L^2(C_u)}\lesssim\epsilon$ and $\|\nablas\Db(\beta,\rho-\rho_{m_0},\sigma,\betab, \alphab)\|_{L^2(\Cb_{\ub})}\lesssim\epsilon$. For $\beta$, $\rho-\rho_{m_0}$, $\sigma$ and $\betab$ this is done by simply taking $\nablas$ to both sides of the renormalized Bianchi equations. For $\|\nablas D\alpha\|_{L^2(C_u)}$ and $\|\nablas\Db\alphab\|_{L^2(\Cb_{\ub})}$, they are obtained using a variation of (\ref{D_alpha}) and (\ref{Db_alphab}) by replacing Lie derivatives by $\nablas$.

We then claim
\begin{equation*}
\|\nablas(\beta,\betab)\|_{L^\infty} \lesssim \epsilon.
\end{equation*}

The proof is similar to the above, we first apply $\nablas^2$ to (\ref{NBE_Db_beta}) and (\ref{NBE_Db_betab}) and then apply (\ref{SobolevCb_L2_L4}) and Sobolev inequalities on $S_{\ub,u}$. This yields $\|\nablas(\beta,\betab)\|_{L^\infty} \lesssim \epsilon$. As a consequence, we have
\begin{equation*}
\|D(\omega-\omega_{m_0}),\Db(\omegab-\omegab_{m_0})\|_{L^\infty}\lesssim\epsilon.
\end{equation*}
In fact, we can apply $D$ to(\ref{NSE_Db_omega}) to derive a propagation equation for $D(\omega-\omega_{m_0})$ which looks like
\begin{equation*}
\Db D(\omega-\omega_0) =  D(\rho-\rho_0) + \lot = \nablas\beta + \lot.
\end{equation*}
We then integrate this equation to conclude. Similarly, we have estimates for $\Db(\omegab-\omegab_{m_0})$.

To prove $\|\nablas^2(\Omega\chib)\|_{L^\infty}\lesssim\epsilon$,  using bounds of $\|\nablas^3\betab\|_{L^2(\Cb_{\ub})}$ given by Proposition \ref{smallness}, we couple (\ref{NSE_div_chibh}) with (\ref{NSE_Db_trchib}) to derive bounds of $\|\nablas^4\chib\|_{L^2(\Cb_{\ub})}$ and $\|\nablas^4\tr\chib\|_{L^2(S_{\ub,u})}$.

Therefore, we have proved
\begin{equation*}
\|g-g_{m_0}\|_{C^2(M_{\epsilon_0},g_{m_0})}\lesssim\epsilon.
\end{equation*}

For higher order estimates on $C^k$ norms, as we remarked before, we can simply use a routine induction argument and we omit the proof. Therefore, we have completed the proof of Theorem \ref{Schwarzschild}.
\end{proof}

\section{Gluing Construction}\label{Gluing Construction}

\subsection{Preparation}
We summarize some key properties of the metric $g$ constructed in the previous sections:
\begin{enumerate}
\item In the region $(\ub, u) \in [u_0,0] \times [u_0, 0]$, $g$ coincides with the standard Minkowski metric;
\item In the region where $\ub \in [0,\delta]$ and $u \in [u_0, -1-\ub]$, the metric $g$ is constructed by Christodoulou in his work \cite{Chr}. In addition, we have imposed condition \eqref{integral=m0} with $|u_0|>2m_0$, so that on the incoming cone $\Cb_{\delta}$, the estimates in Section \ref{Section Geometry on incoming cone} hold. For a sufficiently small $\delta$, $\tr\chi(\delta,u_0)>0$ and hence $S_{\delta,u_0}$ is not a trapped surface, and if $2m_0>1$, according to Theorem 17.1 in \cite{Chr}, $S_{\delta,-1-\delta}$ is a trapped surface.

    We remark that, for $0<m_0<\frac{1}{2}$, it is straightforward from the proof of \cite{Chr} that $g$ actually exists up to $\ub+u=-2m_0$ and $S_{\delta,-2m_0-\delta}$ is a trapped surface, once we choose a sufficiently small $\delta$.
\item In the region $(\ub, u) \in [\delta,\delta+\epsilon_0] \times [u_0,u_0+\epsilon_0]$, the metric $g$ is $\epsilon$-close (in $C^{k-3}$ norms) to the Schwarzschild metric $g_{m_0}$.
\end{enumerate}

We fix a sphere $S_1=S_{\delta,u_1}$ near $S_{\delta,u_0}$ by choosing $u_1$ close to $u_0$ in such a way that $\tr\chi(\delta,u_1)>0$ and $\tr\chi_{m_0}(\delta,u_1)>0$. We emphasize that the choice of $S_1$ does not depend on $\delta$ if $\delta$ is sufficiently small. Because $g_{m_0}$ is the Schwarzschild metric with mass $m_0$, we can choose the time function $t$ of $g_{m_0}$ in the Boyer-Lindquist coordinates (we regard the Schwarzschild metric as a member in the Kerr family) such that $t(\delta,u_1)=0$. We also regard $t$ as a smooth function in the $\ub$-$u$ plane  $(\ub, u) \in [\delta,\delta+\epsilon_0] \times [u_0,u_0+\epsilon_0]$.

Now in the region $(\ub, u) \in [\delta,\delta+\epsilon_0] \times [u_0,u_0+\epsilon_0]$, we choose a three dimensional hypersurface
\begin{equation*}
H=\sum_{t(\ub,u)=0}S_{\ub,u},
\end{equation*}
and we use $S_2$ to denote the intersection of $H$ with $C_{u_0}\cup\Cb_{\delta+\epsilon_0}$. We define the following two radii $r_1=r|_{S_1}$ and $r_2=r|_{S_2}$ where $r(\ub,u)=r_{m_0}(\ub,u)$ is the radius of $S_{\ub,u}$ in Schwarzschild metric $g_{m_0}$. By construction, $2m_0<r_1<r_2$ and $H$ is space-like with respect to $g_{m_0}$. Thus, if $\epsilon$ is sufficiently small, $H$ is also space-like with respect to the metric $g$.

Let $(\bar{g}_{m_0},\bar{k}_{m_0})$ be the induced metric and the second fundamental form of $H$ as a submanifold in the Schwarzschild space-time with mass ${m_0}$. Therefore, for $r_1\le r\le r_2$, we have
\begin{equation*}
\bar{g}_{m_0}=(1-\frac{2m_0}{r})^{-1}\mathrm{d}r^2+r^2\mathrm{d}\sigma_{S^2}^2,
\end{equation*}
and
\begin{equation*}
\bar{k}_{m_0} \equiv 0.
\end{equation*}
 Let $(\bar{g},\bar{k})$ be the metric and the second fundamental form of $H$ induced by $g$, since $g$ is $\epsilon$-close to the Schwarzschild metric $g_{m_0}$ in $C^{k-3}$ norms, we have
\begin{equation*}
\|\bar{g}-\bar{g}_{m_0}\|_{C^{k-3}(H,\bar{g}_{m_0})}\lesssim\epsilon,
\end{equation*}
and
\begin{equation*}
\|\bar{k}\|_{C^{k-4}(H,\bar{g}_{m_0})}\lesssim\epsilon.
\end{equation*}

We recall the form of the Kerr metric in the Boyer-Lindquist coordinates
\begin{align*}
g_{m,(0,0,a)}=&(-1+\frac{2mr}{r^2+a^2\cos^2\theta})\D t^2-\frac{2mra\sin^2\theta}{r^2+a^2\cos^2\theta}\D t\D\varphi+\frac{r^2+a^2\cos^2\theta}{r^2-2ma+a^2}\D r^2\\
&+(r^2+a^2\cos^2\theta)^2\D\theta^2+\sin^2\theta(r^2+a^2+\frac{2mra^2\sin^2\theta}{r^2+a^2\cos^2\theta})\D\varphi^2.
\end{align*}
where $m > a \geq 0$. The lower index $(0,0,a)$ specifies the angular momentum vector: if we use the spherical polar coordinates $x^1=r\cos\varphi\sin\theta$, $x^2=r\sin\varphi\sin\theta$ and $x^3=r\cos\theta$, then the vector $(0,0,a)$ is the same as the axis (of the rotation of the Kerr black hole) and the norm of $(0,0,a)$ is the angular momentum $a$. For an arbitrary vector $\mathbf{a}=(a_1,a_2,a_3)\in\mathbb{R}^3$ with $|\mathbf{a}|=a$  and an isometry $\Omega_{\mathbf{a}}\in SO(3)$ mapping $\mathbf{a}$ to $(0,0,a)$, we can also define a family of Kerr metric $g_{m,\mathbf{a}}=(\mathrm{id}_{\mathbb{R}}\times\Omega_{\mathbf{a}})^* \,g_{m,(0,0,a)}$ where $\mathrm{id}_{\mathbb{R}}$ is the identity map of $t$ axis. We remark that this definition does not depend on the choice of $\Omega_{\mathbf{a}}$.

We choose the slice $H$ to be $t=0$ in the above Kerr space-time and we use $(x^1,x^2,x^3)$ and $(r, \theta, \varphi)$ as coordinates on $H$. Let  $(\bar{g}_{m,(0,0,a)},\bar{k}_{m,(0,0,a)})$ be the induced metric and the second fundamental form of $H$, thus, we have
\begin{align*}
\bar{g}_{m,(0,0,a)}=\frac{r^2+a^2\cos^2\theta}{r^2-2mr+a^2}\mathrm{d}r^2+(r^2+a^2\cos^2\theta)\mathrm{d}\theta^2+(r^2+a^2+\frac{2mra^2\sin\theta}{r^2+a^2\cos^2\theta})\sin^2\theta\mathrm{d}\varphi,
\end{align*}
and
\begin{align*}
\bar{k}_{m,(0,0,a)}=\frac{2}{r^2+a^2+\frac{2mra^2\sin^2\theta}{r^2+a^2\cos^2\theta}}&(\frac{ma}{r^2+a^2\cos^2\theta}(r^2-a^2)\mathrm{d}r\mathrm{d}\varphi\\
&+\frac{2m^2r^2a^3\sin2\theta\sin^2\theta}{(r^2+a^2\cos^2\theta)^2}\mathrm{d}\theta\mathrm{d}\varphi).
\end{align*}
For an arbitrary vector $\mathbf{a}=(a_1,a_2,a_3)\in\mathbb{R}^3$ with $|\mathbf{a}|=a$, we can also define the corresponding of initial data set by $\bar{g}_{m,\mathbf{a}}=\Omega_{\mathbf{a}}^*\bar{g}_{m,(0,0,a)}$ and $\bar{k}_{m,\mathbf{a}}=\Omega_{\mathbf{a}}^*\bar{k}_{m,(0,0,a)}$.  They also correspond to the $t=0$ slices.

We now recall the definition for the constraint map $\Phi$, that is,
$$\Phi(\bar{g},\bar{\pi})=(\mathcal{H}(\bar{g},\bar{\pi}),\mathrm{div}\bar{\pi}),$$
where $\bar{\pi}=\bar{k}-\tr \bar{k}\bar{g}$  and  $\mathcal{H}(\bar{g},\bar{\pi})=R(\bar{g})+\frac{1}{2}(\tr\bar{\pi})^2-|\bar{\pi}|^2$. Let $D\Phi^*_{(\bar{g},\bar{\pi})}$ be the formal $L^2$-adjoint of the linearized operator of $\Phi$ at $(\bar{g},\bar{\pi})$. When $(\bar{g},\bar{\pi}) = (\bar{g}_{m_0}, 0)$, we use $K$ to denote the null space $\mathrm{Ker}D\Phi^*_{(\bar{g}_{m_0},0)}$. We will explicitly write down $K$. In $(x^1,x^2,x^3)$ coordinates, we use $\Omega_i (i=1,2,3)$ to denote the following vectors fields
$$\Omega_1=x^2\frac{\partial}{\partial x_3}-x^3\frac{\partial}{\partial x_2}, \Omega_2=x^3\frac{\partial}{\partial x_1}-x^1\frac{\partial}{\partial x_3},\Omega_3=x^1\frac{\partial}{\partial x_2}-x^2\frac{\partial}{\partial x_1},$$
and in $(r,\theta,\varphi)$ coordinates, they can be written as
\begin{align}\label{Omega}
\Omega_1=-\sin\varphi\partial_{\theta}-\cos\varphi\cot\theta\partial_{\varphi},\ \Omega_2=\cos\varphi\partial_{\theta}-\sin\varphi\cot\theta\partial_\varphi,\ \Omega_3=\partial_{\varphi}.
\end{align}
Then we have the following description on $K$:
\begin{lemma}
$K$ is a $4$ dimensional real vector space. In fact, $K$ is spanned by the following four elements: $((1-\frac{2m_0}{r})^{\frac{1}{2}},0)$,$(0,\Omega_1)$, $(0,\Omega_2)$ and $(0,\Omega_3)$.
\end{lemma}

\begin{proof}
We first recall the form of the linear operator $D\Phi_{(\bar{g},\bar{\pi})}^*$ (see also Section \ref{The Work of Corvino-Schoen}). For a smooth function $f$ and a vector field $X$, we have $D\Phi_{(\bar{g},\bar{\pi})}^*(f,X)=D\mathcal{H}_{(\bar{g},\bar{k})}^*(f)+D\mathrm{div}_{(\bar{g},\bar{k})}^*(X)$, where
\begin{equation*}
\begin{split}
D\mathcal{H}^*_{(\bar{g},\bar{\pi})}(f) &= ((L^*_{\bar{g}}f)_{ij}+(\tr_{\bar{g}}\bar{\pi} \cdot \bar{\pi}_{ij}-2 \bar{\pi}_{ik}\bar{\pi}^{k}{}_{j})f,(\tr_{\bar{g}}\bar{\pi} \cdot \bar{g}_{ij}-2 \bar{\pi}_{ij})f),\\
D \text{div}^*_{(\bar{g},\bar{\pi})}(X)  &=\frac{1}{2}(\mathcal{L}_X \bar{\pi}_{ij}+ \bar{\nabla}_k X^k \pi_{ij}-(X_i (\bar{\nabla}_k \pi^k{}_j + X_j \bar{\nabla}_k \pi^{k}{}_i)\\
&\quad-(\bar{\nabla}_m X_k \pi^{km} + X_k \bar{\nabla}_m \pi^{mk})g_{ij}, \quad -\mathcal{L}_Xg_{ij}),
\end{split}
\end{equation*}
and $L_{\bar{g}}^*(f)=-(\Delta_{\bar{g}}f)\bar{g}+\mathrm{Hess}_{\bar{g}}(f)-f\mathrm{Ric}_{\bar{g}}$. In the current case, $(\bar{g}, \bar{\pi}) = (g_{m_0},0)$, thus,
\begin{equation}\label{equation for D phi dual on Schwarzschild with m0}
D\Phi_{(\bar{g}_{m_0},0)}^*(f,X)=(L_{\bar{g}_{m_0}}^*f,-\frac{1}{2}\mathcal{L}_X\bar{g}_{m_0}).
\end{equation}

First of all, we consider the second equation $\mathcal{L}_X\bar{g}_{m_0}=0$ in \eqref{equation for D phi dual on Schwarzschild with m0}. It amounts to say that $X$ is a Killing vector field on constant $t$ slice for $\bar{g}_{m_0}$. Thus, $X$ is spanned by $\Omega_1$,$\Omega_2$ and $\Omega_3$, i.e. the infinitesimal rotations of the Schwarzschild space-time. We remark that $m_0 \neq 0$ in this case otherwise we may also have translations of the Minkowski space-time.

Secondly, we consider the first equation in \eqref{equation for D phi dual on Schwarzschild with m0}:
\begin{align*}
L_{\bar{g}_{m_0}}^*f=-(\Delta_{\bar{g}_{m_0}}f)\bar{g}_{m_0}+\mathrm{Hess}_{\bar{g}_{m_0}}(f)-f\mathrm{Ric}_{\bar{g}_{m_0}}=0.
\end{align*}
Taking the trace, we obtain $\Delta_{\bar{g}_{m_0}}f=0$. So the above equation is equivalent to
\begin{equation}\label{Hess f equal to f Ric}
\mathrm{Hess}_{\bar{g}_{m_0}}(f)=f\mathrm{Ric}_{\bar{g}_{m_0}}.
\end{equation}
We show that $f$ depends only on $r$. We consider the $(\partial_r,\partial_\theta)$-component of \eqref{Hess f equal to f Ric}, that is, $\partial_r\partial_\theta f-\frac{1}{r}\partial_\theta f=0$.  Thus, $\partial_\theta f=\alpha r$ for some $\alpha\in\mathbb{R}$. We then turn to the $(\partial_r,\partial_r)$-component of \eqref{Hess f equal to f Ric}, it can be written as
\begin{align*}
\partial_r^2f+\frac{m_0}{r^2}(1-\frac{2m_0}{r})^{-1}\partial_rf=(-\frac{2m_0}{r^3}(1-\frac{2m_0}{r})^{-1})f.
\end{align*}
We can take derivative $\partial_\theta$ on both sides and then substitute $\partial_\theta f=\alpha r$ to the above equation. Those operations yield $\frac{\alpha m_0}{r^2}=-\frac{2\alpha m_0}{r^2}$. Therefore, $\alpha=0$ since $m_0\ne0$. We conclude that $\partial_\theta f=0$. We then proceed in the same manner to show $\partial_{\varphi}f=0$. Hence, $f$ depends only on $r$. To obtain the exact form of $f$, we consider the $(\partial_\theta,\partial_\theta)$-component of \eqref{Hess f equal to f Ric}, that is
\begin{align*}
\partial_\theta^2f+r(1-\frac{2m_0}{r})\partial_rf=\frac{m_0}{r}f.
\end{align*}
Since $f$ is simply a function of $r$, we have $\frac{\partial_rf}{f}=\frac{m_0}{r^2}(1-\frac{2m_0}{r})^{-1}$. Therefore, $f(r)=\alpha(1-\frac{2m_0}{r})^{\frac{1}{2}}$ for some $\alpha\in\mathbb{R}$. It is straightforward to verify that $f(r)=\alpha(1-\frac{2m_0}{r})^{\frac{1}{2}}$ is indeed a solution of \eqref{Hess f equal to f Ric} for all $\alpha\in\mathbb{R}$. Hence, we complete the proof of the lemma.
\end{proof}

\subsection{Gluing Construction}
We choose a sufficiently large $k$ sufficient in Theorem \ref{Higherorder} so that we can use the Local Deformation Lemma in Section \ref{The Work of Corvino-Schoen}. To run the gluing construction, we take $\phi=\phi(r)$ a smooth cut-off function on the three slice $H$ such that $\phi \equiv 1$ near $S_1$ and $\phi \equiv 0$ near $S_2$. We define a candidate for the final metric as follows:
\begin{align*}
(\tilde{g},\tilde{k})= (\phi\bar{g}+(1-\phi)\bar{g}_{m,\mathbf{a}},\phi\bar{k} + (1-\phi)\bar{k}_{m,\mathbf{a}}),
\end{align*}
At this stage, $(\tilde{g},\tilde{k})$ may not satisfy the constraint equations \eqref{constraint} and we will use the techniques introduced in Section \ref{The Work of Corvino-Schoen} to deform $(\tilde{g},\tilde{k})$ to be a solution of the constraint equations.

We choose two free parameters $m \in \mathbb{R}_{>0}$ and $\mathbf{a} \in \mathbb{R}^3$, so that
\begin{equation*}
|m-m_0| + |\mathbf{a}|\le C_0\epsilon,
\end{equation*}
where $C_0$ is large positive constant to be determined later. The constant $C_0$ is independent of the choice of $\epsilon$, $m$ and $\mathbf{a}$.

By the construction, it is easy to see that
\begin{align*}
\|\tilde{g}-\bar{g}_{m_0}\|_{C^{k-4,\alpha}_{\rho^{-1}}}\lesssim\epsilon,\ \|\tilde{k}\|_{C^{k-5,\alpha}_{\rho^{-1}}}\lesssim\epsilon,
\end{align*}
for some $0<\alpha<1$ and $\rho$ is a weight function decays near $\partial H=S_1\cup S_2$ (see Section \ref{The Work of Corvino-Schoen} for  definitions). We remark that the above norms are computed with respect to $\bar{g}_{m_0}$. Since both $(\bar{g}, \bar{k})$ and $(\bar{g}_{m,\mathbf{a}}, \bar{k}_{m,\mathbf{a}})$ solve the constraint equations, it is easy to see that
\begin{align*}
\|\Phi(\tilde{g},\tilde{\pi})\|_{C^{k-6,\alpha}_{\rho^{-1}}}\lesssim\epsilon,
\end{align*}
where $\tilde{\pi}=\tilde{g}-\tr_{\tilde{g}}\tilde{k}\tilde{g}$.

We now follow the procedure in \cite{C} and \cite{C-S} to deform $(\tilde{g},\tilde{\pi})$. Let $\zeta$ be a bump function on $H$ which is compactly supported between $S_1$ and $S_2$, we define $\zeta K$ to be the following set vector space
\begin{equation*}
\zeta K = \{ (\zeta f, \zeta X) | (f, X) \in K\}.
\end{equation*}
It is the obstruction space for the gluing for data close to the Schwarzschild data. We remark that if $\zeta$ is sufficiently close to constant $1$ and if the data $(g,\pi)$ is sufficiently to those of the Schwarzschild data $(g_{m_0},0)$, then the restriction of the operator $D\Phi_{(g,\pi)}^*$ on the orthogonal complement of $\zeta K$ is still injective. In this case, we can still apply the Local Deformation Theorem (Theorem 2 of \cite{C-S} or Section \ref{The Work of Corvino-Schoen}) to deform the image into $\zeta K$ instead of $K$. We also notice that, in this case, if we use $V = L^2(C(H))\times L^2((\mathcal{S}(H))$ to denote the pair of square integrable symmetric two tensors on $H$, thus, $V = K \oplus (\zeta K)^{\perp}$. After applying the Local Deformation Lemma, there will be no component in $(\zeta K)^\perp$ . We also fix a weight function $\rho$ which behaves as $\rho\sim d(\cdot,\partial H)^N$ for sufficiently large $N$ near the boundary.
According to the Local Deformation Theorem , if $\epsilon$ is sufficiently small,
there exists a pair $(h,\omega)\in \mathcal{S}^{k-4,\alpha}(H) \times \mathcal{S}^{k-5,\alpha}(H)$ such that $\Phi(\tilde{g}+h, \tilde{\pi}+\omega)\in\zeta K$, $(h,\omega)=0$ near $\partial H=S_1\cup S_2$. Moreover, $(h,\omega)$ satisfies the following estimates
\begin{align*}
\|(h,\omega)\|_{C^{{k-5},\alpha}_{\rho^{-1}}}\lesssim \epsilon.
\end{align*}
We also remark that, instead of the polynomial weight $\rho$ defined above, we can also use an exponential weight $\rho$, namely $\rho\sim\mathrm{e}^{-\frac{1}{d}}$ near the boundary, to solve $(h,\omega)$ in the smooth class instead of H\"{o}lder classes (because $\Phi(\tilde{g},\tilde{\pi})$ is strictly supported in the interior of $H$). One can refer to \cite{C} and \cite{C-S} for more details.

We are going to argue that we can choose suitable parameters $(m,\mathbf{a})$ so that $\Phi(\tilde{g}+h, \tilde{\pi}+\omega)=0$. According to the above discussion, in order to show this, it suffices to show the projection of $\Phi(\tilde{g}+h, \tilde{\pi}+\omega)$ to $K$ vanishes, i.e. the following projection map $\mathcal{I}$ hits zero:
\begin{equation*}
\mathcal{I}: B_{C_0 \epsilon} \rightarrow \mathbb{R}^4, \quad (m, \mathbf{a}) \mapsto (\mathcal{I}_0, \mathcal{I}_1,\mathcal{I}_2, \mathcal{I}_3),
\end{equation*}
where $B_{C_0 \epsilon} = \{(m,\mathbf{a}) \in \mathbb{R}^4 ||m-m_0| + |\mathbf{a}|\leq C_0\epsilon \}$,
\begin{equation*}
\mathcal{I}_0 (m,\mathbf{a}) = \int_H(1-\frac{2m_0}{r})^{\frac{1}{2}}\mathcal{H}(\tilde{g}+h, \tilde{\pi}+\omega)\mathrm{d}\mu_{\bar{g}_{m_0}},
\end{equation*}
and
\begin{align*}
\mathcal{I}_i(m,\mathbf{a}) = \int_H\mathrm{div}_{\tilde{g}+h}(\tilde{\pi}+\omega)_j\Omega_i^j\D\mu_{\bar{g}_{m_0}}.
\end{align*}
for $i=1,2,3$.

We first analyze $\mathcal{I}_0$. We consider the Taylor expansion of $\mathcal{H}$ in Banach space $C^{k-5,\alpha}$ near the point $(\tilde{g},\tilde{\pi})$, thus, we have
\begin{align*}
\mathcal{I}_{0}(m,\mathbf{a}) = \int_H(1-\frac{2m_0}{r})^{\frac{1}{2}}(\mathcal{H}(\tilde{g}, \tilde{\pi})+D\mathcal{H}_{(\tilde{g},\tilde{\pi})}(h,\omega)+O(\|(h,\omega)\|_{C^{k-5,\alpha}}^2))\mathrm{d}\mu_{\bar{g}_{m_0}}.
\end{align*}
We observe that the second and the last terms in the above integral are of size $O(\epsilon^2)$ with respect to our small parameter $\epsilon$. This is obvious for the last term. For the second term, recall that $(1-\frac{2m_0}{r})^{\frac{1}{2}} \in \mathrm{Ker}D\mathcal{H}^*_{(\bar{g}_{m_0},0)}$ and $(h,\omega)$ vanishes near $\partial H$, therefore,
\begin{equation*}
\int_H(1-\frac{2m_0}{r})^{\frac{1}{2}}D\mathcal{H}_{(\bar{g}_{m_0},0)}(h,\omega)\mathrm{d}\mu_{\bar{g}_{m_0}} =0.
\end{equation*}
We then subtract the above quantity from the second term to retrieve one more $\epsilon$ as follows
\begin{align*}
\int_H(1-\frac{2m_0}{r})^{\frac{1}{2}}D\mathcal{H}_{(\tilde{g},\tilde{\pi})}(h,\omega)\mathrm{d}\mu_{\bar{g}_{m_0}}&=
\int_H(1-\frac{2m_0}{r})^{\frac{1}{2}}(D\mathcal{H}_{(\tilde{g},\tilde{\pi})}-D\mathcal{H}_{(\bar{g}_{m_0},0)})(h,\omega)\mathrm{d}\mu_{\bar{g}_{m_0}}\\
&\lesssim\|(\tilde{g}-\bar{g}_{m_0},\bar{\pi})\|_{C^2}\|(h,\omega)\|_{C^2}= O(\epsilon^2)
\end{align*}
for $k\ge 7$. Finally, we consider the first term in $\mathcal{I}_0$. Because $|\mathbf{a}|\le C_0\epsilon$, we will make use of the following key observation:
\begin{equation*}
\bar{g}_{m,\mathbf{a}}=(1-\frac{2m}{r})^{-1}\D r^2+r^2\D\sigma_{S^2}^2+O(\epsilon^2).
\end{equation*}
This is obvious according to the formula for $\bar{g}_{m,\mathbf{a}}$ in last section. We also write
\begin{align*}
\bar{g}=(1-\frac{2m_0}{r})^{-1}\D r^2+r^2\D\sigma_{S^2}^2+h_{\epsilon}
\end{align*}where $h_{\epsilon}$ is a two tensor and $|h_\epsilon|\lesssim\epsilon$. Therefore, we have
\begin{align*}
\int_H(1-\frac{2m_0}{r})^{\frac{1}{2}}\mathcal{H}(\tilde{g}, \tilde{\pi})\D\mu_{\bar{g}_{m_0}} =& \int_H R(f(r)\D r^2+r^2\D\sigma^2_{S^2})\D\mu_{\bar{g}_0}\\
&+\int_H (1-\frac{2m_0}{r})^{\frac{1}{2}}D\mathcal{H}_{(\bar{g}_{m_0},0)}(\phi h_\epsilon,0)\D\mu_{{\bar{g}}_{m_0}}+O(\epsilon^2)
\end{align*}
where $f(r)=\phi(1-\frac{2m_0}{r})^{-1}+(1-\phi)(1-\frac{2m}{r})^{-1}$. To compute the first term on the right hand side, we have to compute the scalar curvature $R(f(r)\D r^2+r^2\D\sigma^2)$. 
By direct computation,
\begin{equation*}
R(f(r)\D r^2+r^2\D\sigma^2)= 2r^{-2}(rf^{-2}(r)f'(r)-f^{-1}(r)+1)=-2r^{-2}\partial_r(r(f^{-1}(r)-1)).
\end{equation*}
By virtue of this formula, we can easily derive
\begin{align*}
\int_H R(f(r)\D r^2+r^2\D\sigma^2_{S^2})\D\mu_{\bar{g}_0}&=-8\pi(r(f^{-1}(r)-1))|_{r_1}^{r_2}=16\pi(m-m_0).\\
\end{align*}
For the second term, we denote
\begin{align*}
\tilde{\epsilon}_0\triangleq-\int_H (1-\frac{2m_0}{r})^{\frac{1}{2}}D\mathcal{H}_{(\bar{g}_{m_0},0)}(\phi h_\epsilon,0)\D\mu_{{\bar{g}}_{m_0}}.
\end{align*}
Recall that $(1-\frac{2m_0}{r})^{\frac{1}{2}}\in\mathrm{Ker}D\mathcal{H}^*_{(\bar{g}_{m_0},0)}$, $\phi=1$ near $S_1$ and $\phi=0$ near $S_2$, therefore $\tilde{\epsilon}_0$ can be written as an integral on $S_1$, which only depends on $h_\epsilon$ on $S_1$, and $|\tilde{\epsilon}_0|\lesssim\epsilon$. Finally, we summarize the above calculations for $\mathcal{I}_0$ as
\begin{equation}\label{I 0}
\mathcal{I}_{0}(m,\mathbf{a})=16\pi(m-m_0)-\tilde{\epsilon}_0+O(\epsilon^2).
\end{equation}

For each $i=1,2,3$, we analyze $\mathcal{I}_i$ as follows:
\begin{align*}
\mathcal{I}_{i}(m,\mathbf{a})&=\int_H\mathrm{div}_{\bar{g}_{m_0}}(\tilde{\pi}+\omega)_j\Omega_i^j\D\mu_{\bar{g}_{m_0}}+\int_H (\mathrm{div}_{\bar{g}_{m}}-\mathrm{div}_{\bar{g}_{m_0}})(\tilde{\pi}+\omega)_j\Omega_i^j\D\mu_{\bar{g}_{m_0}}\\
&=\int_H\mathrm{div}_{\bar{g}_{m_0}}(\tilde{\pi}+\omega)_j\Omega_i^j\D\mu_{\bar{g}_{m_0}}+O(\epsilon^2)
\end{align*}
Recall that $\mathrm{div}_{\bar{g}_{m_0}}^* (\Omega_i)=0$, $(\tilde{g},\tilde{\pi})$ coincides with $(\bar{g}, \bar{\pi})$ on $S_1$ and $(\tilde{g},\tilde{\pi})$ coincides with $(\bar{g}_{m,\mathbf{a}}, \bar{\pi}_{m,\mathbf{a}})$ on $S_2$. We can apply Stokes Theorem to derive
\begin{align*}
\mathcal{I}_i(m,\mathbf{a})=\int_{S_2}(\bar{k}_{m,\mathbf{a}})_{lj}\Omega_i^j(\partial_r)^l\D\mu_{r_2}-\int_{S_1}\bar{k}_{lj}\Omega_i^j(\partial_r)^l\D\mu_{r_1}+O(\epsilon^2),
\end{align*}
where $\D\mu_r$ is the volume form on the round sphere with radius $r$. Let $\tilde{\epsilon}_i=\int_{S_1}\bar{k}_{lj}\Omega_i^j(\partial_r)^l\D\mu_{r_1}$  and it is a constant coming from the integral on the inner sphere $S_1$. Since $g$ is close to $g_{m_0}$, we know that $|\tilde{\epsilon}_i|\lesssim \epsilon$. We emphasize that $\tilde{\epsilon}_k$ is independent of $(m,\mathbf{a})$ and the constant $C_0$. Thus,
\begin{align*}
\mathcal{I}_i(m,\mathbf{a})=\int_{S_2}(\bar{k}_{m,\mathbf{a}})_{lj}\Omega_i^j(\partial_r)^l\D\mu_{r_2}-\tilde{\epsilon}_i+O(\epsilon^2).
\end{align*}
To understand the local behavior of $\mathcal{I}_i$, we study the differential of the map
\begin{align*}
J:\mathbf{a}\mapsto (\int_{S_2}(\bar{k}_{m,\mathbf{a}})_{lj}\Omega_1^j(\partial_r)^l\D\mu_{r_2},\int_{S_2}(\bar{k}_{m,\mathbf{a}})_{lj}\Omega_2^j(\partial_r)^l\D\mu_{r_2},\int_{S_2}(\bar{k}_{m,\mathbf{a}})_{lj}\Omega_3^j(\partial_r)^l\D\mu_{r_2}),
\end{align*}
at $\mathbf{a}=0$ and the parameter $m$ is fixed.

We first compute
\begin{align*}
\D J \mid_{\mathbf{a} = 0}(0,0,1)=\frac{\D}{\D t} \mid_{t=0}(\int_{S_2}(\bar{k}_{m,(0,0,t)})_{lj}\Omega_i^j(\partial_r)^l\D\mu_{r_2})_{i=1,2,3}.
\end{align*}
When $i=1$, we have
\begin{align*}
&\int_{S_2}(\bar{k}_{m,(0,0,t)})_{lj}\Omega_1^j(\partial_r)^l\D\mu_{r_2}\\
=-&\int_0^{\pi}\D\theta\int_0^{2\pi} \cos\varphi\cot\theta\bar{k}_{m,(0,0,t)}(\partial_\varphi,\partial_r)|_{r=r_2}\cdot r_2^2\sin\theta\D\varphi = 0.
\end{align*}
since $\bar{k}_{m,(0,0,t)}(\partial_\varphi,\partial_r)$ does not depend on $\varphi$ and $\int_0^{2\pi}\cos\varphi\D\varphi=0$. When $i=2$, similarly, we have
\begin{align*}
\int_{S_2}(\bar{k}_{m,(0,0,t)})_{lj}\Omega_2^j(\partial_r)^l\D\mu_{r_2}=0.
\end{align*}
When $i=3$, in view of the exact formula for $k_{m,(0,0,a)}$ in the previous section, the $\D r\D\varphi$ component of $\bar{k}_{m,(0,0,t)}$ is of the form $\frac{2mt}{r^2}+O(t^2)$ for sufficiently small $t$. Thus,
\begin{align*}
\int_{S_2}(\bar{k}_{m,(0,0,t)})_{lj}\Omega_3^j(\partial_r)^l\D\mu_{r_2}
&=\int_0^{\pi}\D\theta\int_0^{2\pi}2mt\sin\theta\D\varphi+O(t^2)\\
&=8\pi mt+O(t^2).
\end{align*}
Therefore, we finally have
\begin{align*}
\frac{\D}{\D t}\mid_{t=0}(\int_{S_2}(\bar{k}_{m,(0,0,t)})_{lj}\Omega_3^j(\partial_r)^l\D\mu_{r_2})=8\pi m,
\end{align*}
or equivalently $\D J \mid_{\mathbf{a} = 0}(0,0,1)=(0,0,8\pi m)$.

To compute $\D J \mid_{\mathbf{a} = 0}(0,1,0)$, we use the $\text{SO}(3)$ symmetry on $\mathbf{a}$. We take $R =\begin{pmatrix}1&0&0\\0&0&1\\0&-1&0\end{pmatrix}$ and use $\D R$ to denote its differential. Apparently, we have
\begin{equation*}
\D R:(\Omega_1,\Omega_2,\Omega_3) \mapsto(\Omega_1,-\Omega_3,\Omega_2).
\end{equation*}
Therefore, we have
\begin{align*}
\D J \mid_{\mathbf{a} = 0}(0,1,0)&=\frac{\D}{\D t}\mid_{t=0}(\int_{S_2}(R^*\bar{k}_{m,(0,0,t)})_{lj}\Omega_i^j(\partial_r)^l\D\mu_{r_2})_{i=1,2,3}\\
&=\frac{\D}{\D t}\mid_{t=0}(\int_{S_2}(\bar{k}_{m,(0,0,t)})_{lj}(\D R(\Omega_i))^j(\partial_r)^l\D\mu_{r_2})\\
&=(0,-8\pi m,0).
\end{align*}
Similarly, we take $R=\begin{pmatrix}0&0&1\\0&1&0\\-1&0&0\end{pmatrix}$ and we gave
\begin{equation*}
\D R:(\Omega_1,\Omega_2,\Omega_3) \mapsto(-\Omega_3,\Omega_2,\Omega_1).
\end{equation*}
As a result, we obtain
\begin{align*}
\D J \mid_{\mathbf{a} = 0}(0,1,0)&=(-8\pi m,0,0).
\end{align*}
Combining all those computations on $\D J \mid_{\mathbf{a} = 0}$ and the fact that $J(0)=0$, for $(m,\mathbf{a}) \in B_{C_0 \epsilon}$, we have
\begin{align*}
J(\mathbf{a})&=8\pi m\begin{pmatrix}-1&0&0\\0&-1&0\\0&0&1\end{pmatrix}\mathbf{a}+O(\epsilon^2)\\
&=8\pi m_0\begin{pmatrix}-1&0&0\\0&-1&0\\0&0&1\end{pmatrix}\mathbf{a}+O(\epsilon^2).
\end{align*}
We now understand the local behavior of $\mathcal{I}$ near $(m_0,0)$:
\begin{equation}\label{local behavior of I}
\mathcal{I}(m,\mathbf{a})=(16\pi(m_0-m),8\pi m_0(-a_1,-a_2,a_3))-(\tilde{\epsilon}_0,\tilde{\epsilon}_i)+O(\epsilon^2),
\end{equation}
where $|\tilde{\epsilon}_i|\lesssim \epsilon$.

We now carry out a degree argument as in \cite{C} and \cite{C-S} to show that $\mathcal{I}(m,\mathbf{a})=0$ for some $(m,\mathbf{a}) \in B_{C_0 \epsilon}$ provided $C_0$ is sufficiently large. We define
\begin{equation*}
\mathcal{I}_1(m,\mathbf{a})=(16\pi(m-m_0),8\pi m_0(-a_1,-a_2,a_3))-(\tilde{\epsilon}_0,\tilde{\epsilon}_k),
\end{equation*}
and we then choose $C_0$ such that $16\pi C_0\ge 2|\tilde{\epsilon}_0|$ and $8\pi m_0 C_0\ge 2 |\tilde{\epsilon}_k|$. Therefore, we have
\begin{equation*}
\mathcal{I}_1(m_0+\frac{\tilde{\epsilon}_0}{16\pi},\frac{1}{8\pi m_0}(-\tilde{\epsilon}_1,-\tilde{\epsilon}_2,\tilde{\epsilon}_3))=0,
\end{equation*}
and $(m_0+\frac{\tilde{\epsilon}_0}{16\pi},\frac{1}{8\pi m_0}(-\tilde{\epsilon}_1,-\tilde{\epsilon}_2,\tilde{\epsilon}_3)) \in B_{C_0 \epsilon}$. Therefore, $\mathcal{I}_1$ is a homomorphism from $B_{C_0\epsilon}$ to another box containing $0$ in $\mathbb{R}^4$ centered at $(\tilde{\epsilon}_0,\tilde{\epsilon}_k)$. We then define
\begin{equation*}
\mathcal{I}(t,m,\mathbf{a})=(16\pi(m-m_0),8\pi m_0(-a_1,-a_2,a_3))-(\tilde{\epsilon}_0,\tilde{\epsilon}_k)+tO(\epsilon^2)
\end{equation*}
on $[0,1]\times B_{C\epsilon_0}$ to be an homotopy between $\mathcal{I}_1$ and $\mathcal{I}$. For sufficiently small $\epsilon$, $0\notin\mathcal{I}([0,1]\times\partial B_{C\epsilon_0})$, so the degree of $\mathcal{I}$ at the value $0$ is equal to the degree of $\mathcal{I}_1$ at the value $0$, which is $1$. This implies that $\mathcal{I}(m,\mathbf{a})=0$ for some $(m,\mathbf{a})\in B_{C_0 \epsilon}$.

Therefore, by a suitable choice of $(m,\mathbf{a}) \in B_{C_0 \epsilon}$, we finally deform $(\tilde{g},\tilde{k})$ to be $({\tilde{g}_S},{\tilde{k}_S})$ where $\tilde{g}_S=\tilde{g}+h$ and $\tilde{k}_S=\tilde{\pi}+\omega+\frac{1}{2}\tr_{\tilde{g}_S}(\tilde{\pi}+\omega){\tilde{g}_S}$, so that $({\tilde{g}_S},{\tilde{k}_S})$ satisfies the constraint equations on $H$, $({\tilde{g}_S},{\tilde{k}_S})=(\bar{g},\bar{k})$ near $S_1$ and $({\tilde{g}_S},{\tilde{k}_S}))=(\bar{g}_{m,\mathbf{a}},\bar{k}_{m,\mathbf{a}})$ near $S_2$. This completes the gluing construction.

\subsection{Proof of the Main Theorem}
Recall the choice of $H$, it is only defined for $\ub \geq \delta$ and it has an inner boundary at $H \cap \Cb_{\delta}$. We would like to extend $H$ to the interior to complete it as a three manifold without boundary. This can be done by choosing a smooth function $u=f(\ub)$, such that
\begin{align*}
\begin{cases}f'(\ub)=-1,\quad\text{for $\ub\le0$},\\
f'(\ub)<0,\quad \text{for $0\le\ub\le\delta$},\\
t(\ub,f(\ub))=0,\quad\text{for $\ub\ge\delta$.}
\end{cases}
\end{align*}
In addition, $f$ can be chosen such that $\tr\chi(\ub,f(\ub))>0$ for all $0\le\ub\le\delta$. The space-like piece in the picture in Section \ref{Main Result} can be viewed as the graph of $f$ in $\ub$-$u$ plane. Suppose that the curve $u=f(\ub)$ intersects the ``central'' line $\ub=u$ at $(\ub_{\text{cen}},f(\ub_{\text{cen}})=\ub_{\text{cen}})$. Let $H_0=\bigcup_{u=f(\ub),\ub_{\text{cen}}\le\ub\le\delta}S_{\ub,u}$ which is smooth space-like hypersurface of $M$ with $(\bar{g},\bar{k})$ as the induced metric and the second fundamental form. It is clear from the construction that $(\bar{g},\bar{k})=(\delta_{ij},0)$ for $\ub_{\text{cen}}\le\ub\le0$ where $\delta$ is the Euclidean metric on the ball of radius $r_0$ which is the radius of $S_{0,f(0)}$ computed in $g$. It is also clear that $0<r_0<r_1$.

We divide $\Sigma=\mathbb{R}^3$ with Cartesian coordinate $(x^1,x^2,x^3)$ into four concentric regions $\Sigma_M$, $\Sigma_C$, $\Sigma_S$ and $\Sigma_K$ as
\begin{align*}
\Sigma_M &=\{x||x|\leq r_0\},\quad \Sigma_C =\{x|r_0\leq |x| \leq r_1\},\\
\Sigma_S &=\{x|r_1\leq |x|\le r_2\},\quad \Sigma_K =\{x||x| \geq r_2\}.
\end{align*}
By construction we have $r_1-r_0=O(\delta)$ and $r_2-r_1=O(\epsilon_0)$. Thus, $H_0$ is diffeomorphic to $\Sigma_M\cup\Sigma_C$ and $H$ is diffeomorphic to $\Sigma_S$. These diffeomorphisms are realized through the Cartesian coordinates in an obvious way. We now define $({\bar{g}_{\Sigma}},{\bar{k}_{\Sigma}})$ on $\Sigma$ as follows
\begin{align*}
({\bar{g}_{\Sigma}},{\bar{k}_{\Sigma}})=\begin{cases}(\bar{g},\bar{k}),\quad\text{on $\Sigma_M\cup\Sigma_C$},\\
({\tilde{g}_S},{\tilde{k}_S}),\quad\text{on $\Sigma_S$},\\
(\bar{g}_{m,\mathbf{a}},\bar{k}_{m,\mathbf{a}}),\quad\text{on $\Sigma_K$}.
\end{cases}
\end{align*}
By construction, $({\bar{g}_{\Sigma}},{\bar{k}_{\Sigma}})$ is smooth on $\Sigma$. It is \textit{Minknowski} inside, namely, $({\bar{g}_{\Sigma}},{\bar{k}_{\Sigma}})=(\delta,0)$ on $\Sigma_M$; similarly, it is \text{Kerr} on the outer region $\Sigma_K$ with $|m-m_0|+|\mathbf{a}|\lesssim\epsilon$. Moreover, according to Christodoulou'w work \cite{Chr}, there will be a trapped surface in the future domain of dependence of $\Sigma_M\cup\Sigma_C$.

To complete the proof of the Main Theorem, we also have to show the non-existence of trapped surfaces on $\Sigma$. In fact, for any two sphere $S$ embedding in $\Sigma$. Let $S_{r_S}=\{x||x|=r_S\}$ be the sphere which is the innermost one of the spheres in the form $S_r=\{x||x|=r\}$ containing $S$. Then $S$ and $S_{r_S}$ tangent at some point $p$. The outer null expansion of $S$, $\theta_S=\tr_{S}{\bar{k}_{\Sigma}}+H_{S}$ where $\tr_{S}{\bar{k}_{\Sigma}}$ is the trace of $\bar{k}_{\Sigma}$ computed on $S$ and $H_S$ is the mean curvature of $S$ in $\Sigma$ with respect to the normal pointing out. By construction $\theta_{S_{r}}>0$ for all $r$. Then at the point $p$, $\tr_{S}\bar{k}_{\Sigma}(p)=\tr_{S_{r_S}}\bar{k}_{\Sigma}(p)$ since $S$ and $S_{r_S}$ tangent at $p$, and $H_{S}(p)\ge H_{S_{r_S}}(p)$ by the maximum principle. So
$$\theta_S(p)=\tr_{S}\bar{k}_{\Sigma}(p)+H_{S}(p)\ge\tr_{S_{r_S}}\bar{k}_{\Sigma}(p)+H_{S_{r_S}}(p)=\theta_{S_{r_S}}(p)>0$$
and then we conclude that $S$ is not trapped at the point $p$. Therefore, we complete the proof of the Main Theorem.

\end{document}